%% file: main.tex
\documentclass[conference]{IEEEtran}

\usepackage{amsthm,mathtools,amssymb,mathrsfs}
\usepackage{color}
\usepackage{tikz}
\usepackage[format=hang]{subfig}
\usepgflibrary{shapes}
\usetikzlibrary{petri,arrows,positioning}

\tikzstyle{ran}=[shape=circle,draw,inner sep=0pt,minimum size=4mm,thick]
\tikzstyle{tran}=[thick,draw,->,>=stealth]

\newcommand{\calF}{\mathcal{F}}

\newcommand{\calP}{\mathcal{P}}

\newcommand{\A}{\mathcal{A}}
\newcommand{\B}{\mathcal{B}}

\newcommand{\C}{\mathscr{C}}
\newcommand{\D}{\mathcal{D}}
\newcommand{\E}{\mathbb{E}}
\newcommand{\calE}{\mathcal{E}}

\newcommand{\G}{\mathcal{G}}

\newcommand{\V}{\mathcal{V}}

\newcommand{\M}{\mathcal{M}}
\newcommand{\N}{\mathcal{N}}
\newcommand{\X}{\mathcal{X}}

\newcommand{\Prob}{\mathit{Prob}}
\newcommand{\Pat}{\mathit{Pat}}
\newcommand{\Nset}{\mathbb{N}}

\newcommand{\Zset}{\mathbb{Z}}
\newcommand{\Qset}{\mathbb{Q}}
\newcommand{\Rset}{\mathbb{R}}
\newcommand{\fpath}{\mathit{FPath}}
\newcommand{\run}{\mathit{Run}}

\newcommand{\len}{\mathit{length}}

\newcommand{\change}{\mathit{change}}

\newcommand{\post}{\mathit{post}^*}
\newcommand{\pre}{\mathit{pre}^*}

\newcommand{\xmin}{x_{\mathrm{min}}}
\newcommand{\size}[1]{|\!|#1|\!|}
\newcommand{\norm}[1]{|\!|\!|#1|\!|\!|}
\renewcommand{\vec}[1]{\pmb{#1}}

\newcommand{\mi}{i}

\newcommand{\conf}{\mathit{conf}}
\newcommand{\sem}[1]{[\![#1]\!]}
\newcommand{\Reg}{\mathit{Reg}}

\newcommand{\ms}[1]{m^{(#1)}}
\newcommand{\ps}[1]{p^{(#1)}}
\newcommand{\xs}[1]{x^{(#1)}}

\newcommand{\gmax}{g_{\mathit max}}

\newcommand{\rmax}{r_{\mathit max}}
\newcommand{\ymin}{y_{\mathit min}}

\newcommand{\tran}[1]{{}\mathchoice%
    {\stackrel{#1}{\rightarrow}}
    {\mathop {\smash\rightarrow}\limits^{\vrule width 0pt height 0pt
                                                depth 4pt\smash{#1}}}
    {\stackrel{#1}{\rightarrow}}
    {\stackrel{#1}{\rightarrow}}
{}}
\newcommand{\ltran}[1]{{}\mathchoice%
    {\stackrel{#1}{\longrightarrow}}
    {\mathop {\smash\longrightarrow}\limits^{\vrule width 0pt height 0pt
                                                depth 4pt\smash{#1}}}
    {\stackrel{#1}{\longrightarrow}}
    {\stackrel{#1}{\longrightarrow}}
{}}

\theoremstyle{plain}
\newtheorem{assumption}{Assumption}

\newtheorem{lemma}{Lemma}

\newtheorem{proposition}{Proposition}
\newtheorem{definition}{Definition}
\newtheorem{corollary}{Corollary}
\newtheorem{theorem}{Theorem}

\newtheorem{remark}{Remark}

\newenvironment{reftheorem}[2]{\begin{trivlist}
\item[\hskip \labelsep {\bfseries #1}\hskip \labelsep {\bfseries #2}]\itshape}{\end{trivlist}}

\begin{document}

\title{Long-Run Average Behaviour of Probabilistic Vector Addition Systems}

\author{\IEEEauthorblockN{Tom\'{a}\v{s} Br\'{a}zdil}
\IEEEauthorblockA{\small Faculty of Informatics\\
Masaryk University\\
Brno, Czech Republic\\
brazdil@fi.muni.cz}
\and
\IEEEauthorblockN{Stefan Kiefer}
\IEEEauthorblockA{\small Department of Computer Science\\
University of Oxford\\
United Kingdom\\
stefan.kiefer@cs.ox.ac.uk}\\
\and
\IEEEauthorblockN{Anton\'{\i}n Ku\v{c}era}
\IEEEauthorblockA{\small Faculty of Informatics\\
Masaryk University\\
Brno, Czech Republic\\
kucera@fi.muni.cz}
\and
\IEEEauthorblockN{Petr Novotn\'{y}}
\IEEEauthorblockA{\small IST Austria\\
	Klosterneuburg, Austria\\
	petr.novotny@ist.ac.at}
}

\maketitle

\begin{abstract}
We study the pattern frequency vector for runs in probabilistic Vector
Addition Systems with States (pVASS). Intuitively, each configuration
of a given pVASS
is assigned one of finitely many \emph{patterns}, and every run can
thus be seen as an infinite sequence of these patterns. The pattern frequency
vector assigns to each run the limit of pattern frequencies computed
for longer and longer prefixes of the run. If the limit does not exist,
then the vector is undefined. We show that for one-counter pVASS,
the pattern frequency vector is defined and takes one of finitely many
values for almost all runs. Further, these values and their associated
probabilities can be approximated up to an arbitrarily small relative error
in polynomial time. For stable two-counter pVASS, we show the same result,
but we do not provide any upper complexity bound.
As a byproduct of our study, we discover counterexamples falsifying
some classical results about stochastic Petri nets published in the~80s.
\end{abstract}

\input{intro.tex}

\input{prelim.tex}
\input{onecounter-new.tex}
\input{twocounter-new.tex}

\input{app-three-counters}

\section*{Acknowledgement}
\noindent
Tom\'{a}\v{s} Br\'{a}zdil and Anton\'{\i}n Ku\v{c}era are supported
by the Czech Science Foundation, Grant \mbox{No.~15-17564S}. The research leading to these results has received funding from the People Programme (Marie
Curie Actions) of the European Union's Seventh Framework Programme (FP7/2007-2013) under
REA grant agreement no [291734].

\bibliographystyle{abbrv}
\bibliography{petri_net,str-short,concur,stefan}

\onecolumn
\appendices
\input{app-one-counter}
\input{app-thms.tex}
\input{app-attr.tex}
\input{stefan.tex}

\end{document}

%% file: intro.tex
\section{Introduction}
\label{sec-intro}

Stochastic extensions of Petri nets are intensively used in
performance and dependability analysis as well as reliability
engineering and bio-informatics.  They have been developed in the
early
eighties~\cite{DBLP:journals/tc/Molloy82,DBLP:journals/tocs/MarsanCB84},
and their token-game semantics yields a denumerable Markov chain.  The
analysis of stochastic Petri nets (SPNs) has primarily focused on
long-run average behaviour.  Whereas for safe nets long-run averages
always exist and can be efficiently computed, the setting of
infinite-state nets is much more challenging.  This is a practically
very relevant problem as, e.g., classical open queueing networks and
biological processes typically yield nets with unbounded state space.
The aim of this paper is to study the long-run average behaviour for
infinite-state nets.  We do so by considering probabilistic
Vector Addition Systems with States (pVASS, for short), finite-state weighted
automata equipped with a finite number of non-negative counters.  A pVASS
evolves by taking weighted rules along which any counter can be
either incremented or decremented by one (or zero). The probability
of performing a given enabled rule is given by its weight divided
by the total weight of all enabled rules. This model is
equivalent to discrete-time SPNs: a counter vector
corresponds to the occupancy of the unbounded places in the net,
and the bounded places are either encoded in the counters
or in the control states. Producing a
token yields an increment, whereas token consumption yields a
decrement. Discrete-time SPNs describe the probabilistic branching of
the continuous-time Markov chains determined by SPNs, and many
properties of continuous-time SPNs can be derived directly from the properties of
their underlying discrete-time SPNs. In fact, discrete-time SPNs
are a model of interest in itself,
see e.g.,~\cite{DBLP:journals/fuin/Kudlek05a}.


Our study concentrates on long-run average \emph{pattern frequencies} for
pVASS. A \emph{configuration} of a given pVASS $\A$ is a pair $p\vec{v}$, where $p$ is the current control state and
$\vec{v} \in \Nset^d$ is the vector of current counter values.
The \emph{pattern} associated to $p\vec{v}$ is a pair $p\alpha$, where
$\alpha \in \{0,*\}^d$, and $\alpha_i$
is either $0$ or $*$, depending on whether $\vec{v}_i$ is zero or positive
(for example, the pattern associated to $p(12,0)$ is
$p(*,0)$). Every run in $\A$ is an infinite sequence of configurations
which determines a unique infinite sequence of the associated patterns.
For every finite prefix of a run $w$, we can compute the frequency
of each pattern in the prefix, and define the \emph{pattern frequency
vector} for $w$, denoted by $F_\A(w)$, as the limit of the sequence
of frequencies computed for longer and longer prefixes of~$w$. If the
limit does not exist, we put $F_\A(w) = {\perp}$ and say that $F_\A$
is not well defined for~$w$. Intuitively, a pattern represents the
information sufficient to determine the set of enabled rules (recall
that each rule can consume at most one token from each counter). Hence,
if we know $F_\A(w)$, we can also determine the limit frequency of
rules fired along~$w$. However, we can also encode various predicates
in the finite control of $\A$ and determine the frequency of (or time
proportion spent in) configurations in $w$ satisfying the predicate.
For example, we might wonder what is the proportion of time spent
in configurations where the second counter is even, which can be
encoded in the above indicated way.

The very basic questions about the pattern frequency vector include
the following:
\begin{itemize}
\item Do we have $\calP(F_\A {=} \perp) = 0$, i.e., is $F_\A$
  well defined for almost all runs?
\item Is $F_\A$ (seen as a random variable) discrete? If so, how many
  values can $F_\A$ take with positive probability?
\item Can we somehow compute or approximate possible
  values of $F_\A$ and the probabilities of all runs that
  take these values?
\end{itemize}
These fundamental questions are rather difficult for general pVASS. In this paper,
we concentrate on the subcase of pVASS with one or two counters, and we also observe
that with three or more counters, there are some new unexpected phenomena that make
the analysis even more challenging. Still, our results can be seen as a basis for designing algorithms that analyze the long-run average behaviour in certain subclasses of pVASS with  arbitrarily many counters (see below). The main ``algorithmic results'' of our paper can be summarized as follows:
\smallskip

\textbf{1.} For a \emph{one-counter} pVASS with $n$ control states, we show
  that $F_\A$ is well defined and takes at most $\max\{2,2n{-}1\}$ different values for almost
  all runs.
  These values and the associated probabilities may be irrational, but can be
  effectively approximated up to an arbitrarily small relative error
  $\varepsilon > 0$ in \emph{polynomial time}.

\textbf{2.} For \emph{two-counter} pVASS that are \emph{stable}, we show
  that $F_\A$ is well defined and takes only finitely many values for almost
  all runs. Further, these values and the associated probabilities can be
  effectively approximated up to an arbitrarily small absolute/relative error
  $\varepsilon > 0$.
\smallskip

Intuitively, a two-counter pVASS $\A$ is \emph{unstable} if the changes
of the counters are well-balanced so that certain infinite-state Markov chains
used to analyze the behaviour of $\A$ may become \emph{null-recurrent}.
Except for some degenerated cases, this null-recurrence
is not preserved under small perturbations in transition probabilities. Hence,
we can assume that pVASS models constructed by estimating some real-life probabilities are stable. Further, the analysis of null-recurrent Markov chains
requires different methods and represents an almost independent task.
Therefore, we decided to disregard unstable two-counter pVASS in this
paper. Let us note that the problem whether a given two-counter pVASS $\A$ is (un)stable
is decidable in exponential time.

The above results for one-counter and stable two-counter pVASS are obtained by showing
the following:
\smallskip

\textbf{(a)} There are finitely many sets of configurations called \emph{regions}, such that almost every run eventually stays in some region, and almost all runs that stay in the same region share the same well-defined value of the pattern frequency vector.

\textbf{(b)} For every region $R$, the associated pattern frequency vector and the
probability of reaching $R$ can be computed/approximated effectively. For one-counter pVASS,
we first identify families of regions (called \emph{zones}) that share the same pattern frequency 
vector, and then consider these zones rather then individual regions. 
\smallskip

For one-counter pVASS, we show that the total number of all regions (and hence also zones) cannot exceed $\max\{2,2n{-}1\}$, where $n$ is the number of control states. To compute/approximate the pattern frequency vector of a given zone $Z$ and the probability of staying in~$Z$, the tail bounds of \cite{BKK:pOC-time-LTL-martingale-JACM}
and the polynomial-time algorithm of \cite{SEY:pOC-poly-Turing} provide all the tools we need.

For two-counter pVASS, we do not give an explicit bound on the number of regions, but we show that all regions are effectively semilinear (i.e., for each region there is a computable Presburger formula which represents the region).
Here we repeatedly use the result of \cite{LS:two-counter-VASS-flat} which says that the reachability  relation of a two-counter VASS is effectively semilinear. Technically, we show that every run eventually reaches a configuration where one or both counters become bounded
or irrelevant (and we apply the results for one-counter pVASS), or a configuration
of a special set $C$ for which we show the existence and effective constructibility of
a finite eager attractor\footnote{A finite eager attractor \cite{AHMS:Eager-limit} for a set of configurations $C$ is a finite set of configurations $A \subseteq C$ such that the probability of reaching $A$ from every configuration of $C \cup \post(A)$ is equal to~$1$, and the probability of revisiting $A$ in more than $\ell$ steps after leaving $A$ decays (sub)exponentially in $\ell$.}.
This is perhaps the most advanced part of our paper, where we need to establish new
exponential tail bounds for certain random variables using an appropriately defined martingale. We believe that these tail bounds and the associated martingale are of broader interest,
because they provide generic and powerful tools for quantitative analysis of two-counter pVASS.
Hence, every run which visits $C$ also visits its finite eager attractor, and the regions
where the runs initiated in $C$ eventually stay
correspond to bottom strongly connected components of this attractor. For each of these
bottom strongly connected components, we approximate the pattern frequency vector
by employing the abstract algorithm of \cite{AHMS:Eager-limit}.

The overall complexity of our algorithm for stable two-counter pVASS could be estimated
by developing lower/upper bounds on the parameters that are used in the lemmata of
Section~\ref{sec-two-counters}. Many of these parameters are ``structural'' (e.g.,
we consider the minimal length of a path from some configuration to some set of
configurations). Here we miss a refinement of the results published in
\cite{LS:two-counter-VASS-flat} which would provide explicit upper bounds.
Another difficulty is that we do not have any lower bound on $|\tau_R|$ in the case when
$\tau_R \neq 0$, where $\tau_R$ is the mean payoff defined in
Section~\ref{sec-two-counters}. Still, we conjecture that these ``structural bounds''
and hence also the complexity of our algorithm are not too high (perhaps,
singly exponential in the size of $\A$ and in $|\tau_R|$), but we leave
this problem for future work.

The results summarized in (a) and (b) give a reasonably deep understanding
of the long-run behaviour of  a given one-counter or a stable two-counter pVASS, which can be used to develop algorithms for other interesting problems. For example, we can decide the existence of a finite attractor for the set of configuration reachable from a given initial configuration, we can provide a sufficient condition which guarantees that all pattern frequency vectors taken with positive probability are rational, etc.
An obvious question is whether these results can be extended to pVASS with
three or more counters. The answer is twofold.
\smallskip

\textbf{I.} The algorithm for stable two-dimensional pVASS presented in Section~\ref{sec-two-counters} in fact ``reduces'' the analysis of a given two-counter pVASS $\A$ to the analysis of several one-counter pVASS and the analysis of some ``special''
configurations of~$\A$. It seems that this approach can be generalized to a recursive
procedure which takes a pVASS $\A$ with $n$ counters, isolates certain subsets of runs whose properties can be deduced by analyzing pVASS with smaller number of counters, and checks that the remaining runs are sufficiently simple so that they can be analyzed directly.
Thus, we would obtain a procedure for analyzing a subset of pVASS with $n$~counters.

\textbf{II.} In Section~\ref{sec-three-counters} we give an example of a \emph{three-counter} pVASS
$\A$ with strongly connected state-space whose long-run behaviour is \emph{undefined} for almost all runs (i.e., $F_\A$ takes the $\perp$ value), and this property is not sensitive to small perturbations in transition probabilities. Since we do not provide a rigorous mathematical analysis of $\A$ in this paper, the above claims are formally just \emph{conjectures} confirmed only by Monte Carlo simulations. Assuming that these conjectures are valid, the method used for two-counter pVASS is not sufficient for the analysis of general three-counter pVASS, i.e., there are new phenomena which cannot be identified by the methods used for two-counter pVASS.
\smallskip



\noindent
\textbf{Related work.}
The problem of studying pattern frequency vector
is directly related to the study of ergodicity properties in
stochastic Petri nets, particularly to the study
of the so-called \emph{firing process}. A classical paper in this
area \cite{DBLP:journals/tse/FlorinN89} has been written by
Florin \& Natkin in the~80s. In the
paper, it is claimed that if the state-space of a given stochastic
Petri net (with arbitrarily many unbounded places) is
strongly connected, then the firing process in ergodic.
In the setting of (discrete-time) probabilistic Petri nets, this
implies that for almost all runs, the limit frequency of transitions
performed along a run is defined and takes the same value.
A simple counterexample to this claim is shown in Fig.~\ref{fig-SPN}.
The net $\N$ has two unbounded places and strongly connected state-space, but the limit frequency of transitions takes two values with positive probability (each with probability $1/2$).
Note that $\N$ can be translated into an equivalent pVASS $\A$ with two counters
which is also shown in Fig.~\ref{fig-SPN}.  Intuitively, if both places/counters are positive, then both of them have a tendency to decrease, i.e., the trend $t_S$
of the only BSCC $S$ of $\C_\A$ is negative in both components (see Section~\ref{sec-prelim}).
However, if we reach a configuration where the first place/counter is zero and the second place/counter is sufficiently large, then the second place/counter starts to \emph{increase},
i.e., it never becomes zero again with some positive probability
(i.e., the the mean payoff $\tau_{R_2}$ is positive, where $R_2$ is the only type~II
region of the one-counter pVASS $\A_2$, see Section~\ref{sec-two-counters}).
The first place/counter stays zero for most of the time, because when it becomes positive, it is immediately emptied with a very large probability. This means
that the frequency of firing $t_2$ will be much higher than the frequency
of firing $t_1$. When we reach a configuration where the first place/counter
is large and the second place/counter is zero, the situation is symmetric, i.e.,
the frequency of firing $t_1$ becomes much higher than the frequency
of firing $t_2$. Further, almost every run eventually behaves according to
one of the two scenarios, and therefore there are two  limit
frequencies of transitions, each of which is taken with probability~$1/2$.
This possibility of reversing the ``global'' trend of the counters after hitting
zero in some counter was not considered in \cite{DBLP:journals/tse/FlorinN89}.
Further, as we already mentioned, we conjecture the existence of a three-counter pVASS $\A$ with strongly connected state-space (the one of Section~\ref{sec-three-counters}) where the limit frequency
of transitions is undefined for almost all runs. So, we must unfortunately conclude
that the results of \cite{DBLP:journals/tse/FlorinN89} are invalid for fundamental
reasons. On the other
hand, the results achieved for one-counter pVASS are consistent
with another paper by
Florin \& Natkin \cite{DBLP:journals/jss/FlorinN86} devoted to stochastic Petri
nets with only one unbounded place and strongly connected state-space,
where the firing process is indeed ergodic (in our terms, the
pattern frequency vector takes only one value with probability~$1$).

\begin{figure}[t]
	\begin{center}
		\begin{tikzpicture}[x=1.1cm,y=3cm,>=stealth',font=\scriptsize,
		every transition/.style={draw,minimum size=5mm},
		every place/.style={draw,minimum size=5mm},
		bend angle=45]
		\node[place] (1) at (0,0) {};
		\node[place] (2) at (1.8,0) {};
		\node[transition] (middle) at (1,0) {$100$}
		edge[pre] (1)
		edge[pre] (2);
		\node[transition] at (0,0.4) {$1$}
		edge[post] (1);
		\node[transition] at (1.8,0.4) {$1$}
		edge[post] (2);
		\node[transition] (1d) at (0,-0.4) {$10$}
		edge[pre] (1);
		\draw (1d.west) edge[post,bend left] (1);
		\draw (1d.east) edge[post,bend right] (1);
		\node[transition] (2d) at (1.8,-0.4) {$10$}
		edge[pre] (2);
		\draw (2d.west) edge[post,bend left] (2);
		\draw (2d.east) edge[post,bend right] (2);
		\node at (0,-.55) {$t_1$};
		\node at (1.8,-.55) {$t_2$};
        \node (s)  [ran] at (5.3,0) {$s$};
        \node (t1) [ran] at (4,.5) {};
        \node (t2) [ran] at (5.3,.5) {};
        \node (u1) [ran] at (4,-.5) {};
        \node (u2) [ran] at (5.3,-.5) {};
        \draw [tran] (t1) -- node[above] {$(1,0),1$} (t2);
        \draw [tran] (u1) -- node[below] {$(0,1),1$} (u2);
        \draw [tran] (t2)  -- node[left]  {$(1,0),1$} (s);
        \draw [tran] (u2)  -- node[left]  {$(0,1),1$} (s);
        \draw [tran, rounded corners, <-] (t1)  -- node[left]  {$(-1,0),10$} +(0,-.4) -- (s);
        \draw [tran, rounded corners,<-] (u1)  -- node[left] {$(0,-1),10$} +(0,.4) -- (s);
        \draw [tran, rounded corners] (s.30) -- +(.6,.2) -- node[right] {$(1,0),1$} +(.3,.3) -- (s);
        \draw [tran, rounded corners] (s.330) -- +(.6,-.2) -- node[right] {$(0,1),1$} +(.3,-.3) -- (s);
        \draw [tran, rounded corners] (s.350) -- +(.4,-.1) -- node[right] {$(-1,-1),100$} +(.4,.1) -- (s);
		\end{tikzpicture}
	\end{center}
	\caption{A discrete-time SPN $\N$ and an equivalent pVASS $\A$.}
	\label{fig-SPN}
\end{figure}
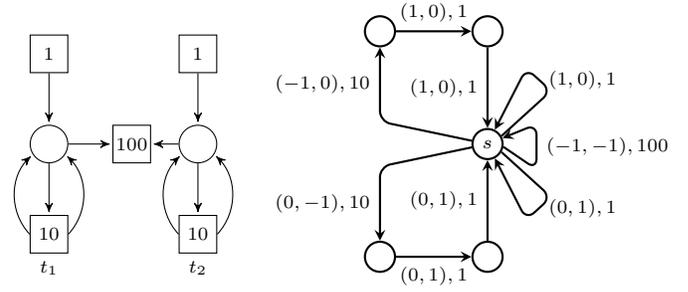

%% file: prelim.tex
\section{Preliminaries}
\label{sec-prelim}

\noindent
We use $\Zset$, $\Nset$, $\Nset^+$, $\Qset$, and $\Rset$
to denote the set of all integers, non-negative
integers, positive integers, rational numbers, and real numbers,
respectively. The absolute value of a given $x \in \Rset$ is denoted
by $|x|$. Let $\delta > 0$, $x \in \Qset$, and $y \in \Rset$.
We say that $x$ approximates $y$ up to
a relative error $\delta$, if either \mbox{$y \neq 0$} and
\mbox{$|x-y|/|y| \leq \delta$}, or $x = y = 0$. Further, we
say that $x$ approximates $y$ up to
an absolute error $\delta$ if \mbox{$|x-y| \leq \delta$}. 
We assume that rational numbers (including integers) are represented
as fractions of binary numbers, and we use $\size{x}$ to denote
the size (length) of this representation.



Let $\V = (V,\tran{})$, where $V$ is a non-empty set of vertices and
${\tran{}} \subseteq V \times V$ a \emph{total} relation
(i.e., for every $v \in V$ there is some $u \in
V$ such that $v \tran{} u$). The reflexive and transitive closure of
$\tran{}$ is denoted by $\tran{}^*$, and the reflexive, symmetric and transitive closure
of $\tran{}$ is denoted by $\leftrightarrow^*$. We say that $\V$ is \emph{weakly connected}
if $s \leftrightarrow^* t$ for all $s,t \in V$.
A \emph{finite path} in $\V$ of
\emph{length} $k \geq 0$ is a finite sequence
of vertices $v_0,\ldots,v_k$, where
$v_i \tran{} v_{i+1}$ for all $0 \leq i <k$. The length of a finite
path $w$ is denoted by $\len(w)$. A \emph{run} in $\V$ is an infinite
sequence $w$ of vertices such that every finite prefix of $w$ is
a finite path in $\V$. The individual vertices of $w$ are denoted by
$w(0),w(1),\ldots$
The sets of all finite paths and all runs in $\V$
that start with a given finite path $w$ are denoted by
$\fpath_{\V}(w)$ and $\run_{\V}(w)$ (or just by $\fpath(w)$ and $\run(w)$
if $\V$ is understood), respectively. For a given set $S \subseteq V$,
we use $\pre(S)$ and $\post(S)$ to denote the set of all $v \in V$ such that
$v \tran{}^* s$ and $s \tran{}^* v$ for some $s \in S$, respectively.
Further, we say that a run $w$ \emph{stays} in~$S$ if there is a $k \in \Nset$
such that for all $\ell \geq k$ we have that $w(\ell) \in S$. The set of all
runs initiated in $s$ that stay in $S$ is denoted by $\run(s,S)$.

A \emph{strongly connected component (SCC)} of $\V$ is a
maximal subset $C \subseteq V$ such that for all $v,u \in C$
we have that $v \tran{}^* u$. A SCC~$C$ of $\V$ is
a \emph{bottom SCC (BSCC)} of $\V$ if for all
$v \in C$ and $u \in V$ such that $v \tran{} u$ we have that
$u \in C$.

We assume familiarity with basic notions of probability theory, e.g.,
\emph{probability space}, \emph{random variable}, or the \emph{expected
value}. Given events $E,F$, we say that \emph{$E$ holds for almost all elements 
of $F$} if $\calP(E \cap F) = \calP(F)$ (in particular, if $\calP(F) = 0$, then any
event holds for almost all elements of $F$).
As usual, a \emph{probability distribution} over a finite or
countably infinite set $A$ is a function
$f : A \rightarrow [0,1]$ such that \mbox{$\sum_{a \in A} f(a) = 1$}.
We call $f$ \emph{positive} if
$f(a) > 0$ for every $a \in A$, and \emph{rational} if $f(a) \in
\Qset$ for every $a \in A$.


\begin{definition}
\label{def-Markov-chain}
  A \emph{Markov chain} is a triple \mbox{$\M = (S,\tran{},\Prob)$}
  where $S$ is a finite or countably infinite set of \emph{vertices},
  \mbox{${\tran{}} \subseteq S \times S$} is a total
  \emph{transition relation},
  and $\Prob$ is a function that assigns to each state $s \in S$
  a positive probability distribution over the outgoing transitions
  of~$s$. As usual, we write $s \tran{x} t$ when $s \tran{} t$
  and $x$ is the probability of $s \tran{} t$.
\end{definition}

To every $s \in S$ we associate the standard probability
space $(\run_{\M}(s),\calF,\calP)$ of runs starting at $s$,
where
$\calF$ is the \mbox{$\sigma$-field} generated by all \emph{basic cylinders}
$\run_{\M}(w)$, where $w$ is a finite path starting at~$s$, and
$\calP: \calF \rightarrow [0,1]$ is the unique probability measure such that
$\calP(\run_{\M}(w)) = \prod_{i{=}1}^{\len(w)} x_i$ where
$w(i{-}1) \tran{x_i} w(i)$ for every $1 \leq i \leq \len(w)$.
If $\len(w) = 0$, we put $\calP(\run_{\M}(w)) = 1$.

If \mbox{$\M = (S,\tran{},\Prob)$} is a strongly connected finite-state Markov chain,
we use $\mu_S$ to denote the unique \emph{invariant distribution of $\M$}.
Recall that by the strong ergodic theorem, (see, e.g., \cite{Norris:book}), 
the limit frequency of visits to the states of $S$ is defined for
almost all $w \in \run(s)$ (where $s \in S$ is some initial state) and it is equal
to~$\mu_S$.

\begin{definition}
	\label{def-pMCA}
	A \emph{probabilistic vector addition system with states (pVASS)} of dimension
	$d \geq 1$ is a triple $\A = (Q,\gamma,W)$, where $Q$ is a finite
	set of \emph{control states}, \mbox{$\gamma \subseteq Q \times \{-1,0,1\}^d \times Q$}
	is a set of \emph{rules}, and $W : \gamma \rightarrow \Nset^+$ is a
	\emph{weight assignment}.
\end{definition}



In the following, we often write $p \tran{\kappa} q$ to denote that
\mbox{$(p,\kappa,q) \in \gamma$}, and $p \ltran{\kappa,\ell} q$
to denote that $(p,\kappa,q) \in \gamma$ and $W((p,\kappa,q)) = \ell$.
The encoding size of $\A$ is denoted by $\size{\A}$, where the weights
are encoded in binary.

\begin{assumption}
\label{asm-simpleVASS}
	From now on (in the whole paper), we assume that $(Q,\tran{})$, where
	$p \tran{} q$ iff $p \tran{\kappa} q$ for some $q$, is weakly connected.
	Further, we also assume that for every pair of control states $p,q$
	there is at most one rule of the form $p \tran{\kappa} q$.
\end{assumption}

\noindent
The first condition of Assumption~\ref{asm-simpleVASS} is obviously safe (if $(Q,\tran{})$
is not weakly connected, then $\A$ is a ``disjoint union'' of several independent pVASS,
and we can apply our results to each of them separately). The second condition is also safe because every pVASS $\A$ can be easily transformed into another pVASS $\A'$ satisfying this condition in the following way: for each control state $s$ of $\A$ and each rule of the form $r \tran{\kappa} s$ we add a fresh control state $s[r,\kappa]$ to $\A'$. Further, for every $s \ltran{\kappa,\ell} t$ in $\A$ we add $s[r,\kappa'] \ltran{\kappa,\ell} t[s,\kappa]$ to $\A'$ (for all states of the form $s[r,\kappa']$ in $\A'$). In other words, $\A'$ is the same as $\A$, but it also ``remembers'' the rule that was used to enter a given control state.

A \emph{configuration} of $\A$ is an element of \mbox{$\conf(\A) = Q \times \Nset^d$},
written as $p\vec{v}$.
A rule $p \tran{\kappa} q$ is \emph{enabled} in a
configuration $p\vec{v}$ if $\vec{v}_i > 0$ for all $1 \leq i \leq d$
with $\kappa_i = -1$. To $\A$ we associate an infinite-state
Markov chain $\M_\A$ whose vertices are the configurations
of $\A$, and the outgoing transitions of a configuration
$p\vec{v}$ are determined as follows:
\begin{itemize}
\item If no rule of $\gamma$ is enabled in $p\vec{v}$, then
  $p\vec{v} \tran{1} p\vec{v}$ is the only outgoing transition of $p\vec{v}$;
\item otherwise, for every rule $p \ltran{\kappa,\ell} q$ enabled in
  $p\vec{v}$ there is a transition  $p\vec{v} \ltran{\ell/T} q(\vec{v}+\kappa)$
  where $T$ is the total weight of all rules enabled in $p\vec{v}$, and
  there are no other outgoing transitions of $p\vec{v}$.
\end{itemize}

In this paper, we also consider the underlying finite-state Markov chain of
$\A$, denoted by $\C_\A$, whose vertices are the control states of $\A$, and
$p \tran{x} q$ in $\C_\A$ iff $p \ltran{\kappa,\ell} q$ in $\A$ and
$x = \ell/T_p > 0$, where $T_p$ is the sum of the weights of all outgoing rules
of~$p$ in~$\A$. Note that every BSCC $S$ of $\C_\A$ can be seen as a strongly
connected finite-state Markov chain, and we use $\mu_S$ to denote the invariant
distribution on the states of~$S$. To each $s \in S$ we associate the vector
\[\change(s) = \sum_{(s,\kappa,t) \in \gamma} \kappa \cdot \frac{W((s,\kappa,t))}{T_s}\]
of expected changes in counter values at~$s$. Further,
we define the \emph{trend} of~$S$, denoted by $t_S$, as the vector
\mbox{$t_S = \sum_{s \in S} \mu_S(s) \cdot \change(s)$}.


A \emph{pattern} of $\A$ is a pair $q\alpha \in Q \times \{0,*\}^d$, and the
set of all patterns of $\A$ is denoted by $\Pat_\A$.
A configuration $p\vec{v}$ \emph{matches} a pattern $q\alpha \in \Pat_\A$ if
$p = q$ and for every $i \in \{1,\ldots,d\}$ we have that $\vec{v}_i = 0$ or
$\vec{v}_i > 0$, depending on whether $\alpha_i =0$ or $\alpha_i = {*}$,
respectively. Intuitively, a pattern represents exactly the information which
determines the set of enabled rules. For all $w \in \run_{\M_\A}(p\vec{v})$,
we define the \emph{pattern frequency vector}
\mbox{$F_\A(w) : \Pat_\A \rightarrow \Rset$} as follows:
\[
  F_\A(w)(q\alpha) = \lim_{k\rightarrow\infty}
     \frac{\#_{q\alpha}(w(0),\ldots,w(k))}{k+1}
\]
where $\#_{q\alpha}(w(0),\ldots,w(k))$ denotes the total number of all indexes
$i$ such that $0 \leq i \leq k$ and $w(i)$ matches the pattern $q\alpha$.
If the above limit does not exist for some $q\alpha \in \Pat_\A$,
we put $F_\A(w) = {\perp}$.  We say
that $F_\A$ is \emph{well defined} for $w$ if $F_\A(w) \neq {\perp}$.
Note that if $F_\A$ is well defined for~$w$, then the sum of all components of
$F_\A(w)$ is equal to~$1$.

Let $R \subseteq \run(p\vec{v})$ be a measurable subset of runs, and let $\varepsilon >0$.
We say that a sequence $(H_1,P_1),\ldots,(H_n,P_n)$, where
$H_i : \Pat_\A \rightarrow \Qset$ and $P_i \in \Qset$, approximates
the pattern frequencies of $R$ up to the absolute/relative error~$\varepsilon$,
if there are pairwise disjoint measurable subsets $R_1,\ldots,R_n$ of $R$
and vectors $F_1,\ldots,F_n$, where $F_i : \Pat_\A \rightarrow \Rset$,
such that
\begin{itemize}
	\item $\sum_{i=1}^n \calP(R_i) = \calP(R)$;
	\item $F_\A(w) = F_i$ for almost all $w \in R_i$;
	\item $H_i(q\alpha)$ approximates $F_i(q\alpha)$ up to the absolute/relative error
	$\varepsilon$ for every $q\alpha \in \Pat_\A$;
	\item $P_i$ approximates $\calP(R_i)$ up to the absolute/relative error~$\varepsilon$.
\end{itemize}
Note that if $(H_1,P_1),\ldots,(H_n,P_n)$ approximates the pattern frequencies of $R$ up to
some absolute/relative error, then the pattern frequency vector is well defined for almost
all $w \in R$ and takes only finitely many values with positive probability. Also note
that neither $F_1,\ldots,F_n$ nor $H_1,\ldots,H_n$ are required to be pairwise different.
Hence, it may happen that there exist $i \neq j$ such that $H_i \neq H_j$ and $F_i = F_j$ (or 
$H_i = H_j$ and $F_i \neq F_j$).


%% file: onecounter-new.tex
\section{Results for one-counter pVASS}
\label{sec-one-counter}

In this section we concentrate on analyzing the pattern frequency
vector for one-dimensional pVASS. We show that
$F_\A$ is well defined and takes at most $|Q| + b$ distinct
values for almost all runs, where $|Q|$ is the number of control states of~$\A$,
and $b$ is the number of BSCCs of $\C_{\A}$. 
Moreover, these values as well as the
associated probabilities can be efficiently approximated up to an
arbitrarily small positive relative error. More precisely, our aim is to
prove the following:
%
\begin{theorem}
\label{thm-one-counter-main}
	Let $\A = (Q,\gamma,W)$ be a one-dimensional pVASS, and let $b$ be the number of BSCCs of $\C_{\A}$.	
	Then there is $n \leq |Q| + b$ computable in time polynomial in $\size{\A}$
	such that for every $\varepsilon > 0$, there are 
	\mbox{$H_1,\ldots,H_n : \Pat_\A \rightarrow \Qset$} computable
	in time polynomial in $\size{\A}$ and $\size{\varepsilon}$, such that for
	every initial configuration $p(k) \in \conf(\A)$ there are 
	$P_1,\ldots,P_n \in \Qset$ computable in time
	polynomial in $\size{\A}$, $\size{\varepsilon}$, and $k$, such that the sequence
	$(P_1,H_1),\ldots,(P_n,H_n)$ approximates
	the pattern frequencies of $\run(p(k))$ up to the relative error~$\varepsilon$.
\end{theorem}

\noindent
Let us note that the ``real'' pattern frequency vectors $F_i$ as well as the probabilities $\calP(F_\A {=} F_i)$ may take irrational values, and they cannot be computed precisely in general.

 %

\begin{remark}
The $|Q| + b$ upper bound on~$n$ given in Theorem~\ref{thm-one-counter-main} is tight. To see this, 
realize that if $|Q| = 1$, then $b = 1$ and the trivial pVASS with the only rule  $p \ltran{0} p$ witnesses that 
the pattern frequency vector may take two different values.  If $|Q| \geq 2$,
we have that $b \leq |Q|-1$. Consider 
a pVASS where $Q = \{p,q_1,\ldots,q_k\}$ and $\gamma$ contains the rules $p \ltran{-1} p$, $p \ltran{-1} q_i$,
and $q_i \ltran{0} q_i$ for all $1 \leq i \leq k$, where all of these rules have the same weight equal to~$1$. For $p(2)$ as the initial configuration, the vector $F_\A$ takes $2k + 1 = 2|Q|-1$ pairwise different values with positive probability.
\end{remark}

For the rest of this section, we fix a one-dimensional pVASS $\A = (Q,\gamma,W)$.
We start by identifying certain (possibly empty) subsets of configurations called \emph{regions}
that satisfy the following properties:
\begin{itemize}
	\item there are at most $|Q| +b $ non-empty regions;
	\item almost every run eventually stays in precisely one region;
	\item almost all runs that stay in a given region have the same well defined
	   pattern frequency vector. 
\end{itemize}
In principle, we might proceed by considering each region $R$ separately and computing/approximating the 
associated pattern frequency vector and the probability of all runs that stay in~$R$. However, this would
lead to unnecessary technical complications. Instead, we identify situations when multiple regions share the \emph{same} pattern frequency vector, consider unions of such regions (called \emph{zones}), and then compute/approximate the pattern frequency vector and the probability of staying in~$Z$ for each zone~$Z$.
Thus, we obtain Theorem~\ref{thm-one-counter-main}. 

Technically, we distinguish among four \emph{types} of regions determined either by a control
state of $\A$ or a BSCC of $\C_\A$. 

\begin{itemize}
	\item Let $p \in Q$. A \emph{type~I region determined by $p$} is either
	the set  $\post(p(0))$ or the empty set, depending on whether $\post(p(0))$
	is a finite set satisfying $\post(p(0)) \subseteq \pre(p(0))$ or not, respectively.
	\item Let $p \in S$, where $S$ is a BSCC of $\C_\A$. A \emph{type~II region determined
		by $p$} is either
	the set  $\post(p(0))$ or the empty set, depending on whether $\post(p(0))$
	is an infinite set satisfying $\post(p(0)) \subseteq \pre(p(0))$ or not, respectively.
	\item Let $S$ be a BSCC of $\C_\A$. 
	  A \emph{type~III region determined by $S$} consists of all  
	  $p(k) \in S \times \Nset^+$ that cannot reach a configuration with
	  zero counter.
	\item Let $S$ be a BSCC of $\C_\A$, and let $R_{I}(S)$ and $R_{II}(S)$
	   be the unions of all type~I and all type~II regions determined by the control 
	   states of~$S$, respectively. Further, let $D(S)$ be the set
       \[
          \bigg(S {\times} \Nset \ \cap\   \pre(R_{I}(S))\bigg) \smallsetminus 
            \bigg(R_{I}(S) \cup \pre(R_{II}(S))\bigg)
       \]
       A \emph{type~IV region determined by $S$}
       is either the set $D(S)$ or the empty set, depending on whether $D(S)$ is infinite or finite, respectively.
\end{itemize} 

\noindent
Note that if $R_1,R_2$ are regions of $\A$ such that $R_1 \cap R_2 \neq \emptyset$, then
$R_1 = R_2$. Also observe that regions of type~I,~II, and~III are closed under $\post$, and each such
region can thus be seen as a Markov chain. Finally, note that
every configuration of a type~IV region can reach a configuration of a type~I region,
and the size of every type~I region is bounded by $|Q|^2$ (if $R = \post(p(0))$ is
a type~I region and $p(0) \tran{}^* q(j)$, then $j < |Q|$, because otherwise the 
counter could be pumped to an arbitrarily large value; hence, $|R| \leq |Q|^2$).   

Let us note that all regions are \emph{regular} in the following sense:
We say that a set $C \subseteq \conf(\A)$ of configurations is \emph{regular} if there is a 
non-deterministic finite automaton $A$ over the alphabet $\{a\}$ such that the set of control states
of $A$ subsumes $Q$ and for every configuration $p(k) \in \conf(\A)$ we have that $p(k) \in C$ iff
the word $a^k$ is accepted by $A$ with $p$ as the initial state. If follows, e.g., from the results
of \cite{EHRS:MC-PDA} that if $C\subseteq \conf(\A)$ is regular, then $\post(C)$ and
$\pre(C)$ are also regular and the associated NFA are computable in time polynomial in $\size{A}$, where
$A$ is the NFA representing~$C$. Hence, all regions are effectively regular which becomes important
in Section~\ref{sec-two-counters}.

Let $S$ be a SCC of $\C_\A$. If $S$ is not a BSCC of $\C_\A$, then the control states
of $S$ may determine at most $|S|$ non-empty regions (of type~I). If $S$ is a BSCC of $\C_\A$, then the control states of $S$ may determine at most $|S|$ non-empty regions
of type~I or~II, and at most one additional non-empty region which is either of type~III or
of type~IV (clearly, it cannot happen that the type~III and type~IV regions determined
by $S$ are both non-empty). Hence, the total number of non-empty regions cannot
exceed $|Q| + b$, where $b$ is the number of BSCCs of~$\C_\A$ (here we also use the assumption 
that $\C_\A$ is weakly connected).

Now we prove that \emph{every} configuration can reach some region in a bounded number of steps. This fact is particularly important for the analysis of two-counter pVASS in Section~\ref{sec-two-counters}.

\begin{lemma}
\label{lem-bounded}
	Every configuration of $\A$ can reach a configuration of some region in at most $11|Q|^4$ transitions.
\end{lemma}
  
By Lemma~\ref{lem-bounded}, the probability of reaching (some) region from an arbitrary initial configuration is at least $\xmin^{11|Q|^4}$, where $\xmin$ is the least positive 
transition probability of $\M_\A$. This implies that almost every $w \in \run(p(k))$ visits some region~$R$. If $R$ is of type~I,~II, or~III, then $w$ inevitably stays in $R$ because these regions are closed under $\post$. If $R$ is a type~IV region, then $w$ either stays in $R$, or later visits a configuration of a type~I region where it stays. Thus, we obtain the following:

\begin{lemma}
\label{lem-stay-in-one}
   Let $p(k)$ be a configuration of $\A$. Then almost every run initiated in $p(k)$
   eventually stays in precisely one region. 
\end{lemma}

As we already mentioned, computing the pattern frequency vector and the probability
of staying in $R$ for each region $R$ separately is technically complicated. Therefore,
we also introduce \emph{zones}, which are unions of regions that are guaranteed to share
the same pattern frequency vector. Formally, a \emph{zone of $\A$} is a set $Z \subseteq \conf(\A)$ satisfying one of the following conditions (recall that $t_S$ denotes the trend of a BSCC $S$):
\begin{itemize}
  \item $Z = R$, where $R$ is a region of type~I.
  \item $Z = R$, where $R$ is a type~III region determined by a BSCC $S$ of $\C_\A$ such that
     $t_S \leq 0$.
  \item $Z = R$, where $R$ is a type~II region determined by $p \in S$ where $S$ is a BSCC of $\C_\A$ 
     satisfying $t_S < 0$.
  \item $Z = R_{II}(S)$, where $S$ is a BSCC of $\C_\A$ such that $t_S = 0$ and $R_{II}(S)$
     is the union of all type~II regions determined by the control states of $S$.
  \item $Z = R_{II}(S) \cup R_{III}(S) \cup R_{IV}(S)$, where $S$ is a BSCC of $\C_\A$ such that $t_S > 0$,
     $R_{II}(S)$ is the union of all type~II regions determined by the control states of $S$, and
     $R_{III}(S)$ and $R_{IV}(S)$ are the type~III and the type~IV regions determined by~$S$, respectively.
\end{itemize}
	
\noindent
The next two lemmata are nontrivial and represent the technical core of this 
section (proofs can be found in Appendix~\ref{app-one-counter}). They crucially depend on the results presented recently in \cite{BKK:pOC-time-LTL-martingale-JACM} and~\cite{SEY:pOC-poly-Turing}.
In the proof of Lemma~\ref{lem-reg-same-value}, we also characterize situations when some
elements of pattern frequency vectors take irrational values.

\begin{lemma}
	\label{lem-reg-same-value}
	Let $p(k)$ be a configuration of $\A$ and $Z$~a zone of~$\A$. Then $F_\A$ is well defined 
	for almost all $w \in \run(p(k),Z)$, and there exists $F : \Pat_\A \rightarrow \Rset$
	such that $F_\A(w) = F$ for almost all $w \in \run(p(k),Z)$. Further, for every
	rational $\varepsilon >0$, there is a vector $H : \Pat_\A \rightarrow \Qset$ computable in time polynomial in
	$\size{\A}$ and~$\size{\varepsilon}$ such that $H(q\alpha)$ approximates $F(q\alpha)$ up to
	the relative error~$\varepsilon$ for every $q\alpha \in \Pat_\A$. 
\end{lemma}


\begin{lemma}
	\label{lem-one-counter-approx}	
	Let $p(k)$ be a configuration of $\A$. Then almost every run initiated in $p(k)$ eventually stays
	in precisely one zone of~$\A$. Further, for every zone $Z$ and every rational $\varepsilon >0$, there 
	is a $P \in \Qset$ computable in time polynomial in $\size{\A}$, $\size{\varepsilon}$, and $k$ such that
	$P$ approximates $\calP(\run(p(k),Z))$ up to the relative error $\varepsilon$. 
\end{lemma}

%% file: twocounter-new.tex
\section{Results for two-counter pVASS}
\label{sec-two-counters}

In this section we analyze the long-run average behavior of two-counter pVASS.
We show that if a given two-counter pVASS is \emph{stable} (see Definition~\ref{def-stable} below), then the pattern frequency vector is well defined takes one of finitely many values for almost all runs. Further, these values and the associated probabilities can be effectively approximated up to an arbitrarily small positive absolute/relative error. 

Let $\A$ be a two-counter pVASS.
When we say that some object (e.g., a number or a vector) is \emph{computable} for every $\sigma \in \Sigma$,
where $\Sigma$ is some set of parameters, we mean that there exists an algorithm which inputs the encodings 
of~$\A$ and $\sigma$, and outputs the object. Typically, the parameter $\sigma$ is some rational $\varepsilon > 0$,
of a pair $(\varepsilon,p\vec{v})$ where $p\vec{v}$ is a configuration. The parameter can also be void,
which means that the algorithm inputs just the encoding of $\A$.

A \emph{semilinear constraint} $\varphi$ is a function $\varphi : Q \rightarrow \Phi$,
where $\Phi$ is the set of all formulae
of Presburger arithmetic with two free variables $x,y$. Each $\varphi$ determines
a semilinear set $\sem{\varphi} \subseteq \conf(\A)$ consisting of all $p(v_1,v_2)$ such that
$\varphi(p)[x/v_1,y/v_2]$ is a valid formula. Since the reachability relation $\tran{}^*$
of $\A$ is effectively semilinear \cite{LS:two-counter-VASS-flat} and semilinear sets are closed under complement and union, all of the sets of configurations we work with (such as $C[R_1,R_2]$ defined below) are effectively semilinear, i.e., the associated semilinear constraint is computable. In particular, the membership problem for these sets is decidable.

Given $p\vec{v} \in \conf(\A)$ and 
$D \subseteq \conf(\A)$, we use $\run(p\vec{v} \rightarrow^* D)$ to denote the set of
all $w \in \run(p\vec{v})$ that visit a configuration of~$D$, and 
$\run(p\vec{v} \not\rightarrow^* D)$ to denote the set 
$\run(p\vec{v}) \smallsetminus \run(p\vec{v} \rightarrow^* D)$. Note that if $D = \emptyset$,
then $\run(p\vec{v} \not\rightarrow^* D) = \run(p\vec{v})$.

Intuitively, our aim is to prove that the set $C = \conf(\A)$ is ``good''
in the sense that there is a computable $n\in \Nset$ such that for every rational 
$\varepsilon>0$, there exists a computable sequence of rational vectors \mbox{$H_1,\ldots,H_n$} 
such that for every $p\vec{v} \in C$, there are computable rational $P_1,\ldots,P_n$ such that
the sequence $(P_1,H_1),\ldots,(P_n,H_n)$ that approximates the pattern frequencies of $\run(p\vec{v})$ 
up to the absolute/relative error $\varepsilon$.
This is achieved by first showing that certain simple subsets of configurations are good, and then
(repeatedly) demonstrating that more complicated subsets are also good because they can be ``reduced'' to
simpler subsets that are already known to be good. Thus, we eventually prove that the whole set $\conf(\A)$ is good. 

For our purposes, it is convenient to parameterize the notion of a ``good'' subset $C$ by another 
subset of ``dangerous'' configurations $D$ so that the above conditions are required to hold only 
for those runs that do not visit~$D$. Further, we require that every configuration of $C$ can avoid 
visiting $D$ with some positive probability which is bounded away from zero.


\begin{definition}
	\label{def-good}
	Let $\A = (Q,\gamma,W)$ be a pVASS of dimension~two, and let $C,D \subseteq \conf(\A)$.
	We say that $C$ is \emph{good for $D$} if the following conditions are satisfied:
	\begin{itemize}
		\item There is $\delta > 0$ such that 
		$\calP(\run(p\vec{v} \rightarrow^* D)) \leq 1 - \delta$ for every
		$p\vec{v} \in C$.
	    \item There is 
		a computable $n \in \Nset$ such that for every $\varepsilon > 0$, there are computable \mbox{$H_1,\ldots,H_n : \Pat_\A \rightarrow \Qset$} such that for every 
		$p\vec{v} \in C$ there are computable $P_{p\vec{v},1},\ldots,P_{p\vec{v},n} \in \Qset$ such that
		$(P_{p\vec{v},1},H_1),\ldots,(P_{p\vec{v},n},H_n)$ approximate the 
		pattern frequencies of $\run(p\vec{v}\not\rightarrow^* D)$ up 
		to the absolute error~$\varepsilon$.
    \end{itemize}
\end{definition}

\noindent
Note that in Definition~\ref{def-good}, we require that $(P_{p\vec{v},1},H_1),\ldots,(P_{p\vec{v},n},H_n)$ approximate the 
pattern frequencies of $\run(p\vec{v}\not\rightarrow^* D)$ up 
to the \emph{absolute} error~$\varepsilon$. As we shall see, we can always compute 
a lower bound for each positive $P_{p\vec{v},i}$ and $H_i$, which implies that
if $P_{p\vec{v},i}$ and $H_i$ can be effectively approximated up to an arbitrarily
small absolute error~$\varepsilon >0$, they can also be 
effectively approximated up to an arbitrarily small \emph{relative}
error~$\varepsilon >0$. 

The next definition and lemma explain what we mean by reducing the analysis of runs initiated in configurations of $C$ to the analysis of runs initiated in ``simpler'' configurations of $C_1,\ldots,C_k$. 

\begin{definition}
\label{def-reduce} Let $\A$ be a pVASS of dimension~two,
	$C \subseteq \conf(\A)$, and $\calE = \{C_1,\ldots,C_k\}$ a set of pairwise disjoint subsets of 
    $\conf(\A)$. We say that $C$ is \emph{reducible} to $\calE$ 
	if, for every $\varepsilon > 0$, there are computable semilinear constraints 
	$\varphi_1,\ldots,\varphi_k$ such that 
	\begin{itemize}
		\item $\sem{\varphi_i} \subseteq C_i$ for every $1 \leq i \leq k$;
		\item for all $1 \leq i \leq k$ and  $p\vec{v} \in \sem{\varphi_i}$,
		we have that $\calP(\run(p\vec{v} \rightarrow^* D_i)) \leq \varepsilon$, where $D_i = \bigcup_{j\neq i} C_j$.
		\item for every $p\vec{v} \in C$ and every $\delta > 0$, there is a computable
		$\ell \in \Nset$ such that the probability of reaching a configuration
		of $\sem{\varphi_1} \cup \cdots \cup \sem{\varphi_k}$ in at most $\ell$ transitions is at least $1 - \delta$.
	\end{itemize}
\end{definition}

\begin{lemma}
	\label{lem-reduce}
	If $C$ is reducible to $\calE = \{C_1,\ldots,C_k\}$ and every $C_i$ is good for $D_i = \bigcup_{j\neq i} C_j$, 
	then $C$ is good for~$\emptyset$. 
\end{lemma}
\begin{proof}
	For every $1 \leq i \leq k$, let $n_i$ be the computable constant for $C_i$ which
	exists by Definition~\ref{def-good}. The constant $n$ for $C$ is defined
	as $n = \sum_{i = 1}^k n_i$. Now let us fix some $\varepsilon > 0$. 
	Since $C$ is reducible to $\{C_1,\ldots,C_k\}$, there are computable
	constraints $\varphi_1,\ldots,\varphi_k$  such that, for every $1 \leq i \leq k$,
	we have that $\sem{\varphi_i} \subseteq C_i$ and 
	\mbox{$\calP(\run(p_i\vec{v}_i \rightarrow^* D_i)) \leq \varepsilon/4$} 
	for every $p_i\vec{v}_i \in \sem{\varphi_i}$. Further, there
	are computable \mbox{$H_{i,1},\ldots,H_{i,n_i} : \Pat_\A \rightarrow \Qset$} such that for every $p_i\vec{v}_i \in \sem{\varphi_i}$, there are computable $P_{p_i\vec{v}_i,1},\ldots,P_{p_i\vec{v}_i,n_i} \in \Qset$ such that
	$(P_{p_i\vec{v}_i,1},H_{i,1}),\ldots,(P_{p_i\vec{v}_i,n_i},H_{i,n_i})$ approximate the pattern frequencies of 
	$\run(p_i\vec{v}_i\not\rightarrow^* D_i)$ up 
	to the absolute error~$\varepsilon/4$. Now let $p\vec{v} \in C$. Then there is a computable $\ell \in \Nset$ such that the probability of reaching a configuration
	of $\sem{\varphi_1} \cup \cdots \cup \sem{\varphi_k}$ in at most $\ell$ transitions is at least $1 - \varepsilon/4$. Hence, we can effectively construct a finite tree $T$ rooted by $p\vec{v}$ which represents the (unfolding of) the part of $\M_\A$ reachable from
	$p\vec{v}$. A branch in this tree is terminated when a configuration of
	$\sem{\varphi_1} \cup \cdots \cup \sem{\varphi_k}$ is visited, or when the length
	of the branch reaches~$\ell$. For every $1 \leq i \leq k$, let $L_i$ be the set
	of all leafs $\alpha$ of $T$  labeled by configurations of $\sem{\varphi_i}$.
	We use $P_\alpha$ to denote the (rational and computable) probability of reaching
	$\alpha$ from the root of $T$, and $\mathit{label}(\alpha)$ to denote the configuration
	which is the label of $\alpha$. For every
	$1 \leq i \leq k$ and every $1 \leq j \leq n_i$, we put 
	$P_{p\vec{v},i,j} = \sum_{\alpha \in L_i} P_\alpha \cdot P_{\mathit{label}(\alpha),j}$. It is 
	straightforward to verify that the sequence
	\begin{align*}
	   &(P_{p\vec{v},1,1},H_{1,1}),\ldots, (P_{p\vec{v},1,n_1},H_{1,n_1}),\\	   
	   &(P_{p\vec{v},2,1},H_{2,1}),\ldots, (P_{p\vec{v},1,n_1},H_{2,n_2}),\\
	   & \quad\vdots\\
	   &(P_{p\vec{v},k,1},H_{k,1}),\ldots, (P_{p\vec{v},k,n_k},H_{k,n_k})
	\end{align*}
	approximates
	the pattern frequencies of $\run(p\vec{v})$ up to the absolute error~$\varepsilon$.
	In particular, realize that almost every $w \in \run(p\vec{v})$ eventually ``decides'' for some $C_i$, i.e., there is $m \in \Nset$ such that $w(m) \in C_i$ and for all $m' > m$ we have $w(m') \not\in D_i$ (this is where we use
	the first condition of Definition~\ref{def-good}). Hence, the pattern frequency
	vector is well defined and approximated up to the absolute error $\varepsilon/4$ 
	by some of the above $H_{i,j}$ for almost all $w \in \run(p\vec{v})$. 
\end{proof}

%



For the rest of this section, we fix a two-counter pVASS $\A = (Q,\gamma,W)$ (recall that
$\A$ satisfies Assumption~\ref{asm-simpleVASS}). For $\mi \in \{1,2\}$, we define
a one-counter pVASS $\A_{\mi} = (Q,\gamma_{\mi},W_{\mi})$ and a labeling 
$L_{\mi} : \gamma_{\mi} \rightarrow \{-1,0,1\}$ as follows:
$s \ltran{\kappa(\mi),\ell} t$ in $\A_{\mi}$ and 
$L_\mi ((s,\kappa(\mi),t)) = \kappa(3{-}\mi)$  iff 
$s \ltran{\kappa,\ell} t$ in $\A$.
Note that $\A_\mi$ is obtained by ``preserving'' the $\mi$-th counter; the change 
of the other counter is encoded in $L_\mi$. Also observe that $\C_\A$, $\C_{\A_1}$,
and $\C_{\A_2}$ are the same Markov chains.

The results of Section~\ref{sec-one-counter} are applicable to $\A_1$ and $\A_2$.
Let $R$ be a type~II or a type~IV region of $\A_i$. We claim that there is a unique $\tau_R \in \Rset$ such that for almost all runs $w \in \run(p(k),R)$, where $p(k) \in R$, we have that the limit 
\[
\lim_{n\rightarrow \infty} \frac{\sum_{j=0}^{n-1}L_i(\mathit{rule}(w,j))}{n}
\]
exists and it is equal to $\tau_R$ (here, $\mathit{rule}(w,j)$) is the unique
rule of $\gamma_i$ which determines the transition $w(j) \tran{} w(j{+}1)$; cf.{}
Assumption~\ref{asm-simpleVASS}). In other words, $\tau_R$ is the unique \emph{mean payoff}
determined by the labeling $L_i$ associated to~$R$. To see this, consider the trend $t_S$
of the associated BSCC $S$ of $\C_\A$. If $R$ is a type~IV region,  
then $\tau_R = t_S(3{-}i)$ for almost all $w \in \run(p(k),R)$ (in particular, note that if 
$t_S(i) \leq 0$ then $\calP(\run(p(k),R)) = 0$; see Section~\ref{sec-one-counter}).
If $t_S(i) \geq 0$ and $R$ is a type~II region, 
then $\tau_R = t_S(3{-}i)$, because the frequency of visits to configurations with zero counter is zero for almost all $w \in \run(p(k))$, where $p(k) \in R$ (see \cite{BKK:pOC-time-LTL-martingale-JACM}). Finally,
if $t_S(i) < 0$ and $R$ is a type~II region, then $R$ is ergodic because the mean recurrence time in every configuration of $R$ is finite 
\cite{BKK:pOC-time-LTL-martingale-JACM}, and hence
$\tau_R$ takes the same value for almost all $w \in \run(p(k),R)$, where $p(k) \in R$.

Although the value of $\tau_R$ may be irrational when $R$ is of type~II and $t_S(i) < 0$, there exists a formula $\Phi(x)$ of Tarski algebra with a fixed 
alternation depth of quantifiers computable in polynomial time such that $\Phi[x/c]$ 
is valid iff $c = \tau_R$. Hence, the problem whether $\tau_R$ is zero (or positive, or negative) is decidable in exponential time \cite{Grigoriev:Tarski-exponential-JSC};
and if $\tau_R < 0$ (or $\tau_R > 0$), there is a computable $x \in \Qset$ such that
$x < 0$ (or $x > 0$) and $|x| \leq |\tau_R|$. 

\begin{definition}
\label{def-stable}
   Let $\A = (Q,\gamma,W)$ be a pVASS of dimension~two. We say that $\A$ is \emph{stable}
   if the following conditions are satisfied:
   \begin{itemize}
   	  \item Let $S$ be a BSCC of $\C_{\A}$ such that the type~IV region determined
   	     by $S$ is non-empty in $\A_1$ or $\A_2$, or there is $p \in S$ such that the type~II region determined by $p(0)$ is~non-empty in $\A_1$ or $\A_2$. Then the trend
   	     $t_S$ is non-zero in both components.
   	  \item Let $R$ by a type~II region in $\A_i$ such that $t_S(i) < 0$, where $S$
   	     is the BSCC of  $\C_{\A}$ associated to~$R$. Then $\tau_R \neq 0$.
   \end{itemize}
\end{definition}

\noindent
Note that the problem whether a given two-counter pVASS $\A$ is stable is decidable
in exponential time. Our aim is to prove the following theorem:

\begin{theorem}
	\label{thm-two-counter-main}
	Let $\A = (Q,\gamma,W)$ be a stable pVASS of dimension~two. Then the set $\conf(\A)$
	is good for $\emptyset$.
\end{theorem} 

For the rest of this section, we fix a pVASS $\A$ of dimension two a present a sequence of observations
that imply Theorem~\ref{thm-two-counter-main}. Note that $\A$ is not
necessarily stable, i.e., the presented observations are valid for \emph{general} two-dimensional 
pVASS. The stability condition is used to rule out some problematic subcases that are not 
covered by these observations. 

In our constructions, we need to consider the following subsets of configurations:
\begin{itemize}
	\item $C[R_1,R_2]$, where $R_1 \in \Reg(\A_1)$ and $R_2 \in \Reg(\A_2)$, is the set 
	 of all $p(m_1,m_1) \in \conf(\A)$ such that $p(m_1) \in R_1$ and $p(m_2) \in R_2$;
    \item $B[b]$, where $b \in \Nset$, consists of all $p\vec{v} \in \conf(\A)$ such that
     for every $q\vec{u} \in \post(p\vec{v})$ we have that $\vec{u}(1) \leq b$
     or $\vec{u}(2) \leq b$;
	\item $C_S[c_1 {\sim} b_1 \wedge c_2 {\approx} b_2 ]$, where $S \subseteq Q$,
	 $b_1,b_2 \in \Nset$, and ${\sim},{\approx}$ are numerical comparisons (such as
	 $=$ or $\leq$) consists of all $p(m_1,m_2) \in \conf(\A)$ such that $p \in S$, 
	 $m_1 \sim b_1$, and $m_2 \approx b_2$. Trivial constraints of the form $c_i \geq 0$ 
	 can be omitted. For example, $C_Q[c_1 = 0 \wedge c_2 \geq 6]$
	 is the set of all $q(0,m) \in \conf(\A)$ where $m \geq 6$, and $C_S[c_1 \leq 2]$ is the set of
	 all $q(n,m)\in \conf(\A)$ where $q \in S$ and $n \leq 2$.
    \item $Z_S$, where $S \subseteq Q$, consists of all $p(m_1,m_2)$ such that
     $p \in S$ and some counter is zero (i.e., $m_1 = 0$ or $m_2 = 0$).
    \item $E_S[b_1,b_2]$, where $S \subseteq Q$ and $b_1,b_2 \in \Nset$, consists of all $p(m_1,m_2)$ such that
     $p \in S$, some counter is zero, and every 
     $q(n_1,n_2) \in \post(p(m_1,m_2))$ satisfies the following:
     \begin{itemize} 
     	\item if $n_1 = 0$, then $n_2 \leq b_2$;
     	\item if $n_2 = 0$, then $n_1 \leq b_1$.
     \end{itemize} 
\end{itemize}


\noindent
Note that all of these sets are semilinear and the associated semilinear constraints are
computable.

A direct consequence of Lemma~\ref{lem-bounded} is the following:
\begin{lemma}
\label{lem-general}
	Let $b = 11|Q|^4$, and let $\calE$ be a set consisting of 
	\mbox{$B[b]$} and
	all $C[R_1,R_2]$ where $R_1 \in \Reg(\A_1)$,
	$R_2 \in \Reg(\A_2)$. Then $\conf(\A)$ is reducible to $\calE$.
\end{lemma}
To prove Lemma~\ref{lem-general}, it suffices to realize that there is a computable
$k \in \Nset$ such that \emph{every} $p\vec{v} \in \conf(\A)$ can reach a configuration
of some $C[R_1,R_2]$ or $B[b]$ in at most $k$~transitions. 

Hence, it suffices to prove that \mbox{$B[b]$}
and all $C[R_1,R_2]$ are good for~$\emptyset$. All cases except for those where
$R_1$ and $R_2$ are of type~II or type~IV follow almost immediately. To handle 
the remaining cases, we need to develop new tools, which we present now.
We start by introducing some notation.

Given a finite path or a run $w$ in $\M_{\A}$ and $\ell\in \Nset$, where \mbox{$\ell \leq \len(w)$}, we denote by $\xs\ell_1(w)$, $\xs\ell_2(w)$, and $\ps\ell(w)$ the value of the first counter, the value of the second counter, and the control state of the configuration $w(\ell)$, respectively.
Further, $T(w)$ denotes either the least $\ell$ such that $\xs\ell_1(w)=0$, or $\infty$ if there is no such $\ell$. For every $i \in \Nset$, 
\mbox{$[p\vec{v}\rightarrow^* q\vec{u},i]$} denotes the probability of all  $w \in \run(p\vec{v})$ such that \mbox{$T(w)\geq i$}, $w(i)=q\vec{u}$, and $w(j)\not = q\vec{u}$ for all $0\leq j<i$. By $[p\vec{v}\rightarrow^* q\vec{u}]=\sum_{i=0}^{\infty} [p\vec{v}\rightarrow^* q\vec{u},i]$ we denote the probability of reaching $q\vec{u}$ from $p\vec{v}$ before time $T$. We also put 
\[
[p\vec{v}\rightarrow^* q(0,*),i] = \sum_{k=0}^{\infty} [p\vec{v}\rightarrow^* q(0,k),i]
\]
and 
\[
 [p\vec{v}\rightarrow^* q(0,*)] = \sum_{k=0}^{\infty} [p\vec{v}\rightarrow^* q(0,k)]\, .
\]
For a measurable function $X$ over the runs of $\M_\A$, we use $\E_{p\vec{v}}[X]$ to denote
the expected value of $X$ over $\run(p\vec{v})$. 

The following theorems are at the very core of our analysis, and represent
new non-trivial quantitative bounds obtained by designing and analyzing a 
suitable martingale. Proofs can be found in \cite{BKKN:pVASS-frequency-arxiv}. 

\begin{theorem}\label{thm:tail-bounds-height}
Let $S$ be a BSCC of $\C_\A$ such that $t_S(2)<0$, and let $R$ be a type~II region of $\A_2$ determined by some state of~$S$.
Then there are rational $a_1,b_1 > 0$ and $0 < z_1 <1$ computable in polynomial space such that the following holds for all $p(0)\in R$, $n \in \Nset$, and $i\in \Nset^+$:
\[
\calP_{p(n,0)} (T < \infty \wedge 
\xs{T}_2\geq i)\quad \leq\quad a_1\cdot z_1^{b_1 \cdot i}\, .
\]
Moreover, if $\calP_{p(n,0)}(T<\infty) = 1$, then 
\[
\E_{p(n,0)}\left[\xs{T}_2\right]\quad \leq\quad  \frac{a_1\cdot z_1^{b_1}}{1-z_1^{b_1}}.
\]
In particular, \emph{none} of the bounds depends on $n$.
\end{theorem}



\begin{theorem}\label{thm:tail-bounds-divergence}
Let $S$ be a BSCC of $\C_\A$ such that $t_S(2)<0$, and let $R$ be a type~II region determined by 
some state of $S$ such that $\tau_R>0$. Then there are rational  
$a_2,b_2>0$ and $0<z_2<1$ computable in polynomial space such that for all configurations $p(n,0)$, where $p(0) \in R$, and all $q\in Q$, the following holds:
\[
[p(n,0)\rightarrow^* q(0,*)]\quad \leq \quad n\cdot a_2\cdot z_2^{n \cdot b_2}
\]
\end{theorem}


\begin{theorem}\label{thm:tail-bounds-convergence}
Let $R$ be a type~II region of $\A_2$ such that $\tau_R<0$. Then there are rational
$a_3,b_3,d_3>0$ and $0<z_3<1$ computable in polynomial space such that for all configurations 
$p(n,0)$, where $p(0) \in R$, and all $q\in Q$, the following holds for all
$i\geq \frac{H\cdot n}{-{\tau_R}}$, where $H$ is computable in polynomial space:
\[
[p(n,0)\rightarrow^* q(0,*),i]\quad \leq \quad i\cdot a_3\cdot 
    z_3^{\sqrt{n\cdot\tau_R\cdot b_3 + i\cdot d_3}}\, .
\]
\end{theorem}


The above theorems are use to prove that certain configurations are 
\emph{eagerly attracted} by certain sets of configurations in the following sense:


\begin{definition}
\label{def-attractor}
   Let $C,D \subseteq \conf(\A)$. We say that 
   $p\vec{v} \in C$ is \emph{eagerly attracted} by 
   $D$ if $\calP(\run(p\vec{v} \tran{}^* D)) = 1$ and there are computable constants 
   $a,z \in \Qset$, $\ell \in \Nset$, and $k\in \Nset^+$ (possibly dependent on $p\vec{v}$), where $a > 0$ and $0<z<1$,
   such that for every $\ell' \geq \ell$, the probability of visiting $D$ from 
   $p\vec{v}$ in at most $\ell'$ transitions is at least  $1 - a\cdot z^{\sqrt[k]{\ell'}}$.
   Further, we say that $C$ is eagerly attracted by $D$ if all configurations of $C$
   are eagerly attracted by $D$, and  $D$ is a \emph{finite eager attractor} if $D$
   is finite and $\post(D)$ is eagerly attracted by~$D$.
\end{definition}

Markov chains with finite eager attractors were studied in~\cite{AHMS:Eager-limit}.
The only subtle difference is that in \cite{AHMS:Eager-limit}, the probability of revisiting the attractor in at most $\ell$ transitions is at least $1 - z^\ell$. However,
all arguments of \cite{AHMS:Eager-limit} are valid also for the sub-exponential bound $1 - a\cdot z^{\sqrt[k]{\ell'}}$
adopted in Definition~\ref{def-attractor} (note that some quantitative bounds given in~\cite{AHMS:Eager-limit}, such as the bound on $K$ in Lemma~5.1 of~\cite{AHMS:Eager-limit}, need to be slightly adjusted to accommodate the sub-exponential bound). In \cite{AHMS:Eager-limit}, it was shown that
various limit properties of Markov chains with finite eager attractors can be effectively
approximated up to an arbitrarily small absolute error $\varepsilon > 0$.
A direct consequence of these results is the following:

\begin{proposition}
\label{prop-attractor}
	Let $D \subseteq \conf(\A)$ be a finite eager attractor. Then $D$ is good 
	for~$\emptyset$.
\end{proposition}

\noindent
Let us also formulate one simple consequence of Theorem~\ref{thm:tail-bounds-divergence}.

\begin{corollary}
	\label{cor-trend-positive}
	For every BSCC $S$ of $\C_\A$ we have  the following:
	\begin{itemize}
		\item If $t_S$ is negative in some component, then every configuration $p\vec{v}$ where
		$p \in S$ is eagerly attracted by~$Z_S$.
		\item If both components of $t_S$ are positive, then for every $\varepsilon >0$ there
		is a computable $b_\varepsilon$ such that for every configuration $p\vec{v}$ where
		$p \in S$ and $\vec{v} \geq (b_\varepsilon,b_\varepsilon)$ we have that
		$\calP(\run(p\vec{v} \tran{}^* Z_S)) < \varepsilon$. 
	\end{itemize}
\end{corollary}  

The following theorem follows from the results about one-counter pVASS presented 
in \cite{BKK:pOC-time-LTL-martingale-JACM}.
\begin{theorem}
	\label{thm-zeros}
	For every BSCC $S$ of $\C_\A$ we have  the following:
	\begin{itemize}
		\item If $t_S$ is negative in some component, then every configuration $p\vec{v}$ where
		$p \in S$ is eagerly attracted by~$Z_S$.
		\item If both components of $t_S$ are positive, then for every $\varepsilon >0$ there
		is a computable $b_\varepsilon$ such that for every configuration $p\vec{v}$ where
		$p \in S$ and $\vec{v} \geq (b_\varepsilon,b_\varepsilon)$ we have that
		$\calP(\run(p\vec{v} \tran{}^* Z_S)) \leq \varepsilon$. 
	\end{itemize}
\end{theorem}  

\noindent
In the next lemmata, we reduce the study of pattern frequencies
for certain runs in $\M_\A$ to the study of pattern frequencies for runs in
one-counter pVASS (i.e., to the results of Section~\ref{sec-one-counter}). 
This is possible because in each of these cases,
one of the counters is either bounded or irrelevant. Proofs of the following lemmata are
straightforward.

\begin{lemma}
\label{lem-bounded-counter}
   For every $b \in \Nset$, the set \mbox{$B[b]$} is good for $\emptyset$.
\end{lemma}

\begin{lemma}
	\label{lem-typeI}
	The set $C[R_1,R_2]$, where $R_1$ or $R_2$ is a type~I or a type~III region, is good for $\emptyset$.
\end{lemma}


So, it remains to consider sets of the form $C[R_1,R_2]$, where the regions $R_1,R_2$
are of type~II or type~IV. We start with the simple case when the trend $t_S$ of
the associated BSCC is positive in both components.

\begin{lemma}
	\label{lem-diverging-trend}
	Let $C[R_1,R_2]$ be a set such that $R_1,R_2$ are regions of type~II or type~IV, and
	the trend $t_S$ of the associated BSCC $S$ of $\C_\A$ is positive in both components.
	Then $C[R_1,R_2]$ is good for $\emptyset$.
\end{lemma}
\begin{proof}
	Let $b \in \Nset$ be a bound such that for every $p\vec{v} \in \conf(\A)$ where
	$p \in S$ and $\vec{v}\geq (b,b)$ we have that there
	exists a ``pumpable path'' of the form $p\vec{v} \tran{}^* p(\vec{v}{+}{\vec{u}})$
	where $\vec{u}$ is positive in both components. Note that such a $b$ exists and
	it is computable (in fact, one can give an explicit upper bound on $b$ in the
	size of $S$; see, e.g., \cite{BJK:eVASS-games}). 
	
	By Lemma~\ref{lem-bounded-counter}, $B[b]$ is good for $\emptyset$. We show that 
	\mbox{$C_S[c_1 \geq b \wedge c_2 \geq b]$} is good for $B[b]$. By our choice of $b$
	and Theorem~\ref{thm-zeros}, there is $\delta > 0$ such that 
	\mbox{$\calP(\run(p\vec{v} \not\rightarrow^*  B[b])) \geq \delta$} for every 
    \mbox{$p\vec{v} \in C_S[c_1 \geq b \wedge c_2 \geq b]$}. Further, almost
	all runs of $\run(p\vec{v} \not\rightarrow^*  B[b])$ have the same pattern frequency vector
	$F_S$ where $F_S(q(*,*)) = \mu_S(q)$ for all $q \in S$, and $F_S(\alpha) = 0$
	for the other patterns. 
	
	Now we prove that $C[R_1,R_2]$ is reducible to 
	$\{B[b], C_S[c_1 \geq b \wedge c_2 \geq b]\}$. By Theorem~\ref{thm-zeros},
	we obtain that for every $\varepsilon > 0$ there is a computable $b_\varepsilon$ such that 
	for every configuration of $q\vec{u}$ where 
	$\vec{u} \geq (b_\varepsilon,b_\varepsilon)$ we have that 
	$\calP(\run(q\vec{u} \tran{}^* Z_S)) \leq \varepsilon$. Let $\varphi$ be a semilinear constraint
	where $\varphi(s) = x {\geq} b{+}b_\varepsilon \wedge y {\geq} b{+}b_\varepsilon$ for all
	$s \in S$, and  $\varphi(s) = \mathit{false}$ for all $s \in Q \smallsetminus S$. 
	Then $\sem{\varphi} \subseteq C_S[c_1 \geq b \wedge c_2 \geq b]$ and for every
    $q\vec{u} \in \sem{\varphi}$ we have that
	$\calP(\run(q\vec{u} \tran{}^* B[b])) \leq \varepsilon$. Further, there exists a computable $k \in \Nset$
	such that every configuration of $C[R_1,R_2]$ can reach a configuration of 
	\mbox{$B[b] \cup \sem{\varphi}$} in at most $k$~transitions.
	This implies that for every $\delta > 0$, there is a computable $\ell \in \Nset$
	such that every configuration of $C[R_1,R_2]$ reaches a configuration of \mbox{$B[b] \cup \sem{\varphi}$} in at most $\ell$ steps with probability
	at least $1-\delta$.
\end{proof}

To prove Theorem~\ref{thm-two-counter-main}, it suffices to show that the following 
sets of configurations are good for~$\emptyset$, where we disregard the subcases ruled out
by the stability condition. In particular, due to Lemma~\ref{lem-diverging-trend} we can
safely assume that some component of $t_S$ is negative.

\begin{itemize}
	\item[(a)] $C[R_1,R_2]$, where both $R_1$ and $R_2$ are of type~II.
	\item[(b)] $C[R_1,R_2]$, where $R_1$ is of type~IV and $R_2$ is of type~II, or
	$R_1$ is of type~II and $R_2$ is of type~IV.
	\item[(c)] $C[R_1,R_2]$, where both $R_1$ and $R_2$ are of type~IV.
\end{itemize}

The most interesting (and technically demanding) is the following subcase of Case~(a).
Here we only sketch the main ideas, a full proof can be found in 
\cite{BKKN:pVASS-frequency-arxiv}.

\begin{lemma}
\label{lem-II-II-gtrend-leftdown-trends-left-down}
	Let $C[R_1,R_2]$ be a set of configurations where 
	$R_1$ and $R_2$ are of type~II, $t_S(2) < 0$, $\tau_{R_1}<0$, and $\tau_{R_2} < 0$. 
	Then $C[R_1,R_2]$ is good for~$\emptyset$.
\end{lemma}
\begin{proof}[Proof Sketch]
Let $C$ be the set of all configurations of the form $q(0,m)\in C[R_1,R_2]$ satisfying 
$m\leq (a_1\cdot z_1^{b_1})/(1-z_1^{b_1})$, where $a_1,b_1,z_1$ are the computable constants of Theorem~\ref{thm:tail-bounds-height}. We prove that $C[R_1,R_2]$ is eagerly attracted by $C$. This immediately implies that $C$ is a~finite eager attractor, hence $C$ is good for $\emptyset$ by Proposition~\ref{prop-attractor}. We also immediately obtain that $C[R_1,R_2]$ is reducible to $\{C\}$, 
which means that  $C[R_1,R_2]$ is good for $\emptyset$ by Lemma~\ref{lem-reduce}.

Let $p\vec{v} \in C[R_1,R_2]$.
Since $t_S(2)<0$ and $\tau_{R_1}<1$, almost every run $w \in \run(p\vec{v})$  eventually visits 
a configuration of $C_S[c_2=0]$, and, from that moment on, visits configurations of both 
$C_S[c_2=0]$ and $C_S[c_1=0]$ infinitely often.

Denote by $\Theta_0(w)$ the least $\ell$ such that $w(\ell)\in C_S[c_2=0]$.
Given $k\geq 1$, denote by $\Theta_k(w)$ the least $\ell\geq \Theta_{k-1}(w)$ such that the following holds:
\begin{itemize}
	\item If $k$ is odd, then $w(\ell)\in C_S[c_1=0]$.
	\item If $k$ is even, then $w(\ell)\in C_S[c_2=0]$.
\end{itemize}
We use Theorems~\ref{thm-zeros},~\ref{thm:tail-bounds-height}, and~\ref{thm:tail-bounds-convergence} to show that there are computable constants $\hat{a}>0$ and $0<\hat{z}<1$ such that for all $k\geq 0$ and all $\ell\in \Nset$ we have that 
\[
\calP_{p\vec{v}}(\Theta_k-\Theta_{k-1}\geq \ell)\quad \leq \quad \hat{a}\cdot (\hat{z})^{\sqrt \ell}
\]
Here $\Theta_{-1}=0$.
Observe that $\Theta_0$ is the sum of the number of transitions needed to visit $Z_S$ for the first time (the first phase) and the number of transitions need to reach $C_S[c_2=0]$ subsequently (the second phase).
Due to Theorem~\ref{thm-zeros}, the probability that the first phase takes more than $\ell$ transitions is bounded by $a\cdot z^{\ell}$ for some computable $a>0$ and \mbox{$0<z<1$}. Note that the length of the second phase depends on the value of $c_2$ after the first phase. However, the probability that this value will be larger than $\ell$ can be bounded by $a\cdot z^{\ell}$ as well. Finally, assuming that the first phase ends in a configuration $q(0,m)$, Theorem~\ref{thm:tail-bounds-convergence} gives a bound $a'\cdot (z')^{\sqrt{\ell-m}}$ on the probability of reaching $C_S[c_2=0]$ in at least $\ell$ transitions. By combining these bounds appropriately, we obtain the above bound on $\Theta_0$.

Now let us consider $\Theta_{k}-\Theta_{k-1}$ for $k>0$. Let us assume that $k$ is even (the other case follows
similarly). The only difference from the previous consideration (for $\Theta_0$) is that now the first phase consists of the part of the run up to the $\Theta_{k-1}$-th configuration, and the second phase from there up to the $\Theta_k$-th configuration. Using Theorem~\ref{thm:tail-bounds-height} and induction hypothesis, we derive a bound $a\cdot z^{\ell}$ on the probability that the height of the second counter in the
$\Theta_{k-1}$-th configuration will be at least $\ell$. Then, as above, we combine this bound with the bound on the probability of reaching $C_S[c_2=0]$ in $\ell$ steps from a fixed configuration of $C_S[c_1=0]$.

In order to finish the proof, we observe that the probability of reaching a configuration of $C$ between
the $\Theta_{k-1}$-th and $\Theta_k$-th configuration is bounded away from zero by a computable constant.
This follows immediately from Theorem~\ref{thm:tail-bounds-height} which basically bounds the expected value of $c_2$ in the $\Theta_k$-th configuration.
Denoting by $\mathit{Rounds}(w)$ the least number $k$ such that $w(\Theta_k(w))\in C$, we may easily
show that $\calP_{p\vec{v}}(\mathit{Rounds}\geq \ell)\leq \bar{c}^{\ell}$ for a 
computable constant $0\leq \bar{c}<1$. 

Finally, we combine the bound on the number of rounds (i.e., the bound on  
$\calP_{p\vec{v}}(\mathit{Rounds}\geq \ell)$)  with the bound on the length of each round
(i.e., the bound on $\calP_{p\vec{v}}(\Theta_k-\Theta_{k-1}\geq \ell)$), and thus obtain the desired 
bound on the number of steps to visit $C$.
\end{proof}

For the other cases (incl.{} Cases~(b) and~(c)), we show that the set of configurations
$C$ we aim to analyze is eagerly attracted by computable semilinear sets of configurations
$C_1,\ldots,C_k$, where each $C_i$ is either good for~$\emptyset$ or good for
$\bigcup_{i\neq j}C_j$. In all these cases, it is easy to see that the configurations
of $C$ reach a configuration of $\bigcup_{i=1}^k C_i$ with probability one, and the
argument that $C$ is \emph{eagerly} attracted $\bigcup_{i=1}^k C_i$ is a simplified
version of the proof of Lemma~\ref{lem-II-II-gtrend-leftdown-trends-left-down}
(in some cases, the proof is substantially simpler than the one of Lemma~\ref{lem-II-II-gtrend-leftdown-trends-left-down}).
Therefore, in these cases we just list the sets $C_1,\ldots,C_k$ and add some intuitive
comments which explain possible behaviour of the runs initiated in configurations
of~$C$.

When defining the aforementioned sets $C_1,\ldots,C_k$, we use the following computable constants $B_{II}, B_{IV}, D_{II} \in \Nset$, which are numbers 
(not necessarily the least ones) satisfying the following conditions:
\begin{itemize}
	\item if $p(0) \in R$, where $R$ is a type~IV region of $\A_i$ for some $i \in \{1,2\}$,
	then $p(0)$ can reach a configuration a type~I region in at most $B_{IV}$ 
	transitions.
	\item if $p(0) \in R$, where $R$ is a type~II region of $\A_i$ such that $t_S(i) < 0$
	and $\tau_R > 0$, then there is a finite path $w$ from $p(0)$ to $p(0)$ of length
	smaller than $B_{II}$ such that
	the total $L_i$-reward of all transitions executed in $w$ is positive. 
	\item for every $p\vec{v} \in \conf(\A)$ and every $i \in \{1,2\}$, if $\vec{v}(i)= 0$, $\vec{v}(3{-}i) \geq D_{II}$, and $p(0) \in R$ for some type~II region of $\A_i$ such that either $t_S(i) > 0$ and $t_S(3{-}i) < 0$, or 
	$t_S(i) < 0$ and $\tau_R < 0$, then there exists $q\vec{u} \in \post(p\vec{v})$
	such that $\vec{u}(i) \geq \max\{B_{II},B_{IV}\}$ and $\vec{u}(3{-}i) = 0$.
\end{itemize}
The existence and computability of $B_{II}$, $B_{IV}$, and $D_{II}$ follows from simple observations
about the transition structure of $\M_\A$ (these constants are in fact small and
their size can be explicitly bounded in $\size{\A}$). 


\begin{lemma}
\label{lem-E-good}
	For all $m,n \in \Nset$ and a BSCC $S$ of $\C_\A$ such that $t_S$ is negative
	in some component, the set $E_s[m,n]$ is good for~$\emptyset$.
\end{lemma}
\begin{proof}
	From the definition of $E_S[m,n]$ and Theorem~\ref{thm-zeros}, we immediately
	obtain that $E_S[m,n]$ is a finite eager attractor (even if $E_S[m,n] = \emptyset$).
	Hence, the claim follows from Proposition~\ref{prop-attractor}. 
\end{proof}

%
%

Now we consider the remaining subcases of Case~(a).

\begin{lemma}
	\label{lem-II-II-gtrend-rightrown-trend-right}
	Let $C[R_1,R_2]$ be a set of configurations where 
	$R_1$ and $R_2$ are of type~II, $t_S(2) < 0$, $t_S(1) > 0$, 
	and~$\tau_{R_2} > 0$. Then $C[R_1,R_2]$ is good for~$\emptyset$.
\end{lemma}
\begin{proof}
	Let $\calE = \{E[B_{II},D_{II}], C_S[c_2 {=} 0 \wedge c_1 {\geq}  B_{II}] \}$.
	Observe that $E[B_{II},D_{II}]$ is good for $\emptyset$
	by Lemma~\ref{lem-E-good}. We show that $C_S[c_2 {=} 0 \wedge c_1 {\geq}  B_{II}]$ is good
	for $E[B_{II},D_{II}]$ and that $C[R_1,R_2]$ reducible to $\calE$. Hence, 
	$C[R_1,R_2]$ is good for~$\emptyset$ by Lemma~\ref{lem-reduce}.
	
	To see that $C_S[c_2 {=} 0 \wedge c_1 {\geq}  B_{II}]$ is good
	for $E[B_{II},D_{II}]$, realize that for every $p\vec{v} \in C_S[c_2 {=} 0 \wedge c_1 {\geq}  B_{II}]$
	we have that almost all runs of $\run(p\vec{v})$ that do not visit a configuration of 
	$E[B_{II},D_{II}]$ eventually behave as if the first counter did not exist,
	which means that the long-run behaviour of almost all of these runs is the same
	as the behavior of the runs of $\A_2$ initiated in $p(0)$ (here we also use the defining property
	of $D_{II}$). Further, it follows
	from the definition of $B_{II}$ and Theorem~\ref{thm:tail-bounds-divergence} that there
	exists a $\delta > 0$ such that \mbox{$\calP(\run(p\vec{v} \not\rightarrow^* E[B_{II},D_{II}])) > \delta$}
	for every \mbox{$p\vec{v} \in C_S[c_2 {=} 0 \wedge c_1 {\geq}  B_{II}]$}.
	
	By Theorem~\ref{thm:tail-bounds-divergence}, for every $\varepsilon > 0$ there exists a 
	computable semilinear constraint $\varphi$ such that 
	\mbox{$\sem{\varphi} \subseteq C_S[c_2 {=} 0 \wedge c_1 {\geq}  B_{II}]$} and for every $q\vec{u} \in \sem{\varphi}$ we have that
	the probability of visiting $C[c_1{=}B_{II}]$ (and hence also $E[B_{II},D_{II}]$) 
	is bounded by~$\varepsilon$.
	
	Now let $p\vec{v} \in C[R_1,R_2]$ and $\delta > 0$. We need to show that there
	is a computable $\ell \in \Nset$ such that the probability of reaching a configuration
	of $E[B_{II},D_{II}] \cup \sem{\varphi}$ in at most $\ell$ transitions is at least $1-\delta$.
	Since $t_S(2) < 0$, every
	$p\vec{v} \in C[R_1,R_2]$ is eagerly attracted by $Z_S$. Similarly as in
	the proof of Lemma~\ref{lem-II-II-gtrend-leftdown-trends-left-down}, we show
	that almost every run visits  $C_S[c_2 {=} 0]$ infinitely many times, and that the
	probability that the length between two consecutive visits to $C_S[c_2 {=} 0]$
	exceeds $\ell$ decays sub-exponentially in $\ell$. Further, the probability of
	vising a configuration of $E[B_{II},D_{II}] \cup \sem{\varphi}$ from a configuration
	of $C_S[c_2 {=} 0]$ is bounded away from zero by a fixed constant. Hence, we can
	argue as in the proof of Lemma~\ref{lem-II-II-gtrend-leftdown-trends-left-down}.
\end{proof}

\begin{lemma}
	\label{lem-II-II-gtrend-leftdown-trends-right-up}
	Let $C[R_1,R_2]$ be a set of configurations where 
	$R_1$ and $R_2$ are of type~II, $t_S(2) < 0$, $t_S(1) < 0$, $\tau_{R_1} > 0$, 
	and~$\tau_{R_2} > 0$. Then $C[R_1,R_2]$ is good for~$\emptyset$.
\end{lemma}
\begin{proof}
	Let $\calE$ be the set consiting of  $E[B_{II},B_{II}]$, $C_S[c_2 {=} 0 \wedge c_1 {\geq}  B_{II}]$,
	and  $C_S[c_1 {=} 0 \wedge c_2 {\geq}  B_{II}]$. 
	Clearly, each $C \in \calE$ is either good for~$\emptyset$ or good for 
	the union of all sets in $\calE \smallsetminus \{C\}$ (see Lemma~\ref{lem-E-good}
	and the proof of Lemma~\ref{lem-II-II-gtrend-rightrown-trend-right}). For every $\varepsilon >0$, there
	are computable semilinear constraint $\varphi_1,\varphi_2$ such that
	\mbox{$\sem{\varphi_1} \subseteq C_S[c_2 {=} 0 \wedge c_1 {\geq}  B_{II}]$},
	\mbox{$\sem{\varphi_2} \subseteq C_S[c_1 {=} 0 \wedge c_2 {\geq}  B_{II}]$}
	satisfying the requirements of Definition~\ref{def-reduce}. Note that there
	is a $k \in \Nset$ such that for every configuration of $Z_S$ there is a finite
	path of length at most $k$ to a configuration of $E[B_{II},B_{II}] \cup \sem{\varphi_1}
	\cup \sem{\varphi_2}$. The rest of the argument is even simpler than in 
	Lemma~\ref{lem-II-II-gtrend-rightrown-trend-right}.   
\end{proof}

\begin{lemma}
	\label{lem-II-II-gtrend-leftdown-trends-right-down}
	Let $C[R_1,R_2]$ be a set of configurations where 
	$R_1$ and $R_2$ are of type~II, $t_S(2) < 0$, $t_S(1) < 0$, $\tau_{R_1} < 0$, 
	and~$\tau_{R_2} > 0$. Then $C[R_1,R_2]$ is good for~$\emptyset$.
\end{lemma}
\begin{proof}
	Let $\calE = \{E[B_{II},D_{II}], C_S[c_2 {=} 0 \wedge c_1 {\geq}  B_{II}]\}$.
	We show that $C[R_1,R_2]$ reducible to $\calE$ similarly as in 
	Lemma~\ref{lem-II-II-gtrend-rightrown-trend-right}.
\end{proof}

The case when $R_1$ and $R_2$ are of type~II, $t_S(2) < 0$, \mbox{$t_S(1) < 0$}, $\tau_{R_1} > 0$, 
and~$\tau_{R_2} < 0$ is symmetric to the case considered in Lemma~\ref{lem-II-II-gtrend-leftdown-trends-right-down}.

Now we continue with Case~(b)

\begin{lemma}
	\label{lem-IV-II-trend-right}
	Let $C[R_1,R_2]$ be a set of configurations where 
	$R_1$ is of type~IV and $R_2$ is of type~II such that $t_S(2) < 0$ and $\tau_{R_2} > 0$.
	Then $C[R_1,R_2]$ is good for~$\emptyset$.
\end{lemma}
\begin{proof}
	Let $\calE$ be the set consisting of $E[B_{II},B_{IV}]$, $C_S[c_2 {=} 0 \wedge c_1 {\geq}  B_{II}]$, 
    and all $C[R'_1,R_2]$, where $R'_1$ is a type~I region reachable from $R_1$ in $\A_1$.
	We show that $C[R_1,R_2]$ reducible to $\calE$ similarly as in previous lemmata.
\end{proof}	

\begin{lemma}
	\label{lem-IV-II-trend-left}
	Let $C[R_1,R_2]$ be a set of configurations where 
	$R_1$ is of type~IV and $R_2$ is of type~II such that $t_S(2) < 0$ and $\tau_{R_2} < 0$.
	Then $C[R_1,R_2]$ is good for~$\emptyset$.
\end{lemma}
\begin{proof}
	Let $\calE$ be the set consisting of $E[D_{II},B_{IV}]$ and
	all $C[R'_1,R_2]$, where $R'_1$ is a type~I region reachable from $R_1$ in $\A_1$.
	Then $C[R_1,R_2]$ reducible to $\calE$ and each $C \in \calE$ is good for~$\emptyset$.
\end{proof}	

\begin{lemma}
	\label{lem-II-IV-gtrend-rightdown}
	Let $C[R_1,R_2]$ be a set of configurations where 
	$R_1$ is of type~II and $R_2$ is of type~IV such that $t_S(2) < 0$ and $t_S(1) > 0$.
	Then $C[R_1,R_2]$ is good for~$\emptyset$.
\end{lemma}
\begin{proof}
	Let $\calE$ be the set consisting~of $E[B_{IV},D_{II}]$ and all $C[R_1,R'_2]$, where $R'_2$ is a type~I region reachable from $R_2$ in $\A_2$.
	Then $C[R_1,R_2]$ reducible to $\calE$. Further, all elements of $\calE$ are good for $\emptyset$. 
\end{proof}

Note that the case when $R_1$ is of type~II and $R_2$ is of type~IV such that 
$t_S(2) < 0$ and $t_S(1) < 0$ is symmetric to the cases covered in Lemma~\ref{lem-IV-II-trend-right} and Lemma~\ref{lem-IV-II-trend-left}.

Finally, in the next lemma we consider Case~(c).

\begin{lemma}
\label{lem-typeIV}
	Let $C[R_1,R_2]$ be a set of configurations where both 
	$R_1$ and $R_2$ are type~IV regions, and the trend $t_S$ of the associated BSCC $S$ of 
	$\C_\A$ is negative in some component. Then $C[R_1,R_2]$ is good for~$\emptyset$. 
\end{lemma}
\begin{proof}
    Let $\calE$ be the set consisting of $E[B_{IV},B_{IV}]$ and
    all $C[R_1,R'_2]$, $C[R'_1,R_2]$, $C[R'_1,R'_2]$, where $R'_i$ is a type~I
    region reachable from $R_i$ in $\A_i$ (for $i \in \{1,2\}$).
    We show that $C[R_1,R_2]$ reducible to $\calE$.
\end{proof}

%% file: app-three-counters.tex
\section{Some notes on three-counter pVASS}
\label{sec-three-counters}


In this section we give an example of a $3$-dimensional pVASS $\A$
such that $\M_\A$ is strongly connected, and the pattern frequency
vector seems to take the $\perp$ value with probability one
(this intuition is confirmed by Monte Carlo simulations,
see below). Further, the example is insensitive to small changes in rule weights,
and it also shows that the method of Section~\ref{sec-two-counters}
based on constructing pVASS of smaller dimension by ``forgetting''
one of the counters and then studying the ``trend'' of this counter
in the smaller pVASS is insufficient for three (or more) counters.

\newcommand{\wsmall}{P}
\newcommand{\wmedium}{Q}
\newcommand{\wlarge}{R}

\begin{figure}
\def\ss{3.8}
\begin{center}
\begin{tikzpicture}[scale=.6,font=\scriptsize,>=stealth']
every node/.style={draw,minimum size=4mm,}
\node[place] (s) at (0,0) {$p$};
\node[place] (t1) at (-5,-3) {$t_1$};
\node[place] (t2) at (0,-3) {$t_2$};
\node[place] (t3) at (5,-3) {$t_3$};

\draw[->] (s) to node[left] {(-1,-1,0);\wmedium} (t1);
\draw[->,bend right] (s) to node[near end,left] {(0,-1,-1);\wmedium} (t2);
\draw[->] (s) to node[right] {(-1,0,-1);\wmedium} (t3);

\draw[->,rounded corners=8pt] (t1) -- node[left]
     {(0,3,0);1}(-5,0) node {} -- (s);
\draw[->,bend right] (t2) to node[near start,right] {(0,0,3);1} (s);
\draw[->,rounded corners=8pt] (t3) -- node[right]
     {(3,0,0);1}(5,0) node {} -- (s);

\tikzset{every loop/.style={min distance=10mm,in=120,out=60,looseness=10}}
\path[->] (s) edge  [loop above] node
{(2,0,0);\wsmall,\quad (0,2,0);\wsmall,\quad (0,0,2);\wsmall,\quad
 (-1,-1,-1);\wlarge} ();

\end{tikzpicture}

\end{center}
\caption{A $3$-dimensional pVASS $\A$. For suitable weights $P,Q,R>0$,
we have that $F_{\A} = {\perp}$ almost surely.}
\label{fig-three-counters}
\end{figure}
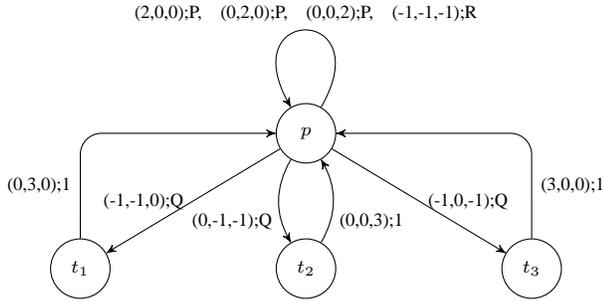

The pVASS $\A$ is shown in Fig.~\ref{fig-three-counters}. Some rules
increase the counter
by more that $1$, so these should be formally replaced by several rules using
auxiliary control states. Intuitively, $\A$ behaves
in the following way. Suppose we start in an initial
configuration $p(m,0,0)$, where $m$ is ``large''. Then, $\A$ starts
to decrease the first counter and increase the second one. On average,
the value of the second counter becomes $2m$ when the first
counter is decreased to zero, and the third counter is kept ``small''.
So, ``on average'' we eventually reach a configuration
$p(0,2m,0)$ in about $2m$ transitions. Then, the second counter is
decreased and the third counter is increased, where the value is
again doubled ``on average'', using $4m$ transitions.
Thus, we reach a configuration $p(0,0,4m)$. Then, we ``pump'' the tokens
from the third counter to the first one, reaching $p(8m,0,0)$ in about
$8m$ transitions. And so on. Observe that the $k$-th phase takes
about $2^k$ transitions, and so at the end of each phase,
about half of the time was spent in configurations with the
``current'' pattern. Hence, the pattern frequency oscillates.

A precise formulation of this phenomenon, and a formal proof that
almost all runs really behave in the above indicated way,
are technically demanding and we do not provide them in this paper.
For the reader's convenience, we have implemented a simple Maple sheet
which can be used to perform Monte Carlo simulations of $\A$ and
observe the above described phenomenon in practice\footnote{Available at \mbox{http://www.cs.ox.ac.uk/people/stefan.kiefer/pVASS-simulation.txt}}.

Note that the oscillation of $\A$ is insensitive to small changes
in rule weights. However, if we modify $\A$ into $\A'$ so that the counter
value is \emph{decreased}
on average in each phase (e.g., we start in $p(m,0,0)$, and then reach
$p(0,m-1,0)$, $p(0,0,m-2)$, $p(m-3,0,0)$, etc., on average), then the \emph{sum}
of the counters has a tendency to decrease and $\M_{\A'}$ has a finite
attractor. This means that the pattern frequency vector is well
defined for almost all runs of $\A'$. Still, the behaviour of
all two-counter machines $\B_1$, $\B_2$, $\B_3$ obtained from
$\A$ by ``forgetting'' the first, the
second, and the third counter, is essentially similar
to the behaviour of $\B_1'$, $\B_2'$, and $\B_3'$ obtained from $\A'$
in the same way (for example, both in $\B_1$ and $\B_1'$, the second
counter has a tendency to increase and the third has a tendency
to decrease). Hence, we cannot distinguish between the behaviour
of $\A$ and $\A'$ just by studying the ``trends'' in the two-counter
pVASS obtained by ``forgetting'' one of the counters. This indicates
that the study of $3$-dimensional pVASS requires different (and perhaps
more advanced) methods than those presented in this paper.


%% file: app-one-counter.tex
\section{Proofs of Section~\ref{sec-one-counter}}
\label{app-one-counter}


\begin{lemma}\label{lem-typeI-reach}
	Let $R \neq \emptyset$ be a type~II region. Then every configuration of $\pre(R)$ can reach a configuration of $R$ in at most $4|Q|^3$ transitions.
\end{lemma}  
\begin{proof}
	Let $R = \post(p(0))$ be a type~II region. Recall that $R$ is strongly connected.
	For every $i \geq 0$, let
	$L_i = \{q \in Q \mid q(i) \in \post(p(0))\}$. We start by showing that there
	is $\tau \leq |Q|$ such that $L_i \subseteq L_{i+\tau}$ for every $i \in \Nset$,
	and if $i \geq |Q|^2$, then $L_i = L_{i+\tau}$. Hence, the structure of 
	$R$ is ``ultimately periodic'' and the period $\tau$ is small. 
	
	Let $\tau$ be the least $j > 0$ such that $p \in L_j$. We claim that \mbox{$\tau \leq |Q|$}.
	Since $R$ is infinite, all $L_i$ are non-empty, and hence there are $0 \leq i < j \leq |Q|$ such that
	$L_i \cap L_j \neq \emptyset$. Let $r \in L_i \cap L_j$. Then $r(i) \tran{}^* p(0)$, hence also $r(j) \tran{}^* p(j{-}i)$, which means
	$p \in L_{j{-}i}$. Now we show that $L_i \subseteq L_{i+\tau}$ for every $i \in \Nset$; to 
	see this, first realize that $p(0) \tran{}^* p(\tau)$. If $r \in L_i$, then $p(0) \tran{}^* r(i)$, and hence also $p(\tau) \tran{}^*r(\tau{+}i)$. This means that $r(\tau{+}i)$ is reachable from $p(0)$, i.e., $r \in L_{i{+}\tau}$. It remains to prove that if $i \geq |Q|^2$, then $L_i = L_{i+\tau}$. Clearly, there is \mbox{$k \leq |Q|$} such that $L_{k\cdot\tau} = L_{(k{+}1)\cdot\tau}$. Since $k\cdot\tau \leq |Q|^2$,
	it suffices show that for every $i \geq k\cdot\tau$ we have that  
	\mbox{$L_{i} \supseteq  L_{i+\tau}$}. Let $s \in L_{i{+}\tau}$, and let
	$r \in L_{(k{+}1)\cdot\tau}$. Then  $r((k{+}1)\cdot\tau) \tran{}^* s(i{+}\tau)$, and this finite path
	inevitably contains a suffix which is a finite path from $t((k{+}1)\cdot\tau)$ to $s(i{+}\tau)$ such that
	$t \in L_{(k{+}1)\cdot\tau}$ and the counter is never decreased below $(k{+}1)\cdot\tau$ 
	along this suffix. Hence, there is also a finite path from $t(k\cdot\tau)$ to $s(i)$, and since $t \in L_{k\cdot\tau}$, we obtain $s \in L_i$.
	
	Now let $q(k) \in \pre(R)$, and let $w$ be a path of minimal length from $q(k)$ to a
	configuration of $R$. Suppose that the last configuration of $w$ is $r(m) \in R$.
	
	First we show that $w$ cannot contain a subpath of the form $t(i) \tran{}^* s(j)$ where
	$i-j \geq |Q|^2$. Suppose that $w$ contains such a subpath. Then $w$ also contains a subpath of the form $v(n{+}\ell\tau) \tran{}^* v(n)$, where $\ell \geq 1$, which can be safely removed from $w$ and the suffix of $w$ after the configuration $v(n)$ can be trivially adjusted so that the resulting path $w'$ leads from $q(k)$ to $r(m{+}\ell\tau)$. Since $r(m{+}\ell\tau) \in R$ (see above), we obtain a contradiction with our choice of~$w$.
	
	Further, we prove that the counter stays bounded by $k + 3|Q|^2$ in every configuration visited by~$w$. Suppose the converse. Then $w$ contains a subpath of the form
	$t(k{+}2|Q|^2) \tran{}^* s(k{+}3|Q|^2)$. By applying the observation of the previous
	paragraph, we obtain that the counter stays above $k+|Q|^2$ in all configurations
	visited by $w$ after $t(k{+}2|Q|^2)$, and above $k+2|Q|^2$ in all configurations
	visited by $w$ after $s(k{+}3|Q|^2)$. In particular, the last configuration $r(m)$
	of $w$ satisfies $m \geq k+2|Q|^2$. Further, the subpath  $t(k{+}2|Q|^2) \tran{}^* s(k{+}3|Q|^2)$ must contain a subpath of the form $v(n) \tran{}^* v(n{+}\ell\tau)$ where $1 \leq \ell \leq |Q|$.
	If we delete this subpath from $w$ and adjust the configurations visited after 
	$v(n{+}\ell\tau)$, we obtain a path $w'$ from $q(k)$ to $r(m{-}\ell\tau)$. Since 
	$r(m)\in R$, $m \geq k+2|Q|^2$, and $\ell\tau \leq |Q|^2$, we obtain that $r(m{-}\ell\tau) \in R$ (see above).
	Thus, we obtain a contradiction with our choice of~$w$.
	
	To sum up, $w$ can visit at most $4|Q|^3$ different configurations, and hence
	its length cannot exceed $4|Q|^3$.
\end{proof}

\begin{reftheorem}{Lemma}{\ref{lem-bounded}}
	Every configuration of $\A$ can reach a configuration of some region in at most $11|Q|^4$ transitions.
\end{reftheorem}
\begin{proof}
	We start with three auxiliary observations. Let $S$ be a BSCC of $\C_\A$, and let $R_{I}(S)$ and $R_{II}(S)$ be the unions of all type~I regions and all type~II regions determined by all $q \in S$, respectively.
	Further, let $R_{III}(S)$ be the type~III region determined by~$S$. We have the following:
	\begin{itemize}
		\item[(a)]  If $R_{III}(S) \neq \emptyset$,
		then for all $q \in S$ and all $\ell \geq |Q|$ we have that $q(\ell) \in R_{III}(S)$.
		\item[(b)] Let $D_S$  be the set
		\[ 
		\bigg(S {\times} \Nset \ \cap\   \pre(R_{I}(S))\bigg) \smallsetminus 
		\bigg(R_{I}(S) \cup \pre(R_{II}(S))\bigg).
		\]
		If $D_S$ contains a configuration $q(\ell)$ where $\ell \geq 4|Q|^3 + |Q|^2$,
		then $D_S$ is infinite (i.e., the type~IV region determined by $S$ is exactly $D_S$).
		\item[(c)] If $q(\ell) \in \pre(R_{I}(S))$, then $q(\ell)$ can reach a configuration of $R_{I}(S)$
		in at most $\ell|Q| + 4|Q|^3$ transitions.
	\end{itemize}
	Observation~(a) follows by observing that if $q(\ell) \tran{}^* s(0)$ where $r,s \in S$ and $\ell \geq |Q|$,
	than the path from $q(\ell)$ to $s(0)$ contains a subpath of the form $t(i) \tran{}^*t(j)$, where
	$i > j$. This means that \emph{every} configuration of $S \times \Nset$ can reach a configuration with
	zero counter, because $S$ is strongly connected.
	
	To prove Observation~(b), assume that $D_S$ contains a configuration $q(\ell)$ where 
	$\ell \geq 4|Q|^3 + |Q|^2$ and $D_S$ is finite. Then there is the largest $\ell'$ such that
	$q(\ell') \in D_S$. Obviously, $\ell' \geq 4|Q|^3 + |Q|^2$. We show that 
	$q(\ell') \in \pre(R_{II}(S))$, which is a contradiction. Recall that every non-empty type~II region
	determined by a control state of $S$ is ultimately periodic and its period $\tau$ is bounded by $|Q|$
	(see the proof of Lemma~\ref{lem-typeI-reach}). Let $\kappa$ be the product of the periods of all non-empty
	type~II regions determined by the control states of~$S$. Then 
	$q(\ell' {+} \kappa) \in \pre(R_{II}(S))$ (otherwise, we have a contradiction with the maximality 
	of $\ell'$). Hence, $q(\ell' {+} \kappa)$ can reach a configuration $v(m)$ of some type~II region
	in at most $4|Q|^3$ transitions, which means that $m \geq \kappa + |Q|^2$, and the 
	configuration $q(\ell')$ can reach the configuration $v(m{-}\kappa)$. By our choice of $\kappa$ 
	and the fact that $m{-}\kappa \geq |Q|^2$, we obtain that $v(m{-}\kappa)$ belongs to the same type~II
	region as $v(m)$ (see the proof of Lemma~\ref{lem-typeI-reach}). 
	
	Observation~(c) is obtained in two steps. We show that 
	\begin{itemize}
		\item[(A)] $q(\ell)$ can reach (some) configuration with zero counter in at most $\ell|Q|+|Q|^2$ transitions;
		\item[(B)] if $s(0) \tran{}^* t(0)$ where $s,t \in S$, then $s(0)$ can reach $t(0)$ in at most $|Q|^3 + |Q|$ transitions.
	\end{itemize}
	Note that Observation~(c) follows immediately from~(A) and~(B). To prove~(A), we distinguish two possibilities.
	If there is a decreasing cycle, i.e., a path of length less than $|Q|$ of the form $t(i) \tran{}^* t(j)$ where $j < i$, then $q(\ell)$ needs at most $|Q|-1$ transitions to reach a configuration with zero counter or a configuration $t(\ell+c)$ where $c < |Q|$. In the second case, at most 
    $(\ell+c)(|Q|-1) \leq \ell|Q| + |Q|^2 - |Q|$ transitions are needed to reach a configuration with zero
    counter from $t(\ell+c)$, and hence $q(\ell)$ can reach a configuration with zero counter in at most $\ell|Q|+|Q|^2$ transitions. If there is no decreasing cycle, then $q(\ell)$ can still reach a configuration with zero counter (because $q(\ell) \in \pre(R_{I}(S))$), and hence there is a path $w$ of minimal length
    from $q(\ell)$ to a configuration with zero counter. It follows easily that if the length of this path
    exceeds $\ell|Q|$, then $w$ either contains a decreasing cycle or can be shortened.
	
    To prove part~(B), consider a path $w$ of minimal length from $s(0)$ to $t(0)$. One can easily show
    that the counter value must be bounded by $|Q|^2$ along $w$, because $w$ could be shortened otherwise.
    Hence, $w$ can visit at most $|Q|^3 + |Q|$ configuration, which means that its length is bounded
    by $|Q|^3 + |Q|$.
	
	Now we can finish the proof of Lemma~\ref{lem-bounded}.
	Let $p(k)$ be a configuration of $\A$. If $p(k) \in \pre(R)$ for some type~II region~$R$,
	then $p(k)$ can reach $R$ in at most $4|Q|^3$ transitions by Lemma~\ref{lem-typeI-reach}. 
	Otherwise, let us first consider the case when $p(k) \tran{}^* t(i)$ for some configuration
	$t(i)$ where $i \geq |Q|$. Then such a $t(i)$ is reachable from $p(k)$ in at most $|Q|^2$ transitions. 
	Further, $t(i)$ can reach a configuration $r(j)$, where $r \in S$
	for some BSCC $S$ of $\C_\A$, in at most $|Q|$ transitions. If $r(j) \in R_{I}(S)$, we are done.
	Otherwise, if $r(j)$ can reach a configuration of
	$R_{III}(S)$, then such a configuration is reachable from $r(j)$ in at most $|Q|^2$ transitions (here we
	use~Observation~(a)). Otherwise, $r(j) \in D_S$. If $j \geq 4|Q|^3 + |Q|^2$, then $D_S$ is the type~IV region determined by $S$ by Observation~(b) and we are done. If $j < 4|Q|^3 + |Q|^2$, then $r(j)$ can reach a
	configuration of $R_{I}(S)$ in at most $9|Q|^4$ transitions by Observation~(c). Hence, $p(k)$ can reach a configuration of some region in at most $11|Q|^4$ transitions. 
	
	It remains to consider the case when  $p(k)$ cannot reach a configuration $t(i)$ such that
	$i \geq |Q|$. Then the total number of configurations reachable from $p(k)$ is bounded by $|Q|^2$.
	Each of these configurations is reachable in at most $|Q|^2$ transitions,
	and some of them must belong to a type~I or a type~III region. 
\end{proof}

\bigskip

Before proving the next lemmata, we need to introduce some notation. For all configurations
$p(k)$ and $q(\ell)$, we use $\run(p(k) \tran{}^* q(\ell))$ to denote the set of all
$w \in \run(p(k))$ that visit $q(\ell)$. We also use $\run(p(k),{\uparrow})$ to denote the
set of all $w \in \run(p(k))$ such that the counter stays positive in some suffix of $w$,
and $\run(p(k),{\uparrow}_S)$, where $S$ is a BSCC of $\C_\A$, to denote those $w \in \run(p(k),{\uparrow})$ which visit a configuration with control state in~$S$. For all
$p,q \in Q$, we use $[p{\downarrow}q]$ to denote the probability of all $w \in \run(p(1))$
that visit $q(0)$ and the counter stays positive in all configurations preceding this 
visit. Finally, we use $[p{\uparrow}]$ to denote $1 - \sum_{q \in Q} [p{\downarrow}q]$.


\begin{reftheorem}{Lemma}{\ref{lem-reg-same-value}}
	Let $p(k)$ be a configuration of $\A$ and $Z$~a zone of~$\A$. Then $F_\A$ is well defined 
	for almost all $w \in \run(p(k),Z)$, and there exists $F : \Pat_\A \rightarrow \Rset$
	such that $F_\A(w) = F$ for almost all $w \in \run(p(k),Z)$. Further, for every
	rational $\varepsilon >0$, there is a vector $H : \Pat_\A \rightarrow \Qset$ computable in time polynomial in
	$\size{\A}$ and~$\size{\varepsilon}$ such that $H(q\alpha)$ approximates $F(q\alpha)$ up to
	the relative error~$\varepsilon$ for every $q\alpha \in \Pat_\A$. 
\end{reftheorem}
\begin{proof}
	For every BSCC $S$ of $\C_\A$, we define a vector $F_S : \Pat_\A \rightarrow \Qset$
	as follows: $F_S(q(0)) = 0$ for all $q \in Q$, $F_S(q(*)) =0$ for all $q \in Q \smallsetminus S$, and $F_S(q(*)) = \mu_S(q)$ for all $q \in S$ (recall that 
	$\mu_S$ is the invariant distribution of~$S$). Note that $F_S$ is a rational vector
	that can be computed in time polynomial in~$\size{\A}$. 
	
	Let $Z$ be a zone of~$\A$. If $Z = \emptyset$, then the claim follows trivially (according to
	the definitions adopted in Section~\ref{sec-prelim}, $F$ can be chosen arbitrarily, and we can
	put $H = F$). Now let $Z$ be a non-empty zone. We proceed by considering possible forms of~$Z$.
	
	Let us first assume that $Z = R$, where $R$ is a type~I region. Then $R$ can be seen as a strongly connected Markov chain with at most $|Q|^2$ vertices (see the remarks before Lemma~\ref{lem-typeI-reach}), and the corresponding invariant distribution $\mu_R$ is computable in time polynomial in~$\size{\A}$. Hence, for almost all $w \in \run(p(k),R)$ we have that $F_\A(w) = F_R$, where $F_R(q(0)) = \mu_R(q(0))$ for all
	$q(0) \in R$, $F_R(q(0)) = 0$ for all $q(0) \not\in R$, and 
	$F_R(q(*)) = \sum_{q(k) \in R, k>0} \mu_R(q(k))$ for all $q \in Q$ (the empty sum is 
	equal to $0$). 
	
	If $Z = R$ where $R$ is a type~III region determined by a BSCC $S$ of $\C_\A$ such that
	$t_S \leq 0$, then the actual counter value does not influence the limit behaviour of runs staying in~$R$, which
	means that $F_\A(w) = F_S$ for almost all $w \in \run(p(k),R)$. 
	
	If $Z = R_{II}(S) \cup R_{III}(S) \cup R_{IV}(S)$, where $S$ is a BSCC of $\C_\A$ such that $t_S > 0$,
	then almost all runs of $\run(p(k),Z))$ are \emph{diverging}, i.e., for every
	$\ell \in \Nset$ and almost every $w \in \run(p(k),Z))$ there exists $m \in \Nset$
	such that the counter value is at least~$\ell$ in every configuration $w(m')$ where $m' \geq m$. Consequently, $F_\A(w) = F_S$ for almost all $w \in \run(p(k),Z)$.
	
	If $Z = R_{II}(S)$, where $S$ is a BSCC of $\C_\A$ such that $t_S = 0$, then the configurations of $Z$ 
	with zero counter are visited infinitely often by almost all runs of $\run(p(k),Z)$, but 
	the expected number of transitions between two consecutive visits to such configurations is infinite (see \cite{BKK:pOC-time-LTL-martingale-JACM}). In other words, visits to
	configurations with zero counter have zero frequency for almost all runs of $\run(p(k),Z)$.
	Consequently, $F_\A(w) = F_S$ for almost all $w \in \run(p(k),Z)$.
	
	Finally, consider the case when $Z = R$, where $R$ is a type~II region determined by $p \in S$ where $S$ is a BSCC of $\C_\A$ satisfying $t_S < 0$. Let $\D_S$ be a finite-state Markov chain where
	the set of vertices is \mbox{$\{q_j \mid j \in \{0,1\}, q(j) \in R\}$}
	and the transitions are defined as follows:
	\begin{itemize}
		\item $q_0 \tran{x} r_j$ in $\D_S$ iff $q(0) \tran{x} r(j)$ in $\M_\A$ (where $j \in \{0,1\}$);
		\item $q_1 \tran{x} r_0$ in $\D_S$ iff $x = [q{\downarrow}r] > 0$.
	\end{itemize}
	Note that the sum of the probabilities of all outgoing transitions
	of every vertex $r_1$ of $\D_S$ is equal to one, because almost all runs of
	$\run_{\M_{\A}}(q(1))$ visit a configuration with zero counter (see \cite{BKK:pOC-time-LTL-martingale-JACM}). Also note that $\D_S$ is strongly connected.
	For every transition of the form $q_1 \tran{} r_0$ in $\D_S$, we define 
	the following conditional expectations: 
	\begin{itemize}
		\item $E[L \mid q_1 \tran{} r_0]$, the \emph{conditional expected length} of a path from $q(1)$ to $r(0)$ in $\M_\A$, under the condition
		that $q(1)$ reaches $r(0)$ via a path where the counter stays positive
		in all configurations except for the last one;
		\item $E[\#_s \mid q_1 \tran{} r_0]$, the \emph{conditional expected number of visits to configurations with control state \mbox{$s \in S$}}
		along a path from $q(1)$ to $r(0)$ in $\M_\A$ (where the visit to $r(0)$
		does not count), under the condition
		that $q(1)$ reaches $r(0)$ via a path where the counter stays positive
		in all configurations except for the last one.
	\end{itemize}
	Let $\mu_{\D_S}$ be the invariant distribution of $\D_S$. Then
	\[
	E[L] = \sum_{q_0 \mathrm{ in } \D_S} \mu_{\D_S}(q_0) + 
	\sum_{q_1 \tran{x} r_0 \mathrm{ in } \D_S} \mu_{\D_S}(q_1) \cdot x \cdot E[L \mid q_1 \tran{} r_0]
	\] 
	is the average number of transitions between two consecutive visits to configurations
	$q(i)$, $r(j)$ of $R$, where $i+j \leq 1$, in a run initiated in a configuration 
	of~$R$. Similarly,
	\[
	E[s] =  
	\sum_{q_1 \tran{x} r_0 \mathrm{ in } \D_S} \mu_{\D_S}(q_1) \cdot x \cdot E[\#_s \mid q_1 \tran{} r_0]
	\]     
	is the average number of visits to a configuration with control state $s$ between two consecutive visits to configurations $q(i)$, $r(j)$ of $R$, where $i+j \leq 1$, in a run initiated in a configuration of~$R$ (the visit to $r(j)$ does not count).
	Now we define a vector $F_{\D_S} : \Pat_\A \rightarrow \Rset$
	as follows: 
	\begin{itemize}
		\item $F_{\D_S}(q(0)) = 0$ for all $q(0) \not\in R$,
		\item $F_{\D_S}(q(*)) = 0$ for all $q(*)$ such that $R$ does not contain any configuration matching $q(*)$,
		\item $F_{\D_S}(q(0)) = \mu_{\D_S}(q_0)/E[L]$ for all $q(0) \in R$,
		\item $F_{\D_S}(q(*)) = E[q]/E[L]$ for all $q(*)$ such that $R$ contains a configuration matching $q(*)$.
	\end{itemize}
	By applying strong ergodic theorem (see, e.g., \cite{Norris:book}), we obtain that 
	$F_\A(w) = F_{\D_S}$ for almost all $w \in \run(p(k),R)$. 
	
	Since the transition probabilities of $\D_S$ may take irrational values, 
	the numbers involved in the definition of $F_{\D_S}$ cannot be computed
	precisely. By Theorem~3.2~(B.b.1) of \cite{BKK:pOC-time-LTL-martingale-JACM}, we
	have that both $E[L \mid q_1 \tran{} r_0]$ and $E[\#_s \mid q_1 \tran{} r_0]$ 
	are bounded by $\alpha = 85000|Q|^6/(\xmin^{5|Q|+|Q|^3}\cdot t_S^4)$, where 
	$\xmin$ is the least transition probability of $\M_\A$. Note that $\size{\alpha}$ is
	polynomial in~$\size{\A}$. Using this bound, a simple error propagation
	analysis reveals that if the transition probabilities of $\D_S$, all components of the invariant distribution $\mu_{\D_S}$, and the conditional expectations 
	$E[L \mid q_1 \tran{} r_0]$,
	$E[\#_s \mid q_1 \tran{} r_0]$ are computed up to a relative error 
	$\varepsilon/(42\cdot|Q|^2 \cdot \alpha)$, then the relative error of every
	component in the approximated $F_{\D_S}$ is bounded by~$\varepsilon$. By \cite{SEY:pOC-poly-Turing},
	the transition probabilities of $\D_S$ can be approximated up to an arbitrarily
	small positive relative error in polynomial time. The values of the conditional
	expectations can be efficiently approximated by applying the results of 
	\cite{BKK:pOC-time-LTL-martingale-JACM} (in \cite{BKK:pOC-time-LTL-martingale-JACM}, the results are formulated just for $E[L \mid q_1 \tran{} r_0]$, but their extension to $E[\#_s \mid q_1 \tran{} r_0]$ is trivial). The invariant distribution
	$\mu_{\D_S}$ can be efficiently approximated by, e.g., applying the result
	of \cite{CM:MC-stationary-mean-passage-LAA} (see also \cite{CM:MC-stationary-overview-LAA} for a more comprehensive overview) which says that if the transition matrix $M_{\D_S}$ of $\D_S$ is
	approximated by $M_{\D_S}'$, then $\norm{\mu_{\D_S} - \mu_{\D_S}'} \leq \varrho \cdot 
	\norm{M_{\D_S} - M_{\D_S}'}$, where $\mu_{\D_S}'$ is the invariant distribution of $M_{\D_S}'$ and
	\[
	\varrho = \frac{1}{2} \cdot \max_j \left\{ \frac{\max_{i \neq j} m_{ij}}{m_{jj}} \right\} 
	\]  
	Here $m_{ij}$, $i \neq j$, is the mean first passage time from state~$i$ to
	state~$j$, and $m_{jj}$ is the mean return time to state~$j$, where all of these
	values are considered for $M_{\D_S}$. Since the least
	transition probability of $\D_S$ is at least $\xmin^{|Q|^3}$ (see 
	\cite{EWY:one-counter-PE}) and $\D_S$ has at most $2|Q|$ states, we have 
	that $m_{ij}$ and $m_{ii}$ are bounded by $2|Q|/\xmin^{2|Q|^4}$.
	This means that every component of $\mu_{\D_S}$ is bounded 
	by $\xmin^{2|Q|^4}/2|Q|$ from below, and $\varrho$ is bounded by 
	$|Q|/\xmin^{2|Q|^4}$ from above. Hence, it suffices to approximate
	the transition probabilities of $\D_S$ up to the absolute error
	$\varepsilon \xmin^{4|Q|^4} / (168|Q|^5 \alpha)$ and compute the invariant distribution for the approximated transition matrix $M_{\D_S}'$.  
\end{proof}

%

\begin{reftheorem}{Lemma}{\ref{lem-one-counter-approx}}	
	\label{lem-one-counter-approx}	
	Let $p(k)$ be a configuration of $\A$. Then almost every run initiated in $p(k)$ eventually stays
	in precisely one zone of~$\A$. Further, for every zone $Z$ and every rational $\varepsilon >0$, there 
	is a $P \in \Qset$ computable in time polynomial in $\size{\A}$, $\size{\varepsilon}$, and $k$ such that
	$P$ approximates $\calP(\run(p(k),Z))$ up to the relative error $\varepsilon$. 
\end{reftheorem}

\begin{proof}
	First, observe that every region is a part of some zone, except for non-empty type~IV
	regions determined by BSCCs of $\C_\A$ with negative trend. If $R$ is such a type~IV region
	and $S$ the associated BSCC where $t_S < 0$, then almost all runs initiated in a configuration
	of $R$ visit a configuration with zero counter infinitely often. Consequently, almost
	all runs initiated in a configuration of $R$ visit a type~I region, which means that
	$\calP(p(k),R) = 0$. Thus, by applying Lemma~\ref{lem-stay-in-one}, we obtain that almost every 
	run initiated in $p(k)$ eventually stays in precisely one zone of~$\A$.

    Let $Z$ be a zone of $\A$. We proceed by considering possible forms of~$Z$.
    \begin{itemize}
	   \item Let $Z = R$, where $R$ is a type~I region determined by a control state~$q$. Note that 
	      the problem whether $R = \emptyset$ is decidable in time polynomial in $\size{\A}$.
	      If $R = \emptyset$, then $\calP(\run(p(k),Z)) = 0$ and we are done. Otherwise,
	      $\calP(\run(p(k),Z)) =  \calP(\run(p(k)\tran{}^* q(0)))$.
	   \item Let $Z = R$ where $R$ is a type~III region determined by a BSCC $S$ of $\C_\A$ such that
	      $t_S \leq 0$. Then $\calP(\run(p(k),Z)) =  \calP(\run(p(k),{\uparrow}_S))$ (note that the
	      equality holds even if $R = \emptyset$).
	   \item Let $Z = R_{II}(S) \cup R_{III}(S) \cup R_{IV}(S)$, where $S$ is a BSCC of $\C_\A$ such 
	      that $t_S > 0$. Then $\calP(\run(p(k),Z)) =  \calP(\run(p(k),{\uparrow}_S))$.
       \item $Z = R_{II}(S)$, where $S$ is a BSCC of $\C_\A$ such that $t_S = 0$.
          If the type~III region determined by $S$ is non-empty (which can be checked in time
          polynomial in $\size{\A}$, then $R_{II}(S) = \emptyset$ and we are done. Otherwise,
          every configuration of $S \times \Nset$ can reach a configuration with zero counter. 
          Let ${\hookrightarrow} \subseteq S \times S$ be a binary relation such that
          $s \hookrightarrow t$ iff $s(0) \tran{} t(0)$ in $\M_\A$. Note that ${\hookrightarrow}$
          is computable in time polynomial in $\size{\A}$. Hence, we can also efficiently compute the BSCCs
          of $(S,{\hookrightarrow})$, and determine all \emph{non-trivial}
          BSCCs $K$ of $(S,{\hookrightarrow})$ such that for some (and hence all) $q \in K$ we have that
          $\post(q(0))$ is infinite. Each non-trivial BSCCs of $(S,{\hookrightarrow})$ corresponds to
          a type~II region determined by a control state of~$S$, and vice versa.
          Let us fix some control state $q_K \in K$ for each non-trivial BSCC $K$, and let $\mathcal{K}$
          be the set of all non-trivial BSCCs of $(S,{\hookrightarrow})$. Then
          $\calP(\run(p(k),Z)) =  \sum_{K\in \mathcal{K}} \calP(\run(p(k)\tran{}^* q_K(0)))$.
       \item Let $Z = R$, where $R$ is a type~II region determined by $q \in S$ where $S$ is a BSCC
          of $\C_\A$ such that $t_S < 0$. If $q$ belongs to a non-trivial BSCC of $(S,{\hookrightarrow})$
          (see the previous item), then $\calP(\run(p(k),Z)) = \calP(\run(p(k)\tran{}^* q(0)))$.
          Otherwise, $\calP(\run(p(k),Z)) = 0$.          
	\end{itemize}
	Hence, it suffices to show how to efficiently approximate $\calP(\run(p(k),{\uparrow}_S))$ and $\calP(\run(p(k)\tran{}^* q(0)))$ where we may further assume that 
	$\calP(\run(r(\ell) \tran{}^* q(0))) = 1$ for every $r(\ell) \in \post(q(0))$.
	In the following we assume that $k = 1$ and we prove that these probabilities can
	be approximated up to a relative error $\varepsilon>0$ in time polynomial in 
	$\size{\A}$ and $\size{\varepsilon}$ (for $k > 1$, we simply introduce 
	$k-1$ fresh control states that are used to increase the counter from $1$ to $k$).  
	
	Let us fix $q \in Q$ such that $\calP(\run(r(\ell) \tran{}^* q(0))) = 1$ for every $r(\ell) \in \post(q(0))$. Hence, 
	if $p(1) \in \post(q(0))$ then $\calP(\run(p(1)\tran{}^* q(0))) = 1$, and
	if $p(1) \not\in \pre(q(0))$ then $\calP(\run(p(1)\tran{}^* q(0))) = 0$.
	Now assume $p(1) \in \pre(q(0)) \smallsetminus \post(q(0))$.
	We construct
	a finite-state Markov chain $\calE_q$ where the set of vertices consists of
	all $r_j$ where $j \in \{0,1\}$ and $r(j) \in \pre(q(0)) \smallsetminus \post(q(0))$, and two fresh vertices $\mathit{good}$, $\mathit{bad}$. The outgoing transitions of a vertex $r_j$ are determined as follows:
	\begin{itemize}
		\item $r_0 \tran{x} s_j$ iff $s_j$ is a vertex of $\calE_q$ and 
		$r(0) \tran{x} s(j)$ in $\M_\A$;
		\item $r_0 \tran{x} \mathit{good}$ iff $x >0$ is the total probability
		of all transitions $r(0) \tran{y} s(j)$ in $\M_\A$ such that 
		$s(j) \in \post(q(0))$;
		\item $r_0 \tran{x} \mathit{bad}$  iff $x >0$ is the total probability
		of all transitions $r(0) \tran{y} s(j)$ in $\M_\A$ such that 
		$s(j) \not\in \pre(q(0))$;
		\item $r_1 \tran{x} s_0$ iff $s_0$ is a vertex of $\calE_q$ 
		and $x = [r{\downarrow}s] > 0$
		\item $r_1 \tran{x} \mathit{good}$ iff $x = \sum_{s(0) \in \post(q(0))} [r{\downarrow}s] > 0$;
		\item $r_1 \tran{x} \mathit{bad}$ iff $x = [r{\uparrow}] + \sum_{s(0) \not\in \pre(q(0))} [r{\downarrow}s] > 0$;
		\item $\mathit{good}\tran{1}\mathit{good}$,  
		$\mathit{bad}\tran{1}\mathit{bad}$.   	
	\end{itemize}
	It follows directly from the construction of $\calE_q$ that 
	$\calP(\run(p(1)\tran{}^* q(0)))$ is equal to the probability of reaching
	$\mathit{good}$ from $p_1$ in $\calE_q$. Note that $\calE_q$ is an absorbing
	finite-state Markov chain with two absorbing vertices $\mathit{good}$ and $\mathit{bad}$. Let $E_q$ be the other (transient) vertices of $\calE_q$,
	and let $U : E_q \rightarrow \Rset$ be the unique vector such that
	$U_v$ is the probability of reaching $\mathit{good}$ from $v$ in $\calE_q$.
	Then $U$ is the unique solution of the system $\vec{x} = A \vec{x} + C$, where
	$A$ is the $|E_q| \times |E_q|$ transition matrix for the transient part of
	$\calE_q$ and $C_v$ is the probability of the transition $v \tran{} \mathit{good}$
	in $\calE_q$ (if there is no such transition, then $C_v =0$). This system
	can be rewritten to the standard form $(I-A)\vec{x} = C$. Note that 
	$(I-A)^{-1}$ (i.e., the \emph{fundamental matrix} of $\calE_q$) satisfies
	$\norm{(I-A)^{-1}} \leq \max_{v \in E_q} m_v$, where $m_v$ is the  
	mean time of reaching an absorbing state from~$v$. Since every vertex
	of $E_q$ can reach $\mathit{good}$ in at most $2|Q|$ transitions and the 
	probability of each of these transitions is at least $\xmin^{|Q|^3}$ (see 
	\cite{EWY:one-counter-PE}), we obtain that $m_v$
	is bounded by $2|Q|/\xmin^{2|Q|^4}$. Since $\norm{I-A} \leq 1$, we obtain that
	the \emph{condition number} of $I-A$, i.e.,  $\norm{I-A} \cdot \norm{(I-A)^{-1}}$
	is bounded by $2|Q|/\xmin^{2|Q|^4}$ from above. By applying the standard result of
	numerical analysis (see, e.g., \cite{IK:book}), we obtain that if
	the coefficients of $A$ and $C$ are approximated so that the resulting 
	matrix $A'$ and vector $C'$ satisfy 
	$\norm{A-A'}/\norm{I-A} \leq \varepsilon \xmin^{2|Q|^4}/8|Q|$ and
	$\norm{C-C'}/\norm{C} \leq \varepsilon \xmin^{2|Q|^4}/8|Q|$, then the unique
	solution $U'$ of $(I-A')\vec{x} = C'$ satisfies 
	$\norm{U-U'}/\norm{U} \leq \varepsilon$. Since such $A'$ and $C'$ are computable
	in time polynomial in $\size{\A}$ and $\size{\varepsilon}$ 
	\cite{SEY:pOC-poly-Turing,BKK:pOC-time-LTL-martingale-JACM}, we are done.
	
	Now let $S$ be a BSCC of $\C_\A$. First, realize that if we change every rule
	$(r,\kappa,s) \in \gamma$, where $r$ belongs to a BSCC of $\C_\A$ different from 
	$S$, to $(r,-1,s)$, then the resulting pVASS $\A'$ satisfies 
	$\calP(\run_{\M_\A}(p(1),{\uparrow}_S)) = \calP(\run_{\M_{\A'}}(p(1),{\uparrow}))$.
	To simplify our notation, we directly assume that  $\calP(\run(p(1),{\uparrow}_S)) = \calP(\run(p(1),{\uparrow}))$, and we show
	how to approximate $\calP(\run(p(1),{\uparrow}))$. Let $\mathit{Diverge}$ be the set of
	all configurations $q(1)$ such that $q \in Q$ and $[q{\uparrow}] > 0$. If
	$p(1) \not\in \pre(\mathit{Diverge})$, we have that $\calP(\run(p(1),{\uparrow}))=0$. Otherwise, we construct a finite-state
	Markov chain $\G$ where the set of vertices consists of
	all $r_j$ where $j \in \{0,1\}$ and $r(j) \in \pre(\mathit{Diverge})$, and two fresh vertices $\mathit{good}$, $\mathit{bad}$. The transitions of
	$\G$ are determined as follows:
	\begin{itemize}
		\item $r_0 \tran{x} s_j$ iff $s_j$ is a vertex of $\G$ and 
		$r(0) \tran{x} s(j)$ in $\M_{\A}$;
		\item $r_0 \tran{x} \mathit{bad}$  iff $x >0$ is the total probability
		of all transitions $r(0) \tran{y} s(j)$ in $\M_{\A}$ such that 
		$s(j) \not\in \pre(\mathit{Diverge})$;
		\item $r_1 \tran{x} s_0$ iff $s_0$ is a vertex of $\G$ 
		and $x = [r{\downarrow}s] > 0$;
		\item $r_1 \tran{x} \mathit{good}$ iff $x = [r{\uparrow}] > 0$;
		\item $r_1 \tran{x} \mathit{bad}$ iff $x = \sum_{s(0) \not\in \pre(\mathit{Diverge})} [r{\downarrow}s] > 0$;
		\item $\mathit{good}\tran{1}\mathit{good}$,  
		$\mathit{bad}\tran{1}\mathit{bad}$.   	
	\end{itemize}
	It is easy to check that $\calP(\run(p(1),{\uparrow}))$ is equal to the probability
	of reaching $\mathit{good}$ from $p_1$ in $\G$. The rest of the argument is the
	same as above, the only difference is that now we also employ the
	$\xmin^{4|Q|^2}\cdot t^3/(7000 \cdot |Q|^3)$ lower bound on positive probability of
	the form $[r{\uparrow}]$, where $t$ is the trend of~$S$ (see Theorem~4.8 in
	\cite{BKK:pOC-time-LTL-martingale-JACM}).
\end{proof}



%% file: app-thms.tex
\newcommand{\twot}{t}
\newcommand{\octrend}{\tau}
\newcommand{\authNote}[1]{}
\newcommand{\trenddef}{\mathit{OC\text{-}D}}

\newcommand{\rep}{\mathit{Repeat}}
\newcommand{\reptime}{TRepeat}
\newcommand{\repconf}[1]{\#_{#1}}
\newcommand{\last}{\mathit{last}}
\newcommand{\msalt}[1]{\widehat{m}^{(#1)}}

\section{Proofs of Section~\ref{sec-two-counters}}
\label{app-two-counters}

The whole Appendix~\ref{app-two-counters} is devoted to proofs of the following three theorems.




\begin{reftheorem}{Theorem}{\ref{thm:tail-bounds-height}}
Let $S$ be a BSCC of $\C_\A$ such that $t_S(2)<0$. Moreover, let $R$ be any type-II region of $\A_2$ determined by some state of $ S$.
Then there are $a_1,b_1\in \Rset_{>0}$ and $z_1\in (0,1)$ computable in polynomial space such that the following holds for all $p\in S$ such that $p(0)\in R$, all $n\in \Nset$, and all $i\in \Nset^+$:
\[
\calP_{p(n,0)} (T < \infty \wedge 
\xs{T}_2\geq i)\quad \leq\quad a_1\cdot z_1^{b_1 \cdot i}
\]
Moreover, if $E_{p(n,0)}<\infty$, then it holds
\[
E_{p(n,0)}\left(\xs{T}_2\right)\quad \leq\quad  \frac{a_1\cdot z_1^{b_1}}{1-z_1^{b_1}}.
\]
In particular, {\em neither} of the bounds depends on $n$.
\end{reftheorem}


\begin{reftheorem}{Theorem}{\ref{thm:tail-bounds-divergence}}
Let $S$ be any BSCC of $\C_\A$ such that $t_S(2)<0$. Furthermore, let $R$ be any type II region of $\A_2$ determined by some state of $S$ and satisfying $\octrend_R>0$. Then there are 
numbers $a_2,b_2>0$ and $0<z_2<1$ computable in space bounded by a polynomial in $\size{\A}$ such that for all configurations $p(n,0)$, where $p(0)$ belongs to $R$, the following holds:
\[
[p(n,0)\rightarrow^* q(0,*)]\quad \leq \quad a_2\cdot z_2^{n \cdot b_2}
\]
\end{reftheorem}


\begin{reftheorem}{Theorem}{\ref{thm:tail-bounds-convergence}}
Let $R$ be any type II region of $\A_2$ such that $\tau_R<0$. Then there are
numbers $a_3,b_3,d_3>0$ and $0<z_3<1$, computable in space bounded by a polynomial in $\size{\A}$, such that for all configurations $p(n,0)$, where $p(0)$ belongs to $R$, and all $q\in Q$ the following holds:
\[
[p(n,0)\rightarrow^* q(0,*),i]\quad \leq \quad i\cdot a_3\cdot z_3^{\sqrt{n\cdot\tau_R\cdot b_3 + i\cdot d_3}}
\]
for all $i\geq \frac{H\cdot n}{-{\tau_R}}$ where $H$ is a computable constant.


\end{reftheorem}

We use the following additional notation: we denote by $ \run(p\vec{v}\rightarrow^* q\vec{u},i)$ the set of all runs $w \in \run(p\vec{v})$ such that $T(w)\geq i$, $w(i)=q\vec{u}$, and for all $0\leq j<i$ we have $w(j)\not = q\vec{u}$. Note that the probability of $ \run(p\vec{v}\rightarrow^* q\vec{u},i)$ is exactly the number $[p\vec{v}\rightarrow^* q\vec{u},i]$.

\subsection{Martingale Techniques}
\label{subsec:twoc-martingale}

In this subsection we use the techniques of martingale theory to prove several technical lemmas that are crucial for the analysis of those sets of configurations $C[R_1,R_2]$ such that $R_2$ is a type II region in which the corresponding counter has a tendency to decrease.

To this end, fix a region $R$ of $\A_2$ whose type  is II  and which is determined by some $p\in S$, where $S$ satisfies $\twot_S(2)<0$. 
We use the stochastic process $\{\ms{\ell}\}_{\ell=0}^{\infty}$ defined in~\cite{BKKNK:pMC-zero-reachability} as follows: for every $\ell \in \Nset$ we put
\begin{equation}
  \ms{\ell} := \begin{cases}\xs\ell_1 - \octrend_R\cdot \ell + \vec{g}\big(\xs\ell_2\big)[\ps\ell] & 
\text{if $\xs{j}_1 > 0 $ for all $0\leq j < \ell$,}\\
\ms{\ell-1} & \text{otherwise.}
  \end{cases} \label{eq-mart-m}
\end{equation}
(Here $\vec{g}$ is a suitable function, defined precisely in~\cite{BKKNK:pMC-zero-reachability}, assigning numerical weights to configurations of $\A_2$.)
In other words, the value $\ms{\ell}(w)$ of the $\ell$-th configuration of a run $w$ is obtained by adding the value of the first counter $\xs\ell_1(w)$ to the value of the second counter $\xs\ell_2(w)$ weighted by the function $\vec{g}$ and by subtracting $\ell$ times the trend $\octrend_R$.

In~\cite{BKKNK:pMC-zero-reachability} we defined the function $\vec{g}$ in such a fay that the process $\{\ms{\ell}\}_{\ell = 0}^{\infty}$ satisfies several important properties: First of all, it is a \emph{martingale} \cite{Williams:book}, which intuitively means that the expected value of the $\ell$-th configuration is always equal to the observed value of the $(\ell-1)$-th configuration, even if we are given the knowledge of values of all configurations up to the $(\ell-1)$-th step.  

 A second crucial observations proved in~\cite{BKKNK:pMC-zero-reachability} is that grows more or less linearly with $\xs\ell_2$.
\begin{lemma}[\cite{BKKNK:pMC-zero-reachability}]\label{lem:bounded-weight}
There is a number $C>0$ computable in space bounded by a polynomial in $\size{\A}$
 such that for all $p\in Q$ and $n\geq 1$ it holds $|\vec{g}(0)[p]|\leq C$ and $|\vec{g}(n)[p]|\leq C \cdot n$.
\end{lemma}
%

However, to apply powerful tools of martingale theory, such as the Azuma's inequality, we need to show that $\{\ms{\ell}\}_{\ell = 0}^\infty$ has {\em bounded-differences}. This is covered in the following lemma, whose proof combines several facts shown in~\cite{BKKNK:pMC-zero-reachability} with rather involved techniques from the theory of stochastic matrices. For better readability, we prove this lemma separately in Appendix~\ref{sub-bounded-difference}.

\begin{lemma}\label{lem:bounded-differences}
There is a bound $B\geq 1$ computable in in space bounded by a polynomial in $\size{\A}$ such that $|\ms{\ell+1}-\ms{\ell}|\leq B$ for every $\ell\in \Nset$.
%
\end{lemma}

%
%

The power of martingale techniques is illustrated in the following lemma, which will be handy in the proof of Theorem~\ref{thm:tail-bounds-height}.

\begin{lemma}
\label{lem:return-bound}
Let $R$ be a type II region of $\A_2$ determined by some $p\in S$, where $S$ is a BSCC of $\C_\A$ such that $t_S(2)<0$.
Moreover, let
$p(n,0)$ be any configuration such that $p(0)\in R$. Denote by $\rep$ the set of all runs $w$ such that $w(0)=w(i)$ for some positive $i$ such that $4B^2/\octrend_R^2\leq i \leq T(w)$. Then
\[
\calP_{p(n,0)}(\rep) < \frac{1}{2}.
\] 
\end{lemma}
\begin{proof}
For any run $w\in \rep$ let $\reptime(w)$ be the smallest $i\geq 4B^2/\octrend_R^2$ such that $w(i)=w(0)$ (for such $w$ we have $\reptime(w)\leq T(w)$). We have 
\begin{equation}
\label{eq:repeat-partition}
\calP_{p(n,0)}(\rep) = \sum_{i=\lceil 4B^2/\octrend_R^2 \rceil}^{\infty}\calP_{p(n,0)}(\rep \wedge \reptime=i)
\end{equation}

Now any run $w$ initiated in $p(n,0)$ satisfies $\ms{0}(w)=n+\vec{g}(0)[p]$. Similarly, any run $w\in\run(p(n,0))\cap\rep$ that satisfies $\reptime(w)=i$ satisfies $\ms{i}(w)=n+\vec{g}(0)[p]-i\cdot\octrend_R$. Hence, 
\begin{equation}
\label{eq:repeat-to-mart}
\calP_{p(n,0)}(\rep \wedge \reptime=i) \quad \leq \quad \calP_{p(n,0)}(|\ms{i}-\ms{0}| \geq i\cdot|\octrend_R|).
\end{equation}
From Azuma's inequality we get
\[
P_{p(n,0)}(|\ms{i}-\ms{0}| \geq i\cdot|\octrend_R|) \quad \leq\quad 2\exp\left(\frac{-\octrend_R^2 \cdot i^2}{2iB^2} \right) \quad = \quad 2\exp\left(\frac{-\octrend_R^2 \cdot i}{2B^2} \right)
\]
Combining this with~\eqref{eq:repeat-to-mart} and~\eqref{eq:repeat-partition} we get
\begin{align*}
\calP_{p(n,0)}(\rep) \quad&\leq\quad \sum_{i=\lceil 4B^2/\octrend_R^2 \rceil}^{\infty} 2\exp\left(\frac{-\octrend_R^2 \cdot i}{2B^2} \right) \quad \leq\quad \frac{2}{\exp\big(\lceil\frac{4B^2}{\octrend_R^2}\rceil \cdot \frac{\octrend_R^2}{2B^2}\big)\cdot\left(1-\exp\big(-\octrend_R^2/2B^2 \big)\right)}\\
&\leq \quad \frac{2}{e^2\cdot(1-\frac{1}{e^2})} \quad < \quad \frac{1}{2}.
\end{align*}
\end{proof}

The following lemma, which is crucial for the proof of Theorems~\ref{thm:tail-bounds-divergence} and~\ref{thm:tail-bounds-convergence}, is also proved using Azuma's inequality.

\begin{lemma}\label{lem:crucial-bound}
Let $R$ be a type II region of $\A_2$ determined by some $p\in S$, where $S$ is a BSCC of $\C_\A$ such that $t_S(2)<0$.
Then there are $a',b',b''\in \Rset_{>0}$ and $c'\in (0,1)$ such that for all $p,q\in Q$, where $p(0)\in R$, and all $i,n,n'\in \Nset$ satisfying $i\geq (2\cdot (C+1)\cdot n')/|\octrend_R|$ the following holds:
If either $\tau_R>0$, or $-\tau_R\cdot \frac{i}{2}\geq n$, then
\[
[p(n,0)\rightarrow^* q(0,n'),i]\leq a'\cdot (c')^{ n\cdot \tau_R\cdot b' + i\cdot b''}
\]
Moreover, $a',b',b'',c'$ are effectively computable in polynomial space.
\end{lemma}
\begin{proof}
Denote $\vec{v}=(n,0)$ and $\vec{u}=(0,n')$. For better readability, we denote by $\tau$ the number $\tau_R$, where $R$ is the region of $\A_2$ containing $p(0)$. Let $w\in \run(p\vec{v}\rightarrow^* q\vec{u},i)$.\authNote{neni def} Then
\begin{align*}
(\ms{i}-\ms{0})(w) & =\xs{i}_1(w) - \octrend\cdot i + \vec{g}\big(\xs{i}_2(w)\big)[\ps{i}(w)] \\
  & \quad -\xs{0}_1(w) + \octrend\cdot 0 - \vec{g}\big(\xs{0}_2(w)\big)[\ps{0}(w)]\\
  & = -\octrend i + \vec{g}(n')[q]-n-\vec{g}(0)[p]\\
  & = \vec{g}(n')[q]-\vec{g}(0)[p]-n-\octrend i.
\end{align*}
Thus,
\[
[p\vec{v}\rightarrow^* q\vec{u},i]\leq \calP(\ms{i}-\ms{0}=\vec{g}(n')[q]-\vec{g}(0)[p]-n-\octrend i).
\]
Note that if $i<n$, then $[p\vec{v}\rightarrow^* q\vec{u},i]=0$. Assume $i\geq n$. Since $i\geq (2\cdot (C+1)\cdot n')/|\octrend|\geq (2\cdot |\vec{g}(n')[q]-\vec{g}(0)[p]|)/|\octrend|$, we have 
\begin{equation}\label{eq:twoc-weight-diff-bound}|\vec{g}(n')[q]-\vec{g}(0)[p]|\leq |\octrend|\frac{i}{2}. \end{equation}

\noindent
Thus, denoting $Z=\vec{g}(n')[q]-\vec{g}(0)[p]$, the following holds:

\begin{itemize}
\item If $\octrend<0$, then from~\eqref{eq:twoc-weight-diff-bound} we have $Z\geq \octrend\frac{i}{2}$ and thus $Z-\octrend\frac{i}{2}\geq 0$. Hence,
\begin{equation}\label{eq:crucial-bound-eq1}
[p\vec{v}\rightarrow^* q\vec{u},i]\quad \leq\quad  \calP(\ms{i}-\ms{0}\geq -n-\octrend\frac{i}{2}).
\end{equation}
\item If $\octrend>0$, then from \eqref{eq:twoc-weight-diff-bound} we have $Z\leq |\octrend|\frac{i}{2} = \octrend\frac{i}{2}$, and thus $Z-\octrend\frac{i}{2}\leq 0$. Hence,
\begin{equation}\label{eq:crucial-bound-eq2}
[p\vec{v}\rightarrow^* q\vec{u},i]\quad \leq\quad  \calP(\ms{i}-\ms{0}\leq -n-\octrend\frac{i}{2}).
\end{equation}
\end{itemize}

\noindent
Now we apply the Azuma's inequality.
First consider $\octrend>0$. Then, by Azuma's inequality, for all $i$ it holds
\begin{align*}
\calP(\ms{i}-\ms{0} \leq -n-\octrend\frac{i}{2})&\leq 2\exp\left(\frac{-(n+\frac{\octrend}{2} i)^2}{2\cdot B\cdot i}\right)\\
& =2\exp\left(\frac{-n^2-n \octrend i-\frac{\octrend^2 i^2}{4}}{2\cdot B\cdot i}\right)\\
& =2\exp\left(\frac{1}{2B}\left(\frac{-n^2}{i}-n \octrend-i\frac{\octrend^2}{4}\right)\right)\\
& \leq 2\exp\left(\frac{1}{2B}\left(-n\octrend-i\frac{\octrend^2}{4}\right)\right)\\
& \leq 2\exp\left(\frac{-\octrend^2}{8B}\left(i+n\octrend\right)\right).
\end{align*}
(For the last inequality we used the fact that $\octrend^2/4 <1$.) 
Combining this inequality with (\ref{eq:crucial-bound-eq2}) we obtain
\[
[p\vec{v}\rightarrow^* q\vec{u},i]\quad \leq \quad 2 \exp\left(\frac{-\octrend^2}{8B}\left(i+n\octrend\right)\right)
\]

It is now easy to compute $a',b',b'', c'$ from the statement of the lemma: it suffices to put $a'=2$, $c'=1/2$ and $b'=b''=\frac{x^2}{8B}$, where $x$ is a number, computable in polynomial space, such that $\octrend \geq x>0$ ($x$ can be computed in polynomial space since $\octrend$ can be encoded in Tarski's algebra).

Now consider $\octrend<0$. To apply Azuma's inequality in this case, we need to assume that $-\octrend \frac{i}{2}\geq n$. Then, as above,
\begin{align*}
\calP(\ms{i}-\ms{0} & \geq -n-\octrend\frac{i}{2})\quad \leq \quad 2 \exp\left(\frac{1}{2B}\left(-n\octrend-i\frac{\octrend^2}{4}\right)\right),
\end{align*}

\noindent
and combining this with~\eqref{eq:crucial-bound-eq1} yields
\[
[p\vec{v}\rightarrow^* q\vec{u},i]\quad \leq\quad 2 \exp\left(\frac{1}{2B}\left(-n\octrend-i\frac{\octrend^2}{4}\right)\right).
\]
Numbers $a',b',b'', c'$ can be now easily computed  (we put $b'=1/2B$ and $b''=x^2/8B$, where $x$ is as above).
\end{proof}

%

\subsection{Proof of Theorem~\ref{thm:tail-bounds-height}}

Fix a region $R$ of $\A_2$ that satisfies the assumptions of Theorem~\ref{thm:tail-bounds-height}. 


%


\authNote{Pouzivam $\len{(beh)}=\infty$}
For the purpose of this proof we define the value $T(w)$ also for finite paths $w$: we put $T(w)=\inf\{i\mid 0\leq i \leq \len(w) \wedge \xs{i}_2 = 0 \}$.
Given a~finite path or a~run $w$, we denote by $\mathit{LVisit}(w)$ the largest number $k$ such that $\xs{k}_2(w)=0$ and for all $0\leq i\leq k$ we have $\xs{i}_1(w)>0$ (if there are infinitely many such $k$, which is possible only if $\len(w)=T(w)=\infty$, we put $\mathit{LVisit}(w)=\infty$). Further, for $w$ such that $\mathit{LVisit}(w)\in \Nset$ we let $\mathit{Last}(w)$ be the configuration $w(\mathit{LVisit}(w))$. Finally, for a finite path or a run $w$ and a configuration $q\vec{u}$ we denote by $\repconf{q\vec{u}}(w)$ the number of occurrences of $q\vec{u}$ on $w$ \emph{before} zeroing the first counter. Formally, we put $$\repconf{q\vec{u}}(w)\quad=\quad|\{i\in \Nset \mid 0\leq i \leq T(w) \wedge w(i)=q\vec{u}\}|.$$

We have
\begin{align}
& \calP_{p(n,0)} (T<\infty \wedge \xs{T}_2\geq i) 
 \quad = \quad\sum_{q\in Q}\sum_{k=1}^{\infty} \calP_{p(n,0)}(T<\infty \wedge \xs{T}_2\geq i\wedge \mathit{Last}=q(k,0)) \nonumber\\
 \nonumber\\
&\quad = \quad\sum_{q\in Q}\sum_{k=1}^{\infty}\sum_{m=1}^{\infty} \calP_{p(n,0)}(T<\infty \wedge \xs{T}_2\geq i\wedge \mathit{Last}=q(k,0)
\wedge \repconf{q(k,0)}=m) \label{eq:oc-height-bound}
\end{align}
Now denote by $A_{q,k}^{i,m}$ the event $T<\infty \wedge \xs{T}_2\geq i\wedge \mathit{Last}=q(k,0)
\wedge \repconf{q(k,0)}=m$. It holds
\[
A_{q,k}^{i,m} = \bigcup_{\substack{w\in\fpath(p(n,0)) \\ \repconf{q(k,0)}=m \\ w(\len(w))=q(k,0)}} \bigcup_{\substack{w'\in\run(q(k,0)) \\T(w')<\infty\\ \xs{T}_2(w')\geq i \\ \bigwedge_{j=1}^{T(w')-1}\xs{j}_2(w')>0  }} \run(w\cdot w').
\]
(Here $w\cdot w' = w(0),w(1),\dots,w(\len(w)-1),w'(0),w'(1),\dots$.) It follows that 
\begin{align}
&\calP_{p(n,0)}(A_{q,k}^{i,m}) \quad=  \sum_{\substack{w\in\fpath(p(n,0)) \\ \repconf{q(k,0)}=m \\ w(\len(w))=q(k,0)}} \calP_{p(n,0)}(\run( w))\cdot\Big(\sum_{\substack{w'\in\run(q(k,0)) \\T(w')<\infty\\ \xs{T}_2(w')\geq i \\ \bigwedge_{j=1}^{T(w')-1}\xs{j}_2(w')>0  }} \calP_{p(n,0)}(\run( w'))\Big) \nonumber\\
&\quad = \quad\calP_{p(n,0)} (\repconf{q(k,0)}\geq m)\cdot \calP_{q(k,0)}(T<\infty \wedge \xs{T}_2\geq i\wedge \bigwedge_{j=1}^{T-1}\xs{j}_2 >0)
.\label{eq:height-split}
\end{align}

\noindent
Note that every run $w$ initiated in $q(k,0)$ that satisfies ${\xs{T}_2(w)\geq i}$ and $\bigwedge_{j=1}^{T} \xs{j}_2(w)>0$ must have a prefix of length at least $i+k$ such that for every $0<j< i+k$ we have $\xs{j}_2(w)>0$. It follows that
\begin{equation}
\label{eq:height-time-bound}
\calP_{q(k,0)}(T<\infty\wedge\xs{T}_2\geq i\wedge \bigwedge_{j=1}^{T-1} \xs{j}_2>0)\quad \leq \quad \calP_{q(0)}(L\geq i+k).
\end{equation}
Here $\calP_{q(0)}(L\geq i+k)$ is measured in $\A_2$ and $L$ assigns to a given run $w$ of $\mathcal{A}_2$ either the least $k>0$ such that the counter is zero in $k$-th step of $w$, or $\infty$ if there is no such $k$ (intuitively, $L(w)$ is the number of steps in which $w$ (re)visits a configuration with zero counter value for the first time).

By~\cite[Section 3.1]{BKK:pOC-time-LTL-martingale-JACM}, we can compute, in polynomial time, a number $z$ such that 
\begin{equation*} 
\calP_{q(0)}(L\geq i+k)\quad \leq\quad 1+2|z|/|t_2(s)| + |Q|\cdot \frac{2d^{i+k}}{1-d},
\end{equation*}
where $d=\exp\big(-\frac{(\twot_S(2))^2}{8(z+|\twot_s(2)|+1)^2}\big)\in(0,1)$. Using this knowledge, we can easily compute, in polynomial time, numbers $a,b \in \Rset_{>0}$, and $c \in (0,1)$ such that  
\begin{equation*}\label{eq:bound-one-counter}
\calP_{q(0)}(L\geq i+k)\quad \leq\quad a\cdot c^{b \cdot (i+k)}.
\end{equation*}
Plugging this into~\eqref{eq:height-time-bound},~\eqref{eq:height-split} and~\eqref{eq:oc-height-bound} we get
\begin{equation}
\label{eq:height-midway}
\calP_{p(n,0)} (T<\infty \wedge \xs{T}_2\geq i) \quad \leq \quad \sum_{q\in Q}\sum_{k=1}^{\infty}a \cdot c^{b\cdot (i+k)}\cdot \Big(\sum_{m=1}^{\infty}\quad\calP_{p(n,0)} (\repconf{q(k,0)}\geq m)\Big).
\end{equation}

\noindent
We now turn our attention to bounding $\calP_{p(n,0)} (\repconf{q(k,0)}\geq m)$. From Lemma~\ref{lem:return-bound} it follows that for any $h\in \Nset$ it holds $$\calP_{p(n,0)} (\repconf{q(k,0)}\geq 4hB^2/\octrend_R^2)\leq 2^{-h}.$$

\noindent
From this it follows that
\begin{equation*}
\sum_{m=1}^{\infty}\calP_{p(n,0)} (\repconf{q(k,0)}\geq m)\quad\leq\quad \frac{4B^2}{\octrend_R^2} \cdot\sum_{h=0}^{\infty}2^{-h} \quad=\quad \frac{8B^2}{\octrend_R^2}.
\end{equation*}

\noindent
Plugging this bound into~\eqref{eq:height-midway} we get
\[
\calP_{p(n,0)} (T<\infty \wedge \xs{T}_2\geq i) \quad \leq \quad \sum_{q\in Q}\sum_{k=1}^{\infty}a \cdot c^{b\cdot (i+k)}\cdot \frac{8B^2}{\octrend_R^2} \quad = \quad |Q|\cdot \frac{a\cdot c^b\cdot 8B^2}{(1-c^b)\cdot \octrend_R^2}\cdot c^{b\cdot i},
\]

\noindent
and from this form the numbers $a_1,b_1,c_1$ in the statement of Theorem~\ref{thm:tail-bounds-height} can be easily computed.

Finally, when $\calP_{p(n,0)}(T<\infty)=1$, we have
\begin{align*}
E_{p(n,0)} \left(\xs{T}_2\right)\quad =\quad \sum_{i\geq 1} \calP_{p(n,0)} (\xs{T}_2\geq i)
\quad \leq\quad  \sum_{i\geq 1}a_1\cdot c_1^{b_1 \cdot i}\quad = \quad\frac{a_1\cdot c_1^{b_1}}{1-c_1^{b_1}}
\end{align*}

%
%

\subsection{Proof of Theorem~\ref{thm:tail-bounds-divergence}}
\smallskip

We start with a simple corollary of the facts that we proved so far. The corollary will be also useful in the proof of Theorem~\ref{thm:tail-bounds-convergence}. 

\begin{corollary}
\label{cor:generic-time-bound}
Let $R$ be a type II region of $\A_2$ determined by some $p\in S$, where $S$ is a BSCC of $\C_\A$ such that $t_S(2)<0$. Moreover, let $p(n,0)$ be any configuration of $\A$ such that $p(0)\in R$. Denote  $k= \frac{| \octrend_R|}{2(C+1)}$, where $C$ is as in Lemma~\ref{lem:bounded-weight}. Then for arbitrary $i\in \Nset$ such that either $i\geq 2\cdot n/|\octrend_R|$ or $\octrend_R<0$ it holds
\[
[p(n,0) \rightarrow^* q(0,*),i] \leq 
i\cdot k\cdot a'\cdot (c')^{ n\cdot\octrend_R\cdot b' +  i\cdot k\cdot b''}   + \frac{a_1}{1-c_1^{b_1}}\cdot c_1^{i\cdot b_1 \cdot k} ,
\]
where $a',b',b'',c'$ are as in Lemma~\ref{lem:crucial-bound} and $a_1,b_1,c_1$ are as in Theorem~\ref{thm:tail-bounds-height}.
\end{corollary}
\begin{proof}
We have
\begin{align*}
&[p(n,0) \rightarrow^* q(0,*),i]  \quad =\quad\sum_{n'=0}^{i} [p(n,0)\rightarrow^* q(0,n'),i] 
\quad = \quad \sum_{n'=1}^{\lfloor i\cdot k\rfloor } [p(n,0)\rightarrow^* q(0,n'),i]+\sum_{n'=\lfloor i\cdot k+1 \rfloor}^{i} [p(n,0)\rightarrow^* q(0,n'),i]\nonumber \\
&\quad  \leq \quad i\cdot k\cdot a'\cdot (c')^{ n\cdot\octrend_R\cdot b' + i\cdot b''} +\sum_{n'=\lfloor i\cdot k+1 \rfloor}^{i} [p(n,0)\rightarrow^* q(0,n'),i] \quad \leq\quad i\cdot k\cdot a'\cdot (c')^{ n\cdot\octrend_R\cdot b' +  i\cdot k\cdot b''}+ \sum_{n'=\lfloor i\cdot k+1 \rfloor}^{i} a_1\cdot c_1^{b_1\cdot i} \nonumber\\
&\quad \leq \quad i\cdot k\cdot a'\cdot (c')^{ n\cdot\octrend_R\cdot b' +  i\cdot k\cdot b''}   + \frac{a_1}{1-c_1^{b_1}}\cdot c_1^{i\cdot b_1 \cdot k} 
.\label{eq:axis-time-bound} 
\end{align*}
Here the first inequality on the second line follows from Lemma~\ref{lem:crucial-bound} (note that any $n' \leq i\cdot k$ satisfies $i\geq (2(C+1)n')/|\octrend_R|$, so the assumptions of this Lemma are satisfied), while the second inequality follows from Theorem~\ref{thm:tail-bounds-height}.
\end{proof}

We now proceed with the proof of Theorem~\ref{thm:tail-bounds-divergence}. Since $\octrend_R>0$, from Corollary~\ref{cor:generic-time-bound} we get 
\[
[p(n,0) \rightarrow^* q(0,*),i] \leq i\cdot k\cdot a'\cdot (c')^{ i\cdot k\cdot b''}   + \frac{a_1}{1-c_1^{b_1}}\cdot c_1^{i\cdot b_1 \cdot k} 
\]

\noindent
But then
\[
[p(n,0) \rightarrow^* q(0,*)] = \sum_{i=n}^{\infty} [p(n,0) \rightarrow^* q(0,*),i] \leq \frac{n\cdot k\cdot a'}{(1-(c')^{k\cdot b''})^2}\cdot(c')^{n\cdot k\cdot b''}+\frac{a_1}{(1-c_1^{b_1})\cdot(1-c_1^{b_1\cdot k})}\cdot c_1^{n\cdot b_1\cdot k}.
\]
The numbers $a_2$, $b_2$, $z_2$ in Theorem~\ref{thm:tail-bounds-divergence} can be straightforwardly computed from this bound.
\subsection{Proof of Theorem~\ref{thm:tail-bounds-convergence} for $t_S(2)<0$}
\smallskip
In this subsection we prove Theorem~\ref{thm:tail-bounds-convergence} under the assumption that the region $R$ is determined by some $p\in S$, where $S$ is a BSCC of $\C_\A$ with $t_S(2)<0$. The case when $t_S(2)>0$ is handled separately in the next subsection.

We put $\tilde{a}_3 = 2\cdot\max\{k\cdot a',a_1/(1-c_1^{b_1})\}$, $\tilde{b}_3 =b'$, $\tilde{z}_3=\max\{c',c_1\}$ and $\tilde{d}_3= k\cdot \min\{b'',b_1\}$.
From Corollary~\ref{cor:generic-time-bound} it immediately follows that
\begin{align*}
[p(n,0) \rightarrow^* q(0,*),i]\quad \leq \quad i\cdot  \tilde{a}_3\cdot (\tilde{z}_3)^{n\cdot \octrend_R\cdot \tilde{b}_3 + i\cdot \tilde{d}_3}
\end{align*} whenever $i\geq \frac{2n}{-\octrend_R}$. From this, the numbers $a_3$, $b_3$, $z_3$ in the statement of the theorem can be easily.
\subsection{Proof of Theorem~\ref{thm:tail-bounds-convergence} for $\twot_S(2)>0$}
Intuitively, if the second counter starts high enough and $t_S(2)>0$, then the probability of reaching zero in the second counter is negligibly small, and hence we may basically ignore the value of the second counter. We obtain a one counter pVASS, preserving the behaviour of the first counter, on which we may easily bound time to zeroing this counter using the previous results. 

Within this subsection we often operate with several probability measures within a single expression. To differentiate between them, we denote by $\calP^{\B}$ the probability measure associated to a pVASS $\B$.

In the proof we use some known results on one counter pVASS. Let $\B$ be a one counter pVASS and let 
$S$ be a BSCC of $\C_\B$ such that $t_S\neq 0$. As shown in the proof of Lemma~5.6 in~\cite{BKK:pOC-time-LTL-martingale-JACM} (see also Proposition~7 in~\cite{BKK:pOC-time-LTL-martingale}) one can compute, in time polynomial in the size of $\B$, a bound $h_S\in\Nset$ and numbers  $a_S>0,c_S\in(0,1)$ such that for all configurations $p(k)$, where $p\in S$, all states $q$ of $\B$, and all $i\geq h_S \cdot k$ it holds 
\begin{equation}
\label{eq:oc-time-bound}
\calP^{\B}_{p(k)}\big(\run(p(k)\rightarrow^* q(0),i )\big)\leq a_S\cdot c_S^{-i}.
\end{equation}\authNote{Poradne zadefinovat na zacatek app}

Another crucial tool from the world of one-counter pVASS is the \emph{divergence gap theorem}~\cite[Theorem 4.8]{BKK:pOC-time-LTL-martingale-JACM}, which bounds a positive non-termination probability away from zero. The theorem states that for any BSCC $S$ of $\B$ with $t_{S}>0$ there is a number $\delta_S\in (0,1)$ computable in polynomial time such that for every $p\in S$ for which the probability of avoiding zero when starting in $p(1)$ is positive it holds
\[
\sum_{q\in Q}\calP^{\B}_{p(1)}\big(\run(p(1)\rightarrow^* q(0))\big) \leq 1-\delta_S.
\]
In~\cite{BKKNK:pMC-zero-reachability} we proved that for any $S$ with $t_S>0$ there is $\ell \in \Nset$ such that for every $p\in S$ the probability of avoiding zero counter from $p(\ell)$ is positive. Now if $R$ is in a type II region of $\B$ determined by some $p\in S$, then from every configuration $q\in S $ such that $q(|S|)\in R$ there is a zero-avoiding path from $q(|S|)$ to some $r(\ell)$, so the probability of avoiding zero from $q(|S|)$ is also positive. 
From this it  follows that for any $p(k)\in R$, where $k\geq |S|$, it holds
\[
\sum_{q\in Q}\calP^{\B}_{p(k)}\big(\run(p(k)\rightarrow^* q(0))\big) \leq (1-\delta_S)^{k-|S|}.
\]

Both of the above results on one-counter pVASS were originally proved using a suitable martingale $\{\msalt{\ell}\}_{\ell=0\infty}$. We will need this martingale in this subsection as well, to prove some additional auxiliary results. The process $\{\msalt{\ell}\}_{\ell=0}^{\infty}$ can be seen as the one-counter analogue of the (substantially more complex) two-counter martingale $\{\ms{\ell}\}_{\ell=0}^{\infty}$ which was defined in one of the previous sections. Formally, let us fix a one-counter pVASS $\B$ and a BSCC $S$ of $\C_\B$. 
We define a stochastic process $\{\msalt{\ell}\}_{\ell=0}^{\infty}$ on runs of $\B$ by putting, for every $\ell \in \Nset$,
\[
\msalt{\ell}=\begin{cases}\xs\ell_1 - t_S\cdot \ell + \vec{z}(\ps{\ell}) & 
\text{if $\xs{j}_1 > 0 $ for all $0\leq j < \ell$,}\\
\ms{\ell-1} & \text{otherwise,}
  \end{cases} 
\]
where $\vec{z}$ is a suitable function  assigning numerical weights to states in $S$. It was shown in~\cite{BKK:pOC-time-LTL-martingale-JACM} that one can compute, in time polynomial in size of $\B$, a function $\vec{z}$ such that 
\begin{itemize}
\item
$\max_{p,q\in S}|\vec{z}(p)-\vec{z}(q)|\leq h_S$, and
\item
the stochastic process $\{\msalt{\ell}\}_{\ell=0}^{\infty}$ is a martingale whenever the initial state belongs to $S$.
\end{itemize}

Now fix a two counter pVASS $\A$ and let $\hat{a}$, $\hat{c}$, and $\hat{h}$ be the maximal $a_S$ and $c_S$ and $h_S$, respectively, among all BSCCs $S$ of $\A_1$ and $\A_2$. Note that $\hat{a}$, $\hat{c}$, and $\hat{h}$ can be computed in time polynomial in the size of $\A$.

\begin{lemma}\label{lem:crucial-bound2}
Let $R$ be a region of $\A_2$ satisfying the assumptions of Theorem~\ref{thm:tail-bounds-convergence} whose corresponding BSCC $S$ satisfies $t_S(2)>0$. Then there are numbers $a_4,b_4>0$, $c_4\in (0,1)$ computable in polynomial time such that 
for every $p\in S$ with $p(0)\in R$, every $n\in\Nset$, every $i\geq \max\{2\hat{h}n,4|S|\}$, and every $k\geq i/8$ it holds
\[
[p(n,k)\rightarrow^* q(0,*),i/2]\leq a_4\cdot (c_4)^{i \cdot b_4}
\]
\end{lemma}
\begin{proof}
Note that 
\begin{equation}
\label{eq:divergence-split}
[p(n,k)\rightarrow^* q(0,*),i/2] \quad\leq \quad\calP^{\A_1}_{p(n)}\big(\run(p(n)\rightarrow^* q(0),i/2)\big) + \sum_{r\in Q}\sum_{j=0}^{i/2}\calP^{\A_2}_{p(k)}\big(\run(p(k)\rightarrow^* r(0),j)\big) .
\end{equation}
This is because the set of runs initiated in $p(n,k)$ that visits a configuration with zero value in the second counter before reaching zero in the first counter has the same probability as the corresponding set of those runs of $\A_1$ initiated in $p(n)$ whose accumulated payoff (i.e. the change of the second counter which is encoded in labels of $\A_1$) does not drop below $-k$ before a configuration with a zero counter is reached. We will bound both summands in~\eqref{eq:divergence-split} separately.

For the first summand, the discussion at the beginning of this subsection shows that \authNote{upravit bound na $i/2$}
\begin{equation}
\label{eq:split-1}
\calP^{\A_1}_{p(n)}\big(\run(p(n)\rightarrow^* q(0),i/2)\big) \leq \hat{a}\cdot\hat{c}^{\frac{i}{2}} \leq  \hat{a}\cdot\hat{c}^{\frac{i}{2}}.
\end{equation}

For the second summand we use the divergence gap theorem, see above. We have
\begin{equation}
\label{eq:split-2}
\sum_{r\in Q}\sum_{j=0}^{i/2}\calP^{\A_2}_{p(k)}\big(\run(p(k)\rightarrow^* r(0),j)\big)  \quad \leq \quad \sum_{r\in Q} \calP^{\A_2}_{p(k)}\big(\run(p(k)\rightarrow^* r(0))\big)\quad \leq \quad (1-\delta)^{k-|S|} \quad \leq \quad (1-\delta)^{\frac{i}{8}}.
\end{equation}
(The last two inequalities follow from our lower bounds on $i$ and $k$.)
Combining~\eqref{eq:split-2},~\eqref{eq:split-1}, and~\eqref{eq:divergence-split} we can easily compute the numbers $a_4,b_4,c_4$ from the statement of the lemma.
\end{proof}

\newcommand{\Over}{\mathit{Over}}
\newcommand{\Cover}{\mathit{COver}}

Now starting in $p(n,0)$ and assuming that $t_S(2)>0$, we may easily show that the second counter quickly grows with high probability.
\begin{lemma}\label{lem:divergence-bound}
Let $R$ be a region of $\A_2$ satisfying the assumptions of Theorem~\ref{thm:tail-bounds-convergence} whose corresponding BSCC $S$ satisfies $t_S(2)>0$. Then there are numbers $a_5,b_5>0$, $c_5\in (0,1)$ computable in polynomial time such that 
for every $p\in S$ with $p(0)\in R$, every $n\in\Nset$, every $ i\geq 16\cdot(\frac{\hat{h}}{t_S(2)})^2$, and every $k\leq  \frac{i}{8t_S(2)}$ it holds
\[
\calP^{\A}_{p(n,0)}(T_k=i/2 \wedge \Over\leq i/4) \leq a_5\cdot c_5^{\sqrt{i}\cdot b_5}.
\]
\end{lemma}
\begin{proof}


For a run $w\in \run(p(n,0)\rightarrow^* r(*,k))$ we denote by 
\begin{itemize}
\item $T_k(w)$ the smallest $\ell$ such that $\xs{\ell}_2(w)=k$,
\item $\Over(w)$ the largest $\ell\leq T_k(w)$ such that $\xs{\ell}_2(w)=1$,
\item $\Cover(w)$ the configuration $w(\Over(w))$,
\item $\repconf{q(m)}$ the cardinality of the set $\{\ell\in \Nset\mid \ell\leq T_k(w) \wedge w(\ell)=q(m)\}$.
\end{itemize}
We have
\begin{equation}
\label{eq:increase-split}
[p(n,0)\rightarrow^* r(*,k),i/2] \quad \leq \quad 
\calP^{\A}_{p(n,0)}(T_k=i/2 \wedge \Over\leq i/4) + \calP^{\A}_{p(n,0)}(T_k=i/2 \wedge \Over> i/4).
\end{equation}
To prove the lemma it suffices to give tail bounds on both summands in~\eqref{eq:increase-split}.

Let us start with the second summand. Denote by $A$ the set of all runs $w\in \run(p(n,0)\rightarrow^* r(*,k))$ such that
\begin{itemize}
\item $T_k(w)=i/2$
\item
 there are at most $\sqrt{i/4}$ indexes $\ell \leq T_k(w)$ such that the second counter is equal to $0$ in $w(\ell)$.
 \item
 $w$ contains a sub-path $w'(0),\dots,w'(\ell)$ with the following properties:
 \begin{itemize}
 \item $i/2 \geq \ell \geq \sqrt{i/4}-1$,
 \item the second counter is equal to 1 in $w'(0)$ and $w'(\ell)$, 
 \item the second counter is positive in all configurations of $w'$
 \end{itemize}
\end{itemize}
Note that the last item in the definition of $A$ is implied by the previous items.
 Next, denote by $B$ the set of all runs $w\in \run(p(n,0)\rightarrow^* r(*,k))$ such that there are at least $\sqrt{i/4}$ indexes $\ell \leq T_k(w)$ such that the second counter is equal to $0$ in $w(\ell)$. Then
\[
\calP^{\A}_{p(n,0)}(T_k=i/2 \wedge \Over> i/4) \leq \calP^{\A}_{p(n,0)}(A) + \calP^{\A}_{p(n,0)}(B).
\]
Again, it suffices to give tail bounds for both summands on the right-side of the previous equation.

First we focus on the probability of $\A$. Let $w$ be any run initiated in some configuration $q(1)$ of $\A_2$, here $q\in S$, such that the least $\ell$ for which $w(\ell)$ has a zero counter satisfies $\ell \geq (\sqrt{i}/2)-1$. Then 
$\msalt{\ell}(w)-\msalt{0}(w)\leq -\ell\cdot t_{S}(2) + \max_{p,q\in S}|\vec{z}(p)-\vec{z}(q)| \leq -\ell\cdot t_{S}(2) + \hat{h} \leq \frac{-\ell}{2}\cdot t_{S}(2)$, the last inequality following from our assumption that $i\geq 16\cdot(\frac{\hat{h}}{t_S(2)})^2$. Using this fact and Azuma's inequality we get that
\[
\calP^{\A}_{p(n,0)}(A)\leq \frac{\sqrt{i}}{2}\cdot 2 \cdot \sum_{\ell=\lfloor\sqrt{i/4}-1\rfloor}^{i/2} \exp\left(-\frac{\ell}{8((t_S(2))^2 + \hat{h}+ 1)} \right) \leq (a')\cdot(c')^{\sqrt{i}\cdot b'},
\]
for suitable polynomially computable numbers $a',b'>0$, $c'\in(0,1)$.

Now we turn our attention to $B$. Since $R$ is a type II region, from every configuration $q(0)$ of $\A_2$ that is reachable from $p(0)$ there is a finite path of length at most $|Q|^2$ ending in a configuration of the form $t(|S|)$, from which the probability of zeroing the counter is at most $1-\hat{\delta}$, as argued above. Hence, denoting $p_{\min}$ the minimal non-zero transition probability in $\C_{\A_2}$ we get
\[
\calP^{\A}_{p(n,0)}(A) \leq \big((1-p_{\min}^{|Q|^2})\cdot(1-\hat{\delta})\big)^{\lfloor \frac{\sqrt{i}}{|Q|^2 \cdot \sqrt{4}}\rfloor} \leq a'' \cdot (c'')^{\sqrt{i}\cdot b''},
\]
for suitable polynomially computable numbers $a'',b''>0$, $c''\in(0,1)$.

Now we go back to~\eqref{eq:increase-split} and bound the number $\calP^{\A}_{p(n,0)}(T_k=i/2 \wedge \Over\leq i/4)$. We can easily show, using similar arguments as in the proof of Theorem~~\ref{thm:tail-bounds-height}, that
\begin{align*}
&\calP^{\A}_{p(n,0)}(A) \leq \sum_{q\in Q}\sum_{j=n-i/2}^{n+i/2}\sum_{m=1}^{i/2}\sum_{o=0}^{i/4} \calP_{p(n,0)}(T_k = \frac{i}{2}\wedge\xs{i/2}_2= k\wedge \Over = o \wedge \Cover=q(j,1) 
\wedge \repconf{q(j,1)}=m)  \\
&\quad \leq \sum_{m=1}^{i/2}\calP_{p(n,0)} (\repconf{q(j,1)}\geq m) \cdot \Big(\sum_{q\in Q}\sum_{j=n-i/2}^{n+i/2}\sum_{o=0}^{i/4}\cdot \calP_{q(j,1)}(T_k = \frac{i}{2}-o \wedge \xs{\frac{i}{2} - o}_2 = k \wedge \bigwedge_{\ell=1}^{\frac{i}{2}-o}\xs{\ell}_2 >0) \Big)\\
&\quad \leq \frac{i}{2}\sum_{q\in Q}\sum_{j=n-i/2}^{n+i/2}\sum_{o=0}^{i/4}\cdot \calP_{q(j,1)}(\underbrace{T_k = \frac{i}{2}-o \wedge \xs{\frac{i}{2} - o}_2 = k \wedge \bigwedge_{\ell=1}^{\frac{i}{2}-o}\xs{\ell}_2 >0}_{\text{denote by }X}).
\end{align*}
Any run $w\in X$ satisfies $\msalt{\frac{i}{2} - o} - \msalt{0} \leq k - (\frac{i}{2} - o)\cdot t_S(2) + \hat{h} \leq  k - (\frac{i}{4} - o)\cdot t_S(2)$. For $k\leq \frac{i}{8\cdot t_S(2)}$ this number is $\leq (\frac{i}{8}-o)\cdot t_{S}(2)$, and hence we can use the Azuma's inequality to get 
\[
\calP^{\A}_{p(n,0)}(A) \leq  \frac{i}{2}\sum_{q\in Q}\sum_{j=n-i/2}^{n+i/2}\sum_{o=0}^{i/4} 2(c''')^{(i-o)b'''}  
\] 
for suitable polynomially computable numbers $b'''>0$, $c'''\in(0,1)$. Hence,
\[
\calP^{\A}_{p(n,0)}(A) \leq   \frac{i}{2}\sum_{q\in Q}\sum_{j=n-i/2}^{n+i/2} 2(c''')^{(i/4)\cdot b'''}/{(1-(c''')^{b'''})}) \leq a''''\cdot (c'''')^{i\cdot b''''}
\]

For suitable polynomially  computable numbers $a'''',b''''>0$, $c''''\in(0,1)$. The numbers $a_5$, $b_5$, $c_5$ can now be computed in polynomial time using the tail bounds given within this proof.
\end{proof}

Now we finish the proof of Theorem~\ref{thm:tail-bounds-convergence}. We have Let $H=\max\{2\hat{h},4|S|,16(\frac{\hat{h}}{t_S(2)})^2\}$. Then for all $i \geq H\cdot n$ it holds

\begin{align*}
&[p(n,0)\rightarrow^* q(0,*),i]=
\sum_{k\leq i}\, \sum_{r\in Q} \, \sum_{\ell=n-i/2}^{n+i/2} [p(n,0)\rightarrow^* r(\ell,k),i/2][r(\ell,k)\rightarrow^* q(0,*),i/2]
\end{align*}
and thus, by Lemma~\ref{lem:divergence-bound}, 
\begin{align*}
&[p(n,0)\rightarrow^* q(0,*),i]\quad\leq\quad
\sum_{k=0}^{i/8} a_5\cdot c_5^{b_5\sqrt{i/2}}+\sum_{k=i/8}^{i/2}\, \sum_{r\in Q}\, \sum_{\ell=n-i/2}^{n+i/2} [p(n,0)\rightarrow^* r(\ell,k),i/2][r(\ell,k)\rightarrow^* q(0,*),i/2]
\end{align*}
But by Lemma~\ref{lem:crucial-bound2}, for $k\geq i/8$ we have that
\[
[r(\ell,k)\rightarrow^* q(0,*),i/2]\leq a_4\cdot (c_4)^{i\cdot b_4}
\]
This together with the previous equation gives us
\begin{align*}
[p(n,0)\rightarrow^* q(0,*),i]\leq \frac{i}{8}\cdot c_5^{b_5\cdot \sqrt{i/2}}+a_4\cdot (c_4)^{i\cdot b_4}.
\end{align*}
The numbers $a_3$, $b_3$, $z_3$ can now be straightforwardly computed using the above inequality.

%% file: app-attr.tex
\subsection{Proof of Lemma~\ref{lem-II-II-gtrend-leftdown-trends-left-down}}
Let $C$ be the set of all configurations of the form $q(0,m)\in C[R_1,R_2]$ satisfying $m\leq \frac{a_1\cdot c_1^{b_1}}{1-c_1^{b_1}}$ where $a_1,b_1,c_1$ come from Theorem~\ref{thm:tail-bounds-height}.
As explained in Section~\ref{sec-two-counters}, Lemma~\ref{lem-II-II-gtrend-leftdown-trends-left-down} is an immediate consequence of the following
\begin{proposition}[The Attractor]\label{thm:attr}
Consider $C[R_1,R_2]$ where both $R_1$ and $R_2$ are of type II. Assume that $t_S(2)<0$, $t_S(1)\not = 0$, $\octrend_{R_1}<0$ and $\octrend_{R_2}<0$. \authNote{I assume that $t_S(1)>0$ and $t_S(2)>0$ has already been solved elsewhere.}
Then $C[R_1,R_2]$ is eagerly attracted to $C$. \authNote{the tailbound is only subexponential!!}
\end{proposition}

Given a configuration $p\vec{v}$ and a set $A$ of configurations we denote by $[p\vec{v}\rightarrow^* A,\geq i]$ the probability that a run that starting in $p\vec{v}$ visits $A$ in at least $i$ steps and does not visit $A$ between the first and the last step.

Given a run $w$, we denote by $T_1(w)$, $T_2(w)$, and $T_{12}(w)$ the least $k$ such that $w(k)\in C_S[c_1=0]$, $w(k)\in C_S[c_2=0]$, and $w(k)\in Z_S=C_S[c_1=0]\cup C_S[c_2=0]$, respectively.

The following lemma reformulates results of Theorem~\ref{thm:tail-bounds-height}, Theorem~\ref{thm:tail-bounds-divergence}, and Theorem~\ref{thm:tail-bounds-convergence} in a bit weaker but more transparent way.
\begin{lemma}\label{lem:attr-fund}
Let us fix $p\vec{v}\in C[R_1,R_2]$.
There are effectively computable numbers $a>0$ and $0<b<1$ (depending on $p\vec{v}$) such that the following holds:
\begin{enumerate}
\item For all $\ell\in \Nset$ we have 
\[
[p\vec{v}\rightarrow^* Z_S,\geq\ell]\quad \leq\quad a\cdot b^{\ell}
\]
\item For all $\ell\in \Nset$ we have
\[
\calP_{p\vec{v}} (\xs{T_{12}}_1+\xs{T_{12}}_2\geq \ell)\quad \leq\quad a\cdot b^{\ell}
\]
\item For all $q\in Q$, $m\in \Nset$ and $\ell\geq m$ we have
\[
[q(m,0)\rightarrow^* C_S[c_1=0],\geq \ell]\quad \leq\quad a\cdot b^{\sqrt{\ell-m}}
\]
and
\[
[q(0,m)\rightarrow^* C_S[c_2=0],\geq \ell]\quad \leq\quad a\cdot b^{\sqrt{\ell-m}}
\]
\item For all $q\in Q$, $m\in \Nset$ and $\ell\in \Nset$ we have 
\[
\calP_{q(m,0)} (\xs{T_1}_2\geq \ell)\quad \leq\quad a\cdot b^{\ell}
\]
\item  For all $q\in Q$, $m\in \Nset$ we have the following:
	\begin{enumerate}
	\item If $t_S(1)<0$, then for all $\ell\in \Nset$ we have 
	\[
	\calP_{q(0,m)} (\xs{T_2}_1\geq \ell)\quad \leq\quad a\cdot b^{\ell}
	\]
	\item If $t_S(1)>0$, then for all $\ell\geq m$ we have 
	\[
	\calP_{q(0,m)} (\xs{T_2}_1\geq \ell)\quad \leq\quad a\cdot b^{\sqrt{\ell-m}}
	\]
	\end{enumerate}
\end{enumerate}
\end{lemma}

Let $w$ be a run starting in $p\vec{v}$.
Denote by $\Theta_0(w)$ the least $\ell$ such that $w(\ell)\in C_S[c_2=0]$.
Given $k\geq 1$, denote by $\Theta_k(w)$ the least $\ell\geq \Theta_{k-1}(w)$ such that the following holds
\begin{itemize}
\item If $k$ is odd, then $w(\ell)\in C_S[c_1=0]$.
\item If $k$ is even, then $w(\ell)\in C_S[c_2=0]$.
\end{itemize}

\begin{lemma}\label{lem:attr-first-phase}
There are effectively computable numbers $\hat{a}>0$ and $0<\hat{b}<1$ such that for all $k\geq 0$ and all $\ell\in \Nset$ we have
\[
\calP_{p\vec{v}}(\Theta_k-\Theta_{k-1}\geq \ell)\quad \leq \quad \hat{a}\cdot (\hat{b})^{\sqrt{\ell}}
\]
\end{lemma}
\begin{proof}
We distinguish two cases $k=0$ and $k>0$.

\noindent {\bf Case $\Theta_0$:}
Note that 
\begin{align*}
\calP_{p\vec{v}}(\Theta_0\geq \ell) &\quad  \leq\quad  [p\vec{v}\rightarrow^* Z_S,\geq \ell/2]+\calP_{p\vec{v}} (\xs{T_{12}}_1+\xs{T_{12}}_2\geq \ell/4)\\
& \qquad + \sum_{q\in Q}\sum_{m=1}^{\ell/4} \calP_{p\vec{v}} (\xs{T_1}_2=m\wedge \ps{T_1}=q)\cdot [q(m,0)\rightarrow^* C_S[c_2=0],\geq \ell/2]\\
&\quad  \leq\quad  a\cdot b^{\ell/2} +a\cdot b^{\ell/4} + \sum_{q\in Q}\sum_{m=1}^{\ell/4} \calP_{p\vec{v}} (\xs{T_1}_2=m\wedge \ps{T_1}=q)\cdot a\cdot b^{\sqrt{\ell/2-m}}\\
& \quad \leq\quad 3\cdot a\cdot b^{\sqrt{\ell/4}}\\
\end{align*}
which can easily be rewritten to the desired form.

\noindent {\bf Case $\Theta_k$:} 
If either $t_S(1)<0$, or $k$ is odd, then using Lemma~\ref{lem:attr-first-phase} 2., 4., 5. a), and induction one may easily prove that
\[
\calP_{p\vec{v}}(\xs{\Theta_k}_1+\xs{\Theta_k}_2\geq \ell)\quad\leq\quad a\cdot b^{\ell}
\]
(Intuitively, whenever we start in $q(m,0)$, then we reach $C_S[c_1=0]$ with probability one and by Lemma~\ref{lem:attr-first-phase} 4., the probability that the height of the second counter at the time is at least $\ell$ is bounded by $a\cdot b^{\ell}$ (independently of $m$). The same holds for configurations $q(0,m)$ since $t_S(1)<0$.)

Now assume that $t_S(1)>0$ and that $k\geq 2$ is even. Note that
\begin{align*}
\calP_{p\vec{v}}(\xs{\Theta_k}_2\geq \ell) &\quad  \leq\quad  \calP_{p\vec{v}} (\xs{\Theta_{t-1}}_1\geq \ell/2)+ \sum_{q\in Q}\sum_{m=1}^{\ell/2} \calP_{p\vec{v}}(\xs{\Theta_{t-1}}_2=m\wedge \ps{\Theta_{t-1}}=q)\cdot\calP_{q(0,m)}(\xs{T_2}_1\geq \ell)\\
& \quad \leq \quad a\cdot b^{\ell/2}+ \sum_{q\in Q}\sum_{m=1}^{\ell/2} \calP_{p\vec{v}}(\xs{\Theta_{t-1}}_2=m\wedge \ps{\Theta_{t-1}}=q)\cdot a\cdot b^{\sqrt{\ell-m}}\\
&\quad  \leq\quad  2\cdot a\cdot b^{\sqrt{\ell/2}}\\
& \quad \leq \quad a'\cdot (b')^{\sqrt{\ell}}
\end{align*}
for suitable $a'>0$ and $0<b'<1$ that are effectively computable and satisfy $a\cdot b^{\sqrt{\ell}}\leq a'\cdot (b')^{\sqrt{\ell}}$.

So for {\em arbitrary} $t_S(1)\not = 0$ and  $k\geq 1$ we have that 
\[
\calP_{p\vec{v}}(\xs{\Theta_k}_1+\xs{\Theta_k}_2\geq \ell)\quad\leq\quad a'\cdot (b')^{\sqrt{\ell}}
\]
Now using the same argument as for $\Theta_0$, we obtain for $k$ odd,
\begin{align*}
\calP_{p\vec{v}}(\Theta_k-\Theta_{k-1}\geq \ell) &\quad  \leq\quad  \calP_{p\vec{v}} (\xs{\Theta_{t-1}}_2\geq \ell/2) + \sum_{q\in Q}\sum_{m=1}^{\ell/2} \calP_{p\vec{v}} (\xs{\Theta_{t-1}}_2=m\wedge \ps{\Theta_{t-1}}=q)\cdot [q(m,0)\rightarrow^* C_S[c_2=0],\geq \ell]\\
&\quad  \leq\quad  a'\cdot (b')^{\ell/2} + \sum_{q\in Q}\sum_{m=1}^{\ell/2} \calP_{p\vec{v}} (\xs{\Theta_{t-1}}_2=m\wedge \ps{\Theta_{t-1}}=q)\cdot a'\cdot (b')^{\sqrt{\ell-m}}\\
&\quad  \leq\quad  2\cdot a'\cdot (b')^{\sqrt{\ell/2}}\\
\end{align*}
which can easily be rewritten into the desired form.
For $k$ even the argument is similar with counters switched.
\end{proof}

Now let us finish the proof of Proposition~\ref{thm:attr}.
Let $w$ be a run starting in $p\vec{v}$. Denote by $\mathit{Rounds}(w)$ the least number $k\geq 0$ such that $w(T_k(w))\in C$. Denote by $\neg\mathit{Term}_k$ the set of all runs $w$ starting in $p\vec{v}$ satisfying $w(\Theta_k(w))\not\in C$. We write $\neg\mathit{Term}_{<k}$ to denote the~set of all runs $w$ starting in $p\vec{v}$ satisfying $w(\Theta_j(w))\not\in C$ for all $0\leq j<k$. Finally, we denote by $\neg\mathit{Term}$ the set of {\em all} runs that {\em do not} visit $C$ before visiting a configuration of $C_S[c_1=0]\smallsetminus C$.

It follows from Theorem~\ref{thm:tail-bounds-height} that there is an effectively computable constant $0\leq \bar{c}<1$ (independent of $\vec{v}(1)$) such that for all $q\in Q$ and all $m\in \Nset$ we have 
\[
\calP_{q(m,0)} (\neg\mathit{Term})\quad \leq\quad \bar{c}
\]
For all $k\geq 1$ odd holds
\begin{align*}
\calP_{p\vec{v}}(\neg\mathit{Term}_k & \wedge \neg\mathit{Term}_{<k})  = \sum_{q\in Q}\sum_{m\in \Nset} \calP_{p\vec{v}}(\xs{\Theta_{k-1}}_1=m\wedge \ps{\Theta_{k-1}}=q\wedge \neg\mathit{Term}_k\wedge \neg\mathit{Term}_{<k})\\
& =  \sum_{q\in Q}\sum_{m\in \Nset} \calP_{p\vec{v}}(\xs{\Theta_{k-1}}_1=m\wedge \ps{\Theta_{k-1}}=q\wedge \neg\mathit{Term}_{<k})\cdot  \calP_{p\vec{v}} (\neg\mathit{Term}_k\mid \xs{\Theta_{k-1}}_1=m\wedge \ps{\Theta_{k-1}}=q\wedge \neg\mathit{Term}_{<k})\\
& =  \sum_{q\in Q}\sum_{m\in \Nset} \calP_{p\vec{v}}(\xs{\Theta_{k-1}}_1=m\wedge \ps{\Theta_{k-1}}=q\wedge \neg\mathit{Term}_{<k})\cdot  \calP_{q(m,0)} (\neg\mathit{Term})\\
& \leq  \sum_{q\in Q}\sum_{m\in \Nset} \calP_{p\vec{v}}(\xs{\Theta_{k-1}}_1=m\wedge \ps{\Theta_{k-1}}=q\wedge \neg\mathit{Term}_{<k})\cdot  \bar{c}\\
& =\calP_{p\vec{v}}(\neg\mathit{Term}_{<k})\cdot \bar{c}
\end{align*}
For all $k\geq 1$ even holds
\[
\calP_{p\vec{v}}(\neg\mathit{Term}_k \wedge \neg\mathit{Term}_{<k})\quad \leq \quad \calP_{p\vec{v}}(\neg\mathit{Term}_{<k})
\]
Now since $\neg\mathit{Term}_{<k}=\neg\mathit{Term}_{k-1}\cap \neg\mathit{Term}_{<k-1}$ we obtain, by induction, that
$\calP_{p\vec{v}}(\mathit{Rounds}\geq \sqrt k)\leq (\bar{c})^{(\lfloor\sqrt k\rfloor)/2-1}$.
Thus for all $i\geq \Nset$ we have
\begin{align*}
[p\vec{v}\rightarrow^* C,\geq i] & \leq \calP(\mathit{Tours}\geq \lfloor\sqrt i\rfloor)+\sum_{k=1}^{\lfloor\sqrt i\rfloor} \calP(\Theta_k-\Theta_{k-1}\geq \lfloor\sqrt i\rfloor)\\
& \leq (\bar{c})^{(\lfloor\sqrt i\rfloor)/2-1}+\sqrt i \cdot \hat{a} (\hat{b})^{\lfloor\sqrt[4] i\rfloor}
\end{align*}
which proves that $p\vec{v}$ is eagerly attracted to $C$.

%% file: stefan.tex
\newcommand{\diffnorm}[1]{\left|#1\right|_\mathit{diff}}%

\subsection{Geometric Sums of Stochastic Matrices} \label{sub-geometric-sums}

For the proof of Lemma~\ref{lem:bounded-differences} in Appendix~\ref{sub-bounded-difference} we will need a general lemma on geometric sums of stochastic matrices,
see Lemma~\ref{lem-mart-coupling} below.
For the proof of Lemma~\ref{lem-mart-coupling} we use a coupling argument on finite-state Markov chains.
As a preparation we first prove the following lemma on finite-state Markov chains.

\begin{lemma} \label{lem-mart-hitting-time-finite-chain}
 Consider a finite-state Markov chain on a set~$Q$ of states with $|Q| = n$. 
 Let $\ymin$ denote the smallest nonzero transition probability in the chain.
 Let $p \in Q$ be any state and $S \subseteq Q$ any subset of~$Q$.
 Define the random variable $T$ on runs starting in~$p$ by
  \[
   T := \begin{cases} k & \text{if the run hits a state in~$S$ for the first time after exactly $k$ steps} \\
                      \mathit{undefined} & \text{if the run never hits a state in~$S$ .}
        \end{cases}
  \]
 We have $\calP(T \ge k) \le 2 c^k$ for all $k \ge n$, where $c := \exp(-\ymin^n/n)$.
 Moreover, if $\calP(T < \infty) = 1$, then we have $\calE T \le 5 n / \ymin^n$,
  we we write~$\calE$ for the expectation with respect to~$\calP$.
\end{lemma}
\begin{proof}
 If $\ymin=1$ then all states that are visited are visited after at most $n-1$ steps and hence $\calP(T \ge n) = 0$.
 Assume $\ymin < 1$ in the following.
 Since for each state the sum of the probabilities of the outgoing edges is~$1$, we must have $\ymin \le 1/2$.
 Call \emph{crash} the event of, within the first $n-1$ steps, either hitting~$S$ or some state~$r \in Q$ from which $S$ is not reachable.
 The probability of a crash is at least $\ymin^{n-1} \ge \ymin^n$, regardless of the starting state.
 Let $k \ge n$.
 For the event where $T \ge k$, a crash has to be avoided at least $\lfloor \frac{k-1}{n-1} \rfloor$ times; i.e.,
 \[
  \calP(T \ge k) \le (1-\ymin^n)^{\lfloor \frac{k-1}{n-1} \rfloor} \,.
 \]
 As $\lfloor \frac{k-1}{n-1} \rfloor \ge \frac{k-1}{n-1} - 1 \ge \frac{k}{n} - 1$, we have
 \begin{align*}
  \calP(T \ge k)
  & \le \frac{1}{1-\ymin^n} \cdot \left((1-\ymin^n)^{1/n}\right)^k
    \le 2 \cdot \left((1-\ymin^n)^{1/n}\right)^k \\
  &  =  2 \cdot \exp\left(\frac1n \log(1-\ymin^n)\right)^k
    \le 2 \cdot \exp\left(\frac1n \cdot (-\ymin^n)\right)^k
    = 2 \cdot c^k \,.
 \end{align*}
Moreover, if $\calP(T < \infty) = 1$, we have:
\begin{align*}
 \calE T
 & = \sum_{k=1}^\infty \calP(T \ge k) \\
 & \le n + \sum_{k=0}^\infty 2 c^k \\
 & = n + \frac{2}{1-\exp(-\ymin^n/n)} \\
 & \le n + \frac{4 n}{\ymin^n} && \text{as $\exp(-\ymin^n/n) \le 1 - \frac{\ymin^n}{2 n}$} \\
 & \le 5 n / \ymin^n
\end{align*}
\end{proof}

Now we are ready to prove the following lemma.

\begin{lemma} \label{lem-mart-coupling}
Let $G \in [0,1]^{Q \times Q}$ be a stochastic matrix with only one BSCC.
Let $\ymin$ denote the smallest nonzero entry of~$G$.
Let $\vec{r} \in \Rset^Q$ be a vector.
Define $\vec{f}(n) := \sum_{i=0}^{n-1} G^i \vec{r}$ for all $n \in \Nset$.
For any vector $\vec{v} \in \Rset^Q$ let us define $\diffnorm{\vec{v}} := \max_{p_1,p_2 \in Q} |\vec{v}[p_1] - \vec{v}[p_2]|$.
Then we have $\diffnorm{\vec{f}(n)} \le C \diffnorm{\vec{r}}$ for all $n \in \Nset$,
 where $C := 10 |Q| / \ymin^{|Q|}$.
\end{lemma}
\begin{proof}
Let $\kappa \in \Rset$, and define $\vec{r}_+ := \vec{r} + \kappa \vec{1}$.
Note that $\diffnorm{\vec{r}_+} = \diffnorm{\vec{r}}$.
Since $G$ is stochastic, we have for all $n \in \Nset$:
\[
 \diffnorm{\sum_{i=0}^{n-1} G^i \vec{r}_+}
 =  \diffnorm{n \kappa \vec{1} + \sum_{i=0}^{n-1} G^i \vec{r}} = \diffnorm{n \kappa \vec{1} + \vec{f}(n)} = \diffnorm{\vec{f}(n)}
\]
So in the following we can assume without loss of generality that $\vec{r} \ge \vec{0}$ and $\vec{r}[p] = 0$ for some $p \in Q$,
 so that we have:
\begin{equation}
 \diffnorm{\vec{r}} = \norm{\vec{r}} \label{eq-lem-mart-coupling-norms}
\end{equation}

Consider the finite-state Markov chain on~$Q$ induced by~$G$.
Let $p_1, p_2 \in Q$ be arbitrary states.
\newcommand{\Xs}[1]{X^{(#1)}}%
We define random runs $\Xs{0}_1, \Xs{1}_1, \ldots$ and $\Xs{0}_2, \Xs{1}_2, \ldots$ in the Markov chain,
 with $\Xs{i}_1, \Xs{i}_2 \in Q$ for all $i \in \Nset$, and $\Xs{0}_1 = p_1$ and $\Xs{0}_2 = p_2$, and
\begin{equation}
\begin{aligned}
 \calP(\Xs{i+1}_1 = q \mid \Xs{i}_1 = p) & \quad = \quad G[p,q] && \text{for all $i \in \Nset$ \quad and} \\
 \calP(\Xs{i+1}_2 = q \mid \Xs{i}_2 = p) & \quad = \quad G[p,q] && \text{for all $i \in \Nset$.}
\end{aligned} \label{eq-mart-lem-mart-coupling-markov}
\end{equation}
We write $\X_1, \X_2$ for the sequences $(\Xs{i}_1)_i$ and $(\Xs{i}_2)_i$ in the following.
For $q \in Q$ we regard $\vec{r}[q]$ as a ``reward'' incurred when the chain is in state~$q$.
For each $n \in \Nset$ we define a random variable $R_1(n)$, the ``accumulated reward before time~$n$'':
 \[
  R_1(n) := \sum_{i=0}^{n-1} \vec{r}[\Xs{i}_1]
 \]
We define $R_2(n)$ similarly, replacing $\Xs{i}_1$ by~$\Xs{i}_2$.
Writing~$\calE$ for expectation, we have for all $n \in \Nset$:
\begin{equation}
 \calE R_1(n) = \vec{f}(n)[p_1] \quad \text{and} \quad \calE R_2(n) = \vec{f}(n)[p_2] \label{eq-mart-lem-mart-coupling-exp}
\end{equation}
We now refine the definition of $\X_1$ and~$\X_2$ by coupling them as follows.
Let $s \in Q$ be a state from the only BSCC of~$G$.
Let $T_1, T_2\in \Nset$ so that:
\begin{align*}
 T_1 & := \min \{ i \in \Nset \mid \Xs{i}_1 = s \} \\
 T_2 & := \min \{ i \in \Nset \mid \Xs{i}_2 = s \}
\end{align*}
Note that $T_1, T_2$ exist almost surely.
We now require from $\X_1, \X_2$ that for each $i \in \Nset$ we have:
\begin{itemize}
 \item if there is $j \in \Nset$ with $i-j = T_2 \ge T_1$, then $\Xs{i}_2 = \Xs{T_1 + j}_1$;
 \item if there is $j \in \Nset$ with $i-j = T_1 \ge T_2$, then $\Xs{i}_1 = \Xs{T_2 + j}_2$.
\end{itemize}
In words: if $\X_1$ reaches~$s$ first, then as soon $\X_2$ also reaches~$s$, it mimics the behavior of~$\X_1$ after it had reached~$s$;
 symmetrically, if $\X_2$ reaches~$s$ first, then $\X_1$ mimics~$\X_2$ in a similar way;
 if they reach~$s$ at the same time, their behavior is henceforth identical.
Note that although $\X_1, \X_2$ are not independent, Equations \eqref{eq-mart-lem-mart-coupling-markov}~and~\eqref{eq-mart-lem-mart-coupling-exp}
 remain valid, as they did not require independence.
By the coupling we have:
\[
 R_1(T_1 + n) - R_1(T_1) = R_2(T_2 + n) - R_2(T_2) \quad \text{for all $n \in \Nset$}
\]
Let $T_1 \le T_2$.
Then it follows for all $n \in \Nset$:
\begin{align*}
 R_1(T_1 + n) - R_2(T_1 + n)
 & = R_2(T_2 + n) - R_2(T_2) + R_1(T_1) - R_2(T_1 + n) \\
 & \le (T_2 - T_1) \norm{\vec{r}} + T_1 \norm{\vec{r}} \le (T_1 + T_2) \norm{\vec{r}}
\end{align*}
Let now $T_2 \le T_1$.
Then we similarly have for all $n \in \Nset$:
\begin{align*}
 R_1(T_1 + n) - R_2(T_1 + n)
 & = R_2(T_2 + n) - R_2(T_2) + R_1(T_1) - R_2(T_1+ n) \\
 & \le R_1(T_1) \le T_1 \norm{\vec{r}}
\end{align*}
So in any case ($T_1 \le T_2$ or $T_2 \le T_1$) we have:
\begin{equation}
 R_1(n) - R_2(n) \le (T_1 + T_2) \norm{\vec{r}} \quad \text{for all $n \in \Nset$} \label{eq-mart-lem-coupling-R-diff-bound}
\end{equation}
We have:
\begin{align*}
 \vec{f}(n)[p_1] - \vec{f}(n)[p_2]
 & = \calE R_1(n) - \calE R_2(n) && \text{by~\eqref{eq-mart-lem-mart-coupling-exp}} \\
 & = \calE \left( R_1(n) - R_2(n) \right) && \text{by linearity of expectation} \\
 & \le \norm{\vec{r}} (\calE T_1 + \calE T_2) && \text{by~\eqref{eq-mart-lem-coupling-R-diff-bound}} \\
 & \le 10 \norm{\vec{r}} |Q| / \ymin^{|Q|} && \text{by Lemma~\ref{lem-mart-hitting-time-finite-chain}} \\
 & = 10 \diffnorm{\vec{r}} |Q| / \ymin^{|Q|} && \text{by~\eqref{eq-lem-mart-coupling-norms}}
\end{align*}
The statement follows, as $p_1, p_2$ were chosen arbitrarily.
\end{proof}

\subsection{Bounded Differences} \label{sub-bounded-difference}

We are going to prove Lemma~\ref{lem:bounded-differences}, stating that the martingale defined in~\eqref{eq-mart-m} has bounded differences.
Recall that we have fixed a region~$R$ of~$\A_2$ whose type is~II and which is determined by some $p \in S$, where $S$ satisfies $\twot_S(2)<0$.
Let us recall from~\cite{BKKNK:pMC-zero-reachability} the definition of the function $\vec{g}: \Nset \to \Rset^{S}$ referred to in~\eqref{eq-mart-m}.
A certain vector~$\vec{g}(0) \in \Rset^{S}$ was defined in~\cite{BKKNK:pMC-zero-reachability},
and since the matrix~$A$ in~\cite{BKKNK:pMC-zero-reachability} is stochastic, we can assume $\vec{g}(0) \ge \vec{0}$.
Using $\vec{g}(0)$ as base case, we define the function
$\vec{g}: \Nset \to \Rset^{S}$ inductively as follows:
 \begin{equation}
   \vec{g}(n + 1) = \vec{r}_\downarrow + G \vec{g}(n) \qquad \text{for all $n \in \Nset$,} \label{eq-mart-g}
 \end{equation}
where $\vec{r}_\downarrow \in \Rset^{S}$ was defined in~\cite{BKKNK:pMC-zero-reachability}, and $G \in \Rset^{S \times S}$ denotes the matrix such that $G[q,r]$ is the probability that starting from $q(1)$
 the configuration $r(0)$ is visited before visiting any configuration $r'(0)$ for any $r' \ne r$.
Since $\twot_S(2)<0$, the matrix~$G$ is stochastic, i.e., $G \vec{1} = \vec{1}$.
We prove:
\begin{lemma} \label{lem-G-one-BSCC}
Matrix~$G$ has only one BSCC.
\end{lemma}
\begin{proof}
For $p, q \in S$ and $k \in \mathbb{Z}$ we write $p \xrightarrow{k} q$ if there is a path in~$\A_2$
 from $p(n)$ to~$q(n+k)$ for some $n \ge 0$.
Observe that there is $k < 0$ with $p \xrightarrow{k} q$ if and only if
 there is a nonempty path in the graph of~$G$ from $p$ to~$q$.

Let $p \in S$ and $q \in \B$, where $\B \subseteq S$ is a BSCC of~$G$.
We need to show that $q$ is reachable from~$p$ in the graph of~$G$.
Let $k \in \mathbb{Z}$ with $p \xrightarrow{k} q$.
Since $q$ is in the BSCC~$\B$ of~$G$, we have $q \xrightarrow{\ell} q$ for some $\ell < 0$.
By combining the two paths, we get $p \xrightarrow{k + \ell} q$ and, by pumping the second path,
 $p \xrightarrow{k + a \ell} q$ for all $a \ge 0$.
So by choosing $a$ large enough we get $p \xrightarrow{m} q$ for some $m < 0$.
Hence there is a nonempty path in the graph of~$G$ from $p$ to~$q$.
\end{proof}

\begin{lemma} \label{prop-bounded-differences}
Let $\gmax > 0$ so that $\vec{0} \le \vec{g}(0) \le \gmax \vec{1}$.
Let $\rmax \ge 1$ such that $\norm{\vec{r}_\downarrow} \le \rmax$.
Let $\xmin$ be the smallest nonzero probability in the description of~$\A_2$.
Then
\[
|\ms{\ell+1}-\ms{\ell}| \ \le \
2 + 2 \gmax + 30 |S| \rmax / \xmin^{|S|^4}
\text{ \qquad for every $\ell\in \Nset$.}
\]
\end{lemma}
\begin{proof}
By~\eqref{eq-mart-g}, using straightforward induction, we obtain for all $n \ge 0$:
\begin{align}
 \vec{g}(n)  & = G^{n} \underbrace{\vec{g}(0)}_{\le \gmax \vec{1}} + \sum_{i=0}^{n - 1} G^i \underbrace{\vec{r}_\downarrow}_{\le \rmax \vec{1}}  \label{eq-mart-nob-bound-comp}
\end{align}
Since $G$ is stochastic, it follows for all $n \ge 0$:
\begin{equation}
\norm{\vec{g}(n+1) - \vec{g}(n)}
\ \le \ \norm{ G^{n+1} \vec{g}(0) - G^n \vec{g}(0) } + \norm{ G^n \vec{r}_\downarrow }
\ \le \ \gmax + \rmax
\label{eq-lem-mart-nob-g-bound-3}
\end{equation}
For a vector $\vec{v} \in \Rset^S$ let us define
 \[
  \diffnorm{\vec{v}} := \max_{p,q \in S} |\vec{v}[p] - \vec{v}[q]|
 \]
Observe that for any $\vec{u}, \vec{v} \in \Rset^S$ we have $\diffnorm{\vec{u} + \vec{v}} \le \diffnorm{\vec{u}} + \diffnorm{\vec{v}}$.
Let $C$ be the constant from Lemma~\ref{lem-mart-coupling} for the matrix~$G$.
We have for all $n \ge 0$:
\begin{align}
 \diffnorm{\vec{g}(n)}
 & \le \diffnorm{G^n \vec{g}(0)} + \diffnorm{\sum_{i=0}^{n-1} G^i \vec{r}_\downarrow} && \text{by~\eqref{eq-mart-nob-bound-comp}} \notag \\
 & \le \gmax + C \diffnorm{\vec{r}_\downarrow} && \text{by Lemmas \ref{lem-mart-coupling}~and~\ref{lem-G-one-BSCC}} \notag \\
 & \le \gmax + 2 C \rmax && \text{as $\diffnorm{\vec{r}_\downarrow} \le 2 \rmax$} \label{eq-lem-mart-nob-g-bound-4}
\end{align}
Recall from~\eqref{eq-mart-m} that $\ms{\ell} = \xs\ell_1 - \octrend_R \ell + \vec{g}\big(\xs\ell_2\big)[\ps\ell]$ for all $\ell \in \Nset$.
Hence we have:
\begin{align*}
|\ms{\ell+1}-\ms{\ell}|
& \le |\xs{\ell+1}_1 - \xs{\ell}_1| + |\octrend_R| \\
& \quad     + \left|\vec{g}\big(\xs{\ell+1}_2\big)[\ps{\ell+1}]
        - \vec{g}\big(\xs\ell_2\big)[\ps\ell]\right| \\
& \le 1 + 1 + \left|\vec{g}\big(\xs{\ell+1}_2\big)[\ps{\ell+1}]
                    - \vec{g}\big(\xs\ell_2\big)[\ps{\ell+1}]\right| \\
& \quad     + \left|\vec{g}\big(\xs\ell_2\big)[\ps{\ell+1}]
                    - \vec{g}\big(\xs\ell_2\big)[\ps{\ell}] \right| \\
& \le 2 + \left|\vec{g}\big(\xs{\ell+1}_2\big)
        - \vec{g}\big(\xs\ell_2\big)\right|
        + \diffnorm{\vec{g}\big(\xs\ell_2\big)} \\
& \le 2 + 2 \gmax + \underbrace{(2 C + 1)}_{\le 3 C} \rmax && \text{by \eqref{eq-lem-mart-nob-g-bound-3} and \eqref{eq-lem-mart-nob-g-bound-4}}
\intertext{We have $C := 10 |S| / \ymin^{|S|}$, where $\ymin$ is the smallest nonzero entry of~$G$.
By~\cite[Corollary 6]{EWY:one-counter} we have $\ymin \ge \xmin^{|S|^3}$, so we have $C \le 10 |S| / \xmin^{|S|^4}$.
Hence:}
|\ms{\ell+1}-\ms{\ell}|
& \le 2 + 2 \gmax + 30 |S| \rmax / \xmin^{|S|^4}
\end{align*}
\end{proof}

Now we can prove Lemma~\ref{lem:bounded-differences}.

\begin{reftheorem}{Lemma}{\ref{lem:bounded-differences}}
There is a bound $B\geq 1$ computable in polynomial space such that $|\ms{\ell+1}-\ms{\ell}|\leq B$ for every $\ell\in \Nset$.
\end{reftheorem}

\begin{proof} 
This follows from Lemma~\ref{prop-bounded-differences},
as the vectors $\vec{g}(0)$ and $\vec{r}_\downarrow$, as defined in~\cite{BKKNK:pMC-zero-reachability}, can easily be expressed in the existential theory of the reals.
\end{proof}